\documentclass[
12pt, 
oneside, 
english, 
onehalfspacing, 
nolistspacing, 
liststotoc, 
parskip, 
headsepline, 
]{DoctoralThesis} 

\usepackage[utf8]{inputenc} 
\usepackage[T1]{fontenc} 
\usepackage{booktabs}
\usepackage{url}
\usepackage{graphicx}
\usepackage{titlesec}
\usepackage{enumitem}
\usepackage[most]{tcolorbox}
\usepackage{pifont}

\usepackage[backend=bibtex,style=numeric,natbib=true,maxbibnames=999]{biblatex} 
\usepackage{titlesec}
\usepackage{xkeyval}

\titleformat{\chapter}[display]
  {\bfseries\LARGE}
  {\filright\Huge\chaptertitlename~\thechapter}
  {3ex}
  {\titlerule[1.5pt]\vspace{1.5ex}\filright}
  [\vspace{1ex}{\titlerule[1.5pt]}]

\usepackage{array}

\usepackage[utf8]{inputenc} 
\usepackage[T1]{fontenc}    

\usepackage{url}            
\usepackage{amsfonts}       
\usepackage{nicefrac}       
\usepackage{microtype}      
\usepackage{lipsum}		
\usepackage{graphicx}

\usepackage{doi}

\usepackage{graphicx}
\usepackage[utf8]{inputenc}
\usepackage{cancel}
\usepackage{amsmath}
\usepackage{url}
\usepackage{multirow}
\usepackage{array,multirow}
\usepackage{algorithm}
\usepackage{algorithmic}
\usepackage{graphicx}
\usepackage{verbatim}
\graphicspath{ {./figure/} }
\usepackage{babel}
\usepackage{xfrac}
\usepackage{dblfloatfix}
\usepackage{etoolbox}
\makeatletter
\patchcmd{\@makecaption}
  {\scshape}
  {}
  {}
  {}
\makeatother
\usepackage{xcolor}
\usepackage{tabularx}
\usepackage{amssymb}
\usepackage{amsthm}
\usepackage{cleveref}
\crefname{chapter}{chapter}{chapters}

\usepackage{array, boldline, makecell, booktabs}
\newtheorem{theorem}{Theorem}
\newtheorem{definition}{Definition}
\newtheorem{lemma}{Lemma}
\newtheorem{Corollary}{Corollary}

\usepackage{chngcntr}
\counterwithin{theorem}{chapter}
\counterwithin{definition}{chapter}
\counterwithin{lemma}{chapter}
\counterwithin{Corollary}{chapter}

\usepackage{multirow}
\usepackage{enumitem}
\usepackage{diagbox}
\usepackage{afterpage}

\usepackage{color,soul}
\usepackage{array,multirow}
\newcommand{\specialcell}[2][c]{%
    \begin{tabular}[#1]{@{}c@{}}#2\end{tabular}}
\usepackage{microtype}
\usepackage{booktabs}
\usepackage[utf8]{inputenc}
\setul{0.5ex}{0.3ex}
\definecolor{Green}{rgb}{0,1,0}
\usepackage[bottom]{footmisc}

\newcommand{\xdownarrow}[1]{%
  {\left\downarrow\vbox to #1{}\right.\kern-\nulldelimiterspace}
}
\usepackage{amsmath}
\usepackage{mathtools}
\newcommand{\mathdash}{\relbar\mkern-19mu\relbar}

\usepackage{nccmath}
\usepackage{mathtools}
\usepackage{soul}
\usepackage[ruled,vlined,linesnumbered,algo2e,resetcount,algochapter]{algorithm2e}
\usepackage{algorithm}

\makeatletter 
 
\@addtoreset{algorithm}{chapter} 
\makeatother

\setbool{listtoc}{false} 
\setbool{nolistspace}{false} 

\newcommand\blfootnote[1]{%
  \begingroup
  \renewcommand\thefootnote{}\footnote{#1}%
  \addtocounter{footnote}{-1}%
  \endgroup
}

\usepackage[autostyle=true]{csquotes} 

\thesistitle{Enhancing Network Resilience through Machine Learning-powered Graph Combinatorial Optimization: Applications in Cyber Defense and Information Diffusion}
\supervisor{Principle Supervisor: Dr. Mingyu Guo\\Co-supervisors: Prof. Hong Shen and A/Prof. Hui Tian} 
\examiner{} 
\degree{Doctor of Philosophy} 
\author{\textbf{Diksha Goel}}
\addresses{} 
\subject{Computer Science} 
\keywords{} 
\university{\href{http://www.adelaide.edu.au}{University of Adelaide}} 
\department{\href{https://cs.adelaide.edu.au}{School of Computer Science}} 
\group{\href{https://cs.adelaide.edu.au/research/cseducation/}{Computer Science Education Research}} 
\faculty{\href{http://faculty.university.com}{Faculty Name}} 

\AtBeginDocument{
\hypersetup{pdftitle=\ttitle} 
\hypersetup{pdfauthor=\authorname} 
\hypersetup{pdfkeywords=\keywordnames} 
}

\usepackage{hyperref}

\begin{document}

\frontmatter 

\pagestyle{plain} 

\begin{titlepage}
\begin{center}
\includegraphics{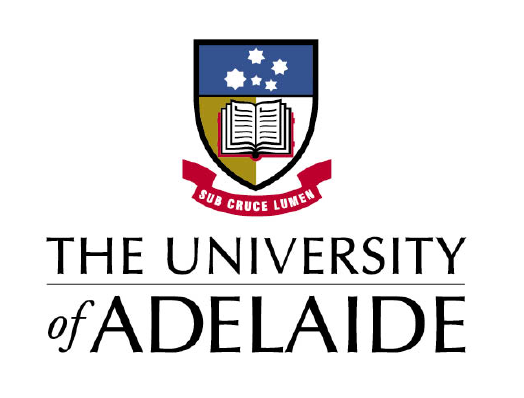}\\
\vspace{1cm}{\LARGE \bfseries \ttitle\par}\vspace{1.5cm}

{\large\authorname}  \\
\vspace{0.7cm}
\supname
\vfill
\vspace{1cm}

\large \textit{This thesis is submitted for the degree of Doctor of Philosophy}\\\textit{in}\\School of Computer and Mathematical Sciences\\The University of Adelaide\\[0.2cm] 
 
\vfill

{\large \today}\\[2.5cm] 

\vfill
\end{center}

\end{titlepage}


{
\hypersetup{linkcolor=black}
\tableofcontents 
\listoffigures 
\listoftables 
}

\begin{abbreviations}{p{0.1\textwidth}p{2\textwidth}} 
   AD & Active Directory \\
   NN & Neural Network \\
   EDO & Evolutionary Diversity Optimisation \\
   DA & Domain Admin \\
   MDP & Markov Decision Process \\
   DP & Dynamic Program \\
   NSP & Non-Splitting Paths \\
   BW & Block-Worthy \\
   MSE & Mean Squared Error \\
   RL & Reinforcement Learning \\
   PPO & Proximal Policy Optimization\\
   EC & Evolutionary Computation \\
   SH & Structural Hole \\
   SHS & Structural Hole Spanner \\
   BC & Betweenness Centrality \\
   GNN & Graph Neural Network \\
   ER & Erdos Renyi \\
   SF & Scale Free \\
   CC &  Closeness Centrality \\
   PC & Pairwise Connectivity \\
   DFS  & Depth First Search \\
   VC &  Vertex Cover \\   
\end{abbreviations}

\begin{symbols}{p{0.1\textwidth}p{2\textwidth}} 
    $p_{d(e)}$ & Detection probability \\
    $p_{f(e)}$ & Failure probability \\
    $p_{s(e)}$ & Success probability \\
    $s$ &  Entry nodes\\
    $k$ & Defender’s budget \\
    $h$  & feedback edges \\
    $bw$ & block-worthy \\    
    $G$ & Original graph \\
    $G'$ & Updated graph \\
    $G_t$ & Snapshot of graph at time $t$ \\
    $V,E$ & Set of nodes and edges \\ 
    $n,m$ & Number of nodes and edges\\
    $l$ & Index of aggregation layer \\ 
    $L$ & Total number of aggregation layers \\
    $||$ & Concatenation operator \\
    $\sigma$ & Non-linearity \\ 
    $z(i)$ & Final embedding of node $i$ \\
    $y(i)$ & Label of node $i$ \\
    $\vec{x}(i)$ & Feature vector of node $i$\\ 
    $d(i)$ & Degree of node $i$ \\
    $N(i)$ & Neighbors of node $i$ \\
    $h^{(l)}(i)$ & Embedding of node $i$ at the $l^{th}$ layer \\
    $m(i)$ & Number of edges in component containing node $i$ \\ 
    $P(G)$ & Pairwise connectivity of graph \\ 
    $P(G\backslash\{i\})$ & Pairwise connectivity of graph without node $i$ \\
    $p_{ij}$ & Path from node $i$ to $j$\\ 
    $u(i,j)$ & Pairwise connectivity between node $i$ and $j$ \\ 
    $c(i)$ & Pairwise connectivity score of node $i$ \\ 
    $c'(i)$ & Pairwise connectivity score of $i$ in updated graph\\
    $C$ & Connected component\\ 
    $C(i)$ & Connected component containing node $i$ \\
    $\mid C(i)\mid$ & Number of nodes in component containing node $i$ \\
    $\mid N(i)\mid$ & Number of neighbors of node $i$ \\
\end{symbols}

\mainmatter 



{\begin{doublespacing}
\chapter*{Abstract} 
\addcontentsline{toc}{chapter}{Abstract}
\setcounter{page}{11} 
\renewcommand{\thepage}{\roman{page}}


With the burgeoning advancements of computing and network communication technologies, network infrastructures and their application environments have become increasingly complex. Due to the increased complexity, networks are more prone to hardware faults and highly susceptible to cyber-attacks. Therefore, for rapidly growing network-centric applications, network resilience is essential to minimize the impact of attacks and to ensure that the network provides an acceptable level of services during attacks, faults or disruptions. In this regard, this thesis focuses on developing effective approaches for enhancing network resilience. Existing approaches for enhancing network resilience emphasize on determining bottleneck nodes and edges in the network and designing proactive responses to safeguard the network against attacks. However, existing solutions generally consider broader application domains and possess limited applicability when applied to specific application areas such as cyber defense and information diffusion, which are highly popular application domains among cyber attackers. These solutions often prioritize general security measures and may not be able to address the complex targeted cyberattacks \cite{wandelt2021estimation, wang2021method}. Cyber defense and information diffusion application domains usually consist of sensitive networks that attackers target to gain unauthorized access, potentially causing significant financial and reputational loss.

\textit{This thesis aims to design effective, efficient and scalable techniques for discovering bottleneck nodes and edges in the network to enhance network resilience in cyber defense and information diffusion application domains.}  We first investigate a cyber defense graph optimization problem, i.e., \textit{hardening active directory systems by discovering bottleneck edges in the network}. We then study the problem of \textit{identifying bottleneck structural hole spanner nodes, which are crucial for information diffusion in the network}. We transform both problems into graph-combinatorial optimization problems and design machine learning based approaches for discovering bottleneck points vital for enhancing network resilience. This thesis makes the following four contributions. We first study defending active directories by discovering bottleneck edges in the network and make the following two contributions. (1) To defend active directories by discovering and blocking bottleneck edges in the graphs, we first prove that deriving an optimal defensive policy is \#P-hard. We design a kernelization technique that reduces the active directory graph to a much smaller condensed graph. We propose an effective edge-blocking defensive policy by \textit{combining neural network-based dynamic program and evolutionary diversity optimization} to defend active directory graphs. The key idea is to accurately train the attacking policy to obtain an effective defensive policy. The experimental evaluations on synthetic AD attack graphs demonstrate that our defensive policy generates effective defense. (2) To harden large-scale active directory graphs, we propose \textit{reinforcement learning based policy that uses evolutionary diversity optimization} to generate edge-blocking defensive plans. The main idea is to train the attacker’s policy on \textit{multiple independent defensive plan environments simultaneously so as to obtain effective defensive policy}. The experimental results on synthetic AD graphs show that the proposed defensive policy is highly effective, scales better and generates better defensive plans than our previously proposed neural network-based dynamic program and evolutionary diversity optimization approach. We then investigate discovering bottleneck structural hole spanner nodes in the network and make the following two contributions. (3) To discover bottleneck structural hole spanner nodes in large-scale and diverse networks, we propose two \textit{graph neural network models, GraphSHS and Meta-GraphSHS}. The main idea is to transform the SHS identification problem into a learning problem and use the graph neural network models to learn the bottleneck nodes. Besides, the Meta-GraphSHS model learns generalizable knowledge from diverse training graphs to create a customized model that can be fine-tuned to discover SHSs in new unseen diverse graphs. Our experimental results show that the proposed models are highly effective and efficient. (4) To identify bottleneck structural hole spanner nodes in dynamic networks, we propose a \textit{decremental algorithm and graph neural network model}. The key idea of our proposed algorithm is to reduce the re-computations by identifying affected nodes due to updates in the network and performing re-computations for affected nodes only. Our graph neural network model considers the dynamic network as a series of snapshots and learns to discover SHS nodes in these snapshots. Our experiments demonstrate that the proposed approaches achieve significant speedup over re-computations for dynamic graphs. 
\chapter*{Declaration} 
\addcontentsline{toc}{chapter}{Declaration}
\setcounter{page}{14} 
\renewcommand{\thepage}{\roman{page}}

\label{Declaration} 

I certify that this work contains no material which has been accepted for the award of any other degree or diploma in my name, in any
university or other tertiary institution and, to the best of my knowledge and belief, contains no material previously published or written by
another person, except where due reference has been made in the text. In addition, I certify that no part of this work will, in the future, be used
in a submission in my name, for any other degree or diploma in any university or other tertiary institution without the prior approval of the
University of Adelaide and where applicable, any partner institution responsible for the joint award of this degree.

I give permission for the digital version of my thesis to be made available on the web, via the University’s digital research repository, the
Library Search and also through web search engines, unless permission has been granted by the University to restrict access for a period of
time.

\vspace{5mm}
\begin{flushright}Diksha Goel\end{flushright}
\vspace{-0.1in}
\begin{flushright}September 2023\end{flushright}

\newpage
\vspace*{0.2\textheight} 
\begin{center}
\emph{Dedicated to my Father, Mother, Sachin and Nikhil.}
\end{center}

\vspace*{\fill} 

\chapter*{Acknowledgement} 
\addcontentsline{toc}{chapter}{Acknowledgement}
\setcounter{page}{16} 
\renewcommand{\thepage}{\roman{page}}

I express my sincere gratitude to my principal supervisor, Dr Mingyu Guo, for his invaluable support, encouragement and guidance throughout my PhD journey. I am deeply grateful to him for always being available for constructive discussions, feedback on papers and motivation during the times of low morale. His mentorship has always inspired me to strive for excellence and excel beyond my limits. I am thankful to him for showing me different aspects of thinking about a problem. His patience and persistent guidance have helped me to work under pressure and complete many deadlines. I am grateful for his unwavering belief in me and his willingness to go above and beyond his duties as a supervisor. He has shaped me as an independent researcher, and I am fortunate to be mentored by him.

I also extend my deepest gratitude to my co-supervisors, Prof. Hong Shen and Prof. Hui Tian, for their support, guidance and assistance during my PhD. Prof. Shen has helped me in developing analytical skills and I am grateful to him for all the brainstorming sessions that have greatly improved the quality of my work. I am thankful to Prof. Tian for constructive discussions, her technical insights and feedback on papers. Their feedback, guidance and mentorship helped me to stay focused and work toward achieving my research goals. I would also like to thank my collaborators, Prof. Frank Neumann, Prof. Hung Nguyen, Dr Aneta Neumann and Dr Max Hector Ward-Graham, for their valuable expertise and insights.

I sincerely thank my friends who have been a constant source of moral support during this challenging journey. Firstly, I would like to express my heartfelt thanks to Hussain for his continuous encouragement, invaluable insights, discussions, and guidance, all of which have been indispensable in my research accomplishments. Indeed, he is truly a gem of a friend. I am grateful to Megha for providing reassurance during difficult times; her presence in my life has been a source of great strength. I extend my thanks to Shagun, Shakshi, Ritu and Daniel for their friendship and unwavering support. I am truly fortunate to have such wonderful friends who have supported me during difficult times.

Finally, I am indebted to my family for their immense support and endless love. I am grateful to my mother for her sacrifices, efforts and persistent belief in me. I am thankful to my father for always supporting and encouraging me. He has always been there for me, offering his wisdom and guidance. He has been a constant source of inspiration for me, and I always look up to him. My younger brothers, Sachin and Nikhil, have always encouraged me to pursue my dreams and reminded me of my strengths. I am blessed to have such a wonderful family, and I owe them an immense debt of gratitude.


\chapter{Introduction} 
\setcounter{page}{1}
\renewcommand{\thepage}{\arabic{page}}

\label{introduction} 

The past few years witnessed tremendous growth in various large-scale networks, including social, collaboration, biological, semantic and criminal networks \cite{freeman2004development, moody2004structure, bell1997transportation, scardoni2009analyzing, doerfel1999semantic,  xu2005criminal}. These networks are essential for cyber defense, information diffusion, transportation, communication and other critical functionalities \cite{sheng2007context, olatinwo2019enabling, yang2010modeling, borgatti2018analyzing}. Moreover, with the advancements in network communication technologies, networks have become increasingly complex and interdependent, making them vulnerable to various disruptions, such as targeted attacks, faults and failures. These disruptions undermine the ability of the network to meet the fundamental functionalities and lead to service outages, data loss, infrastructure damage and network breakdown, with significant economic and social loss \cite{cashell2004economic, zimba2019economic}. Network resilience enables the network to withstand and recover from disruptions quickly and effectively, in turn minimizing the impacts of attacks \cite{tipper2014resilient}. Resilient networks are essential for reducing the downtime associated with service interruptions, preventing further damages or cascading failures \cite{barker2017defining}. Besides, a resilient network enhances security and makes it more difficult for cyber attackers to penetrate or disrupt the network. Therefore, it is essential to design solutions for enhancing network resilience. Attackers generally target bottleneck nodes or edges in the network to compromise its capability to provide essential network services \cite{dinh2010approximation,ye2004routing}. Cyber attacks on bottleneck points in the network may result in total network failure. Therefore, it is essential to determine the bottleneck nodes and edges that hold advantageous positions in the network and develop proactive strategies to protect these points and improve network resilience \cite{dinh2011new}. In the literature, several approaches have been proposed to assess network vulnerabilities by discovering key network nodes and edges in order to improve network resilience against cyber attacks \cite{dinh2011new, dinh2011precise, dinh2010approximation, shen2012discovery}. However, the limitation is that most of the existing solutions are designed for general application domains and have limited relevance when applied to particular application domains, such as cyber defense and information diffusion. Networks in cyber defense and information diffusion application domains are highly sensitive, and attacks on these networks may cause significant financial and reputational loss to the organizations. Hence, there is a need to design effective solutions for enhancing network resilience in cyber defense and information diffusion application domains.


This thesis aims to design effective, efficient and scalable techniques for identifying bottleneck nodes and edges in the network to improve the overall network resilience. We aim to enhance network resilience by focusing on cyber defense and information diffusion application domains. Accordingly, this thesis investigates two graph-combinatorial optimization problems, first in cyber defense and second in the information diffusion application domain. We first investigate a cyber defense problem of defending active directory graphs\footnote{In this thesis, “network” and “graph” are used interchangeably.} by discovering bottleneck edges in the network. We then study the problem of finding bottleneck structural hole spanner nodes, which are crucial for information diffusion in the network. The description of both application domains is discussed below.

\section{Cyber Defense}
Cyber attackers target sensitive networks to gain unauthorized access to the systems, potentially causing significant financial and reputational loss \cite{abomhara2015cyber}. In contrast, cyber defense aims to defend networks, systems and data from cyber attackers. One critical aspect of cyber defense is identifying network vulnerabilities and implementing security policies \cite{skopik2016problem}. However, the increased complexity and sensitivity of the networks have made it challenging to discover and mitigate the network vulnerabilities \cite{saxena2020impact}. Organizational networks are highly sensitive, and the attacker tries to reach high-privileged accounts to gain unauthorized access to valuable information \cite{husak2018survey}. One way of defending the network is to discover bottleneck points and block access to these points to prevent attackers from reaching the high-privileged accounts. However, blocking access to the nodes in an organizational network can negatively impact the network functionality. Another solution is to restrict access to the bottleneck edges in the network, preventing cyber attackers from reaching the high-privileged accounts. Blocking access to the bottleneck edges improves network security and makes it more resilient to cyber-attacks. This thesis studies a specific application scenario in cyber defense, i.e., “\textit{\textbf{defending active directory graphs from cyber attackers}}” by blocking access to the bottleneck edges so as to make the network more resilient. 

\vspace{0.15in}

\noindent \textbf{Active Directory.} Active Directory (AD) is a directory service developed by Microsoft for \textit{Windows domain networks}. It is designed to manage and secure network resources, such as user accounts, computers, printers, etc. AD allows administrators to assign permissions and monitor network activities, making managing and authenticating users and computers in enterprise environments crucial. AD is considered a default security management system for Windows domain networks \cite{dias2002guide}. Microsoft domain network constitutes significant market shares among small and big organizations globally, due to which AD are promising targets for cyber attacks. Microsoft reported that 90\% of Fortune 1000 companies use AD, and as per the study by Enterprise Management Associates [2021], 50\% of the surveyed organizations have experienced an AD attack since 2019. Microsoft reported that 95 million AD accounts are targeted daily by cyber attackers, and 1.2 million Azure AD accounts are compromised monthly. Moreover, 80\% of all breaches target access to the privileged account, called \textsc{Domain Admin} (DA). 
\begin{figure}[h!]
\centering
\includegraphics[width=0.25\paperwidth]{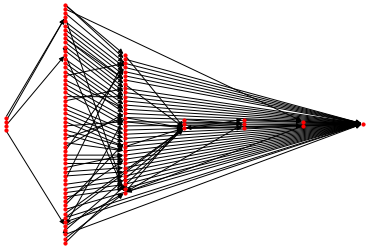}
 \caption{AD attack graph containing 500 computers.}
 \label{fig:attack_graph}
\end{figure}

AD attack graphs are widely used attack graph models among industrial practitioners and real-world attackers to model and analyze potential attack scenarios in an AD environment. The structure of AD describes an attack graph, with a node representing accounts/computers/etc., and directed edge $(i,j)$ indicating that an attacker can gain access to node $j$ from node $i$ via known exploits or existing accesses. Various applications and tools are designed to investigate the AD attack graphs; however, \textsc{BloodHound} is the most popular tool, that identifies various attack paths in AD graph structures. \textcolor{blue}{Figure \ref{fig:attack_graph}} illustrates an instance of AD attack graph created using \textsc{DBCreator\footnote{\textsc{DBCreator} is a synthetic AD graph generator tool from the \textsc{BloodHound} team.}} and only the nodes reachable to DA are shown. \textsc{BloodHound} simulates an identity snowball attack, in which an attacker begins from a low-privileged account (gains access through phishing attack) and subsequently moves to other nodes with the goal of reaching the highest-privileged account, DA. \textcolor{blue}{Figure \ref{fig:bloodhound}} illustrates an example of \textsc{BloodHound} snowball identity attack. \textsc{BloodHound} employs Dijkstra’s shortest path algorithm to determine the path from entry node to DA. \textsc{BloodHound} has made it much easier for the attacker to attack AD. Given the sensitive nature of organizational AD and high number of cyber attacks targeting AD, security professionals are designing various solutions to defend AD. \textsc{BloodHound} is motivated by an academic paper \cite{dunagan2009heat}, where the authors designed a heuristic to selectively block some edges from the attack graph to disconnect the graph and prevent attackers from reaching DA. Notably, edge blocking in an AD environment can be attained by either monitoring the edges or revoking access. 
\textsc{GoodHound}\footnote{https://github.com/idnahacks/GoodHound}, \textsc{BloodHound Enterprise}\footnote{https://bloodhoundenterprise.io/} and \textsc{ImproHound}\footnote{https://github.com/improsec/ImproHound} are various industry solutions that use edge blocking to provide defense. 

\begin{figure}[h!]
\centering
 \includegraphics[width=0.40\paperwidth]{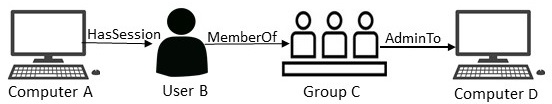}
 \caption{Example of \textsc{BloodHound} snowball identity attack.}
 \label{fig:bloodhound}
\end{figure}

This thesis aims to defend AD graphs by identifying the bottleneck edges that can be blocked to prevent the attackers from reaching the DA. We convert the problem of defending AD graphs into a graph combinational optimization problem and design machine learning-based defensive policies to defend AD graphs. 

\section{Information Diffusion}
Networks have become increasingly crucial for modelling interactions, as they provide frameworks for understanding the connections between various entities and the flow of information between them. Various application domains where the modelling of interactions is crucial include information diffusion, social networks, collaboration networks, email networks, transportation systems and power grids. Cyber attacker targets network entities, specifically bottleneck nodes, which are imperative for information diffusion in the network, to render the network useless by disrupting the information flow. This may lead to severe consequences, including communication loss, financial loss, reputation damage, service downtime and operational disruptions. Bottleneck nodes are crucial for optimizing information dissemination and improving the network’s efficiency. Therefore, it is essential to discover bottleneck network nodes and protect them from cyber attackers to ensure efficient and effective information dissemination, enhancing network resilience. This thesis studies a specific application scenario in the information diffusion domain, i.e., “\textit{\textbf{discovering bottleneck structural hole spanner nodes in the network}}” to enhance the network’s resilience.

\vspace{0.15in}

\noindent \textbf{Structural Hole Spanners. } The last decade witnessed tremendous growth of various large-scale networks, such as biological, semantic, collaboration, criminal and social networks. There is a huge demand for efficient and scalable solutions to study the properties of these large networks. A network consists of communities where the nodes share similar characteristics \cite{chen2019contextual}, and these communities are crucial for information diffusion in the network. The nodes having connections with the diverse communities get positional advantages in the network. This notion serves as a base for the {Theory of} {Structural Holes} \cite{burt2009structural}. The theory states that the \textit{{Structural Holes (SH)}} are the positions in the network that can bridge different communities and bring the beholders into an advantageous position. The absence of connections between different communities creates gaps, which is the primary reason for the formation of SHs in the network \cite{lou2013mining}. The nodes that fill SHs by bridging different communities are known as \textit{\textbf{Structural Hole Spanners (SHS)}} \cite{lou2013mining}. SHSs get various positional benefits such as access to novel ideas from diverse communities, more control over information flow, etc. \textcolor{blue}{Figure \ref{fig:SHS}} shows the SHS between communities in the network. There are many vital applications of SHSs, such as community detection \cite{gupta2020overlapping}, opinion control \cite{kuhlman2013controlling}, information diffusion \cite{bonifazi2022approach, zhang2019most}, viral marketing \cite{castiglione2020cognitive} etc. In case of an epidemic disease, discovering SHSs and quarantining them can help stop the spread of infection. In addition, SHSs can be used to advertise a product to different groups of users for viral marketing. 
\begin{figure}[t!]
 \centering
 \includegraphics[width=0.35\paperwidth]{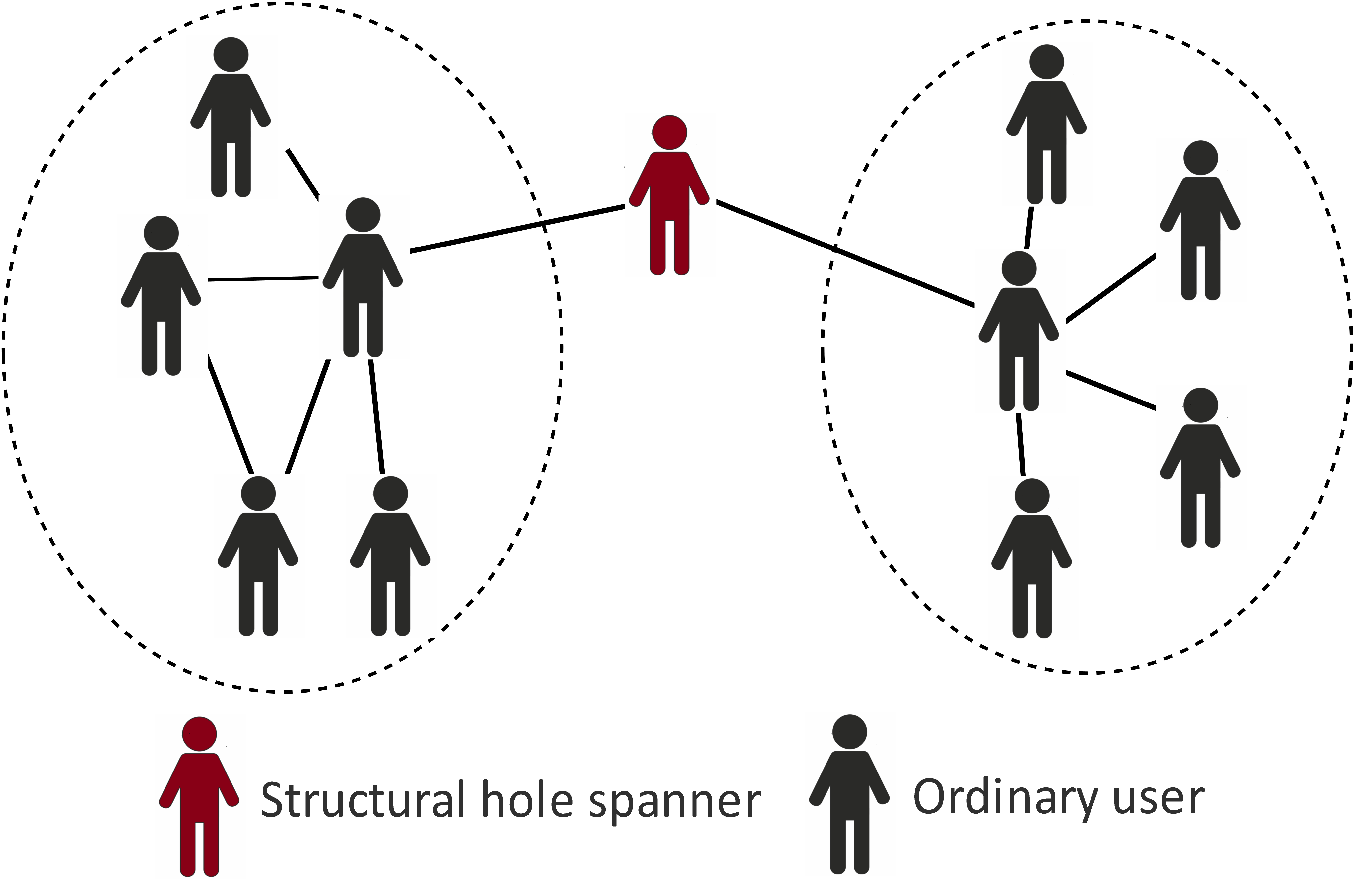}
 \caption{Illustration of SHS in network.}
 \label{fig:SHS}
\end{figure}
A number of centrality measures such as Closeness Centrality \cite{rezvani2015identifying}, Constraint \cite{burt1992structural}, and Betweenness Centrality (BC) \cite{freeman1977set} exist in the literature to define SHSs. SHS nodes lie on the maximum number of shortest paths between the communities \cite{rezvani2015identifying}; removal of the SHS nodes will disconnect multiple communities and block information flow among the nodes of the communities \cite{lou2013mining}. Based on this, we have two implications about the properties of SHSs; 1) SHSs bridge multiple communities; 2) SHSs control information diffusion in the network. Lou et al. \cite{lou2013mining} showed that the removal of an SHS node disconnects a maximum number of communities and blocks information propagation between the nodes of the communities. Attack on SHS nodes will disrupt the information propagation in the network; therefore, it is essential to discover and safeguard these nodes.

This thesis aims to discover bottleneck SHS nodes responsible for information diffusion in the network, and attacks on these nodes can significantly impact the spread of information. Therefore, by discovering and protecting these nodes, we can prevent cyber attackers from exploiting the vulnerabilities in these nodes to compromise the network, making it harder for the attacker to disrupt the network information flow. This thesis develops algorithmic and machine learning-based solutions to discover SHS nodes in the network.

\section{Research Scope and Challenges}

This thesis aims to discover bottleneck nodes and edges in the network to enhance network resilience by focusing on cyber defense and information diffusion application domains. In the context of cyber defense, we aim to address the problem of “\textit{defending AD graphs by discovering and blocking bottleneck edges in the network}” so as to enhance the network’s ability to withstand cyber attacks. Additionally, we study the problem of “\textit{identifying bottleneck SHS nodes in the network}”. SHS nodes are crucial for optimizing the information flow in the network, and attacks on these nodes may severely disrupt the network information flow.

\subsection{Research Scope}
This thesis investigates the above-mentioned research topics and addresses the following specific research problems.

\begin{itemize}
  \item For \textit{defending active directory graphs}, we focus on two key research problems. Firstly, we investigate the problem of designing effective defensive policies for discovering bottleneck edges in the AD graphs to protect these graphs from cyber attackers. Secondly, we address the scalability issue that arises while designing scalable defensive policies for defending large-scale AD graphs.

 \item For \textit{discovering structural hole spanner nodes}, we focus on identifying bottleneck SHS nodes in different types of networks. Firstly, we study discovering SHS nodes in static networks, i.e., the networks where the structure remains static. Secondly, we consider identifying SHS nodes in diverse networks, i.e., the networks with different structures. Finally, we investigate the problem of identifying SHS nodes in dynamic networks, i.e., the networks where the structure evolves continuously.
\end{itemize}

\subsection{Research Challenges}

Given the research scope, this thesis addresses the following research challenges:
\begin{enumerate}
    \item \textbf{To design effective and scalable edge-blocking defensive policies for defending active directory graphs.}
    
    Due to the popularity of Active Directories among small and large organizations, AD systems have become an attractive target for cyber attackers. However, defending AD systems is a novel research area, and not much work has been done in this field. Considering the highly sensitive nature of organizational AD systems, any unauthorized access to AD may cause significant financial and reputation loss to the organization. Therefore, there is a need to design effective approaches for defending AD graphs from cyber attackers. Moreover, the designed strategies should be scalable to meet the current organizational network size and accommodate any possible future growth.

    \item \textbf{To design effective and efficient approaches for discovering SHS nodes in large-scale and diverse networks.}
    
    Existing solutions for discovering SHS nodes require high runtime and fail to scale to large networks. Therefore, there is a need to design approaches for discovering SHS nodes in large-scale networks. Generally, machine learning models are used to discover SHS nodes in a network by analyzing patterns in the network and utilizing graph-based algorithms to identify and categorize nodes based on their relationships and attributes. This enables automated node detection and classification in networks. However, the challenge is discovering SHS nodes across different types of networks for which the traditional one-model-fit-all approach fails to capture the inter-graph differences, particularly in the case of diverse networks. Besides, re-training a machine learning model on different types of large networks can be time-consuming. Therefore, it is essential to design machine learning models that are aware of differences across graphs and customize accordingly, avoiding the need to re-train the model for every type of network individually.
    
    \item \textbf{To develop efficient approaches for identifying SHS nodes in dynamic networks.}
    
    Numerous solutions have been proposed to discover SHS nodes in static networks; however, real-world networks are highly dynamic, due to which SHS nodes in the network change over time. Currently, there is no solution that discovers SHS nodes in dynamic networks. Traditional SHS discovering algorithms are time-consuming and might not work efficiently for dynamic networks. Additionally, the network might have already changed by the time these conventional algorithms recompute SHSs in the new network. Therefore, there is a need to develop efficient techniques that can quickly discover SHS nodes in dynamic networks.
\end{enumerate}

\section{Thesis Contributions and Publications}
The primary objective of this thesis is to design effective, efficient and scalable approaches for discovering bottleneck nodes and edges in the network. This thesis aims to enhance network resilience in cyber defense and information diffusion application domains. In this regard, we study the following two research problems. 

\begin{itemize}

    \item Firstly, this thesis investigates a cyber defense graph combinatorial optimization problem of hardening active directory systems by discovering bottleneck edges in the AD graph to improve the overall resilience of AD systems. Defending AD systems is a crucial research problem and has practical applications in organizations that depend on AD for authentication and authorization. AD is a security management system that holds information about users, computers and other resources in an organization. In case the attacker gains unauthorized access to the organizational AD systems, all the organizational confidential information may be compromised. Therefore, there is an urgent need to design defensive approaches for protecting AD graphs and preventing data breaches.
    
    \item Secondly, this thesis studies the problem of identifying bottleneck nodes, known as SHS nodes, which are essential for information diffusion in the network. Discovering bottleneck SHS nodes is a theoretical research area in network analysis, and mathematical techniques are used to analyze the relationships between different nodes to discover SHSs in the network. Even though the SHS discovering problem is theoretical in nature, it has substantial practical implications in terms of network analysis and optimization. The organization can use SHS nodes to strengthen network connectivity and improve efficiency and resilience. This research area has numerous applications in various domains, including social networks, recommendation systems and biological networks. 
\end{itemize}

\noindent To address the above-discussed problems, we convert the aforementioned problems into graph-combinatorial optimization problems and aim to discover bottleneck nodes and edges in the network to enhance network resilience in cyber defense and information diffusion application domains. We leverage state-of-the-art machine learning and graph combinatorial optimization techniques to design novel approaches for identifying bottleneck nodes and edges in the network. We propose various neural network, reinforcement learning, graph neural network and evolutionary diversity optimization-based solutions for addressing the aforementioned research problems. Notably, this thesis is an amalgam of practical and theoretical results. 

The outcome of this thesis is effective, efficient and scalable approaches for discovering bottleneck nodes and edges in the network, in turn enhancing network resilience. The proposed approaches will provide organizations with the necessary solutions to determine the bottleneck points in the network; and once these points are discovered, essential proactive actions can be performed to protect the organizational networks. This thesis makes the following contributions.

\vspace{0.15in}

\noindent \textbf{\textcolor{blue}{Chapter \ref{Chapter_AD_NNDP}} -}  This chapter presents an effective defensive policy for determining the bottleneck edges that can be blocked to defend organizational AD graphs. In this chapter, we study a Stackelberg game model between one attacker and one defender on an AD attack graph, where the attacker’s goal is to maximize their chances of successfully reaching the most privileged account. The defender aims to block a set of edges to minimize the attacker’s success rate. We first show that the problem is \#P-hard; therefore, it can not be solved exactly. We design a kernelization procedure that exploits the structural features of AD graphs to obtain a much smaller condensed graph and converts the attacker’s problem from a condensed graph to a dynamic program. For small AD graphs, we can solve the dynamic program; however, for larger AD graphs, it is computationally challenging to solve the dynamic program. Therefore, to solve the attacker’s problem, we propose an approach that involves training a Neural Network (NN) to learn the recursive relationship of the dynamic program. Besides, we design an Evolutionary Diversity Optimization (EDO) based policy to solve the defender’s problem, i.e., to determine which edges to block in the AD graph. The trained NN serves as an efficient fitness function for the defender’s EDO. Moreover, EDO generates a diverse set of blocking plans that act as training samples for the NN. We perform extensive experiments on synthetic AD attack graphs, and our experimental results show that the proposed approach is highly effective. This chapter has already been published as:
\begin{itemize}
\item[\ding{182}] \textbf{Diksha Goel}, Max Hector Ward-Graham, Aneta Neumann, Frank Neumann, Hung Nguyen and Mingyu Guo, \textit{Defending Active Directory by Combining Neural Network based Dynamic Program and Evolutionary Diversity Optimisation}, In Proceedings of the Genetic and Evolutionary Computation Conference (GECCO), Boston, US, 2022 \cite{Goel2022defending}.
\end{itemize}

\vspace{0.15in}

\noindent \textbf{\textcolor{blue}{Chapter \ref{AD_RL}} -} This chapter presents another defensive edge-blocking policy for hardening AD graphs. This policy is specifically designed to be applicable and effective in the context of large-scale AD graphs. In this chapter, we study attacker-defender Stackelberg game on an AD graph in configurable environment settings, where each environment represents an edge-blocking plan. The defender tries various environment configurations to develop the best defensive configuration and protect the AD graphs. In contrast, the attacker observes the environment configurations and designs an attacking policy to maximize their chances of reaching the DA. We propose a {{Reinforcement Learning (RL)}} based policy to maximize the attacker’s chances of successfully reaching the DA. The RL agent simultaneously interacts with “\textit{multiple independent environments}” by suggesting actions to maximize the overall reward. We design a Critic network-assisted Evolutionary Diversity Optimization based approach that generates numerous environment configurations to solve the defender’s problem. Defender’s approach utilizes the trained RL critic network to evaluate the fitness of the environment configurations. The defender adopts the technique of rejecting the configurations that are advantageous for the attacker and replacing them with better ones. The attacker and defender play against each other in parallel. We perform experiments on large-scale AD attack graphs, and our results demonstrate that our proposed defensive approach is highly effective, approximates the attacker’s problem more accurately, generates better defensive plans and scales better. This chapter has already been published as:

\begin{itemize}
\item[\ding{183}] \textbf{Diksha Goel}, Aneta Neumann, Frank Neumann, Hung Nguyen and Mingyu Guo, \textit{Evolving Reinforcement Learning Environment to Minimize Learner’s Achievable Reward: An Application on Hardening Active Directory Systems}, In Proceedings of the Genetic and Evolutionary Computation Conference (GECCO), Lisbon, Portugal, 2023 \cite{goel2023evolving}. 
\end{itemize}

\vspace{0.15in}

\noindent \textbf{\textcolor{blue}{Chapter \ref{SHS_GNN}} -} This chapter presents various effective and efficient approaches for identifying bottleneck structural hole spanner nodes in large-scale and diverse networks. We first propose {{GraphSHS}}, a graph neural network model for efficiently discovering SHSs in large-scale networks. GraphSHS aims to minimize the computational cost while achieving high accuracy. GraphSHS uses the network structure and features of nodes to learn the low-dimensional node embeddings and then uses these embeddings to discover SHS nodes in the network. Inter-graph differences exist in diverse graphs, due to which it is not possible for GraphSHS to effectively discover SHSs across diverse networks. Therefore, to discover SHSs across diverse networks, we propose {\textit{Meta-GraphSHS}}, {meta-}learning based {graph} neural network for discovering {s}tructural {h}ole {s}panner nodes. Meta-GraphSHS aims to effectively discover SHS nodes in diverse networks without re-training the model on every network dataset to adapt to cross-network property changes. Meta-GraphSHS is based on learning generalizable knowledge from diverse training graphs and then using the learned knowledge to create a customized model by fine-tuning the parameters to suit each new graph. We theoretically show that the depth of the proposed graph neural network model should be at least $\Omega(\sqrt{n}/\log n)$ to calculate the SHSs problem accurately. We evaluate the performance of the proposed models through extensive experiments on synthetic and real-world datasets. Our experimental results show that both the proposed models GraphSHS and Meta-GraphSHS are highly efficient and effective in discovering SHSs in large-scale networks and diverse networks. This chapter is under submission as:

\begin{itemize}
\item[\ding{184}] \textbf{Diksha Goel}, Hong Shen, Hui Tian and Mingyu Guo, \textit{Effective Graph-Neural-Network based Models for Discovering Structural Hole Spanners in Large-Scale and Diverse Networks} \cite{goel2023effective}. 
\end{itemize}

\vspace{0.15in}

\noindent \textbf{\textcolor{blue}{Chapter \ref{SHS_LCN}} -} This chapter presents efficient approaches for discovering bottleneck structural hole spanner nodes in dynamic networks. Our approaches aim to reduce the computational costs while achieving high accuracy. To discover SHS nodes in dynamic networks, we first propose an efficient {{Tracking-SHS algorithm}} that maintains Top-$k$ SHS nodes by discovering a set of affected nodes. Tracking-SHS aims to maintain and update the SHS nodes faster than recomputing them from the ground. We derive some properties to determine the affected nodes due to updates in the network. In order to avoid unnecessary recomputations for unaffected nodes, Tracking-SHS algorithm utilizes the knowledge from the initial runs of the static algorithm. In addition, we propose {{GNN-SHS}}, a graph neural network model that discovers SHSs in dynamic networks. GNN-SHS model works for both incremental and decremental edge updates of the network. We consider the dynamic network as a series of snapshots and discover SHSs in these snapshots. GNN-SHS aggregates embeddings from the node’s neighbors, and the final embeddings are used to identify SHS nodes. We conduct a theoretical analysis of the proposed Tracking-SHS algorithm, and our results demonstrate that the algorithm achieves high speedup over recomputations. Besides, we perform experiments on synthetic and real-world datasets. Our results demonstrate that the proposed Tracking-SHS algorithm and GNN-SHS model achieve significant speedup over baselines. This chapter is based on our following three papers:

\begin{itemize}
\item[\ding{185}] \textbf{Diksha Goel}, Hong Shen, Hui Tian and Mingyu Guo, \textit{Maintenance of Structural Hole Spanners in Dynamic Networks}, In 46\textsuperscript{th} IEEE Conference on Local Computer Networks (LCN), Edmonton, AB, Canada, 2021 \cite{goel2021maintenance}.

\item[\ding{186}] \textbf{Diksha Goel}, Hong Shen, Hui Tian and Mingyu Guo, \textit{Discovering Structural Hole Spanners in Dynamic Networks via Graph Neural Networks}, In The 21\textsuperscript{st} IEEE/WIC/ACM International Conference on Web Intelligence and Intelligent Agent Technology (WI-IAT), Niagara Falls, Canada, 2022 \cite{goel2022discovering}.

\item[\ding{187}] \textbf{Diksha Goel}, Hong Shen, Hui Tian and Mingyu Guo, \textit{Discovering Top-k Structural Hole Spanners in Dynamic Networks} \cite{goel2023discovering}.
\end{itemize}

\section{Other Publications}
In addition to the six publications mentioned above, I co-authored two papers during my PhD candidature. Despite being closely related, the contributions from these publications are not included in this thesis. The publications are listed below:

\begin{itemize}
\item[\ding{188}] Mingyu Guo, \textbf{Diksha Goel}, Guanhua Wang, Yong Yang, Muhammad Ali Babar, \textit{Cost Sharing Public Project with Minimum Release Delay} \cite{guo2022cost}.

\item[\ding{189}] Mingyu Guo, \textbf{Diksha Goel}, Guanhua Wang, Runqi Guo, Yuko Sakurai, Muhammad Ali Babar, \textit{Mechanism Design for Public Projects via Three Machine Learning Based Approaches} \cite{guo2022mechanism}.
\end{itemize}

\section{Thesis Organization}
The rest of the thesis is organized as follows. \textcolor{blue}{Chapter \ref{Related_Work}} provides a comprehensive literature review of the problems investigated in this thesis. \textcolor{blue}{Chapter \ref{Chapter_AD_NNDP}} studies effective edge-blocking defensive policies for defending AD graphs from cyber attackers. \textcolor{blue}{Chapter \ref{AD_RL}} presents effective and scalable edge-blocking strategies for protecting large-scale active directories. \textcolor{blue}{Chapter \ref{SHS_GNN}} proposes graph neural network models for discovering bottleneck SHS nodes in large-scale and diverse networks. \textcolor{blue}{Chapter \ref{SHS_LCN}} presents an algorithm and graph neural network model to identify SHS nodes in dynamic networks. Finally, \textcolor{blue}{Chapter \ref{Chapter_conclusion}} summarizes the main contributions of this thesis and suggests possible extensions and areas for future work.

\chapter{Literature Review} 

\label{Related_Work} 

Identifying bottleneck nodes and edges in the network and developing proactive responses to protect these points can enhance network resilience. We specifically focus on enhancing network resilience with applications in cyber defense and the information diffusion domain. In this chapter, we first review the existing literature on defending AD graphs, including methods for discovering the most vital edges in the network. Furthermore, in order to model the attacker-defender interactions on an AD graph, we review the existing Stackleberg game models on attack graphs. Later, we analyze the existing literature on identifying bottleneck SHS nodes, which are crucial for information diffusion in the network. We particularly focus on identifying SHS nodes in various types of networks, including static, diverse and dynamic networks. 


\section{{Active Directories}}
This section discusses various techniques for defending AD graphs from cyber attackers. We first discuss the existing solutions available to protect AD. We then explore existing work for determining the most vital edges in the graph. Lastly, we investigate the Stackelberg game models on attack graphs.

\subsection{Defending Active Directory}
Active Directory is the default security management system for Windows Domain Networks \cite{dias2002guide}. The primary functionality of AD is to facilitate administrators to handle the permissions and manage access to the network resources. Consequently, it is necessary to defend the organizational AD to protect their assets and ensure smooth functioning. Generally, edge-blocking strategies are used to protect the AD graphs from cyber attackers. Edge blocking in an AD graph is accomplished by either revoking access or implementing surveillance measures to prevent attackers from reaching the domain admin \cite{dunagan2009heat}. Guo et al. \cite{guo2021practical} studied the shortest path edge interdiction problem for protecting AD graphs from cyber attackers. The authors formalized the problem as a Stackelberg game between an attacker and a defender, where the attacker seeks to reach the highest privileged account via the shortest attack path, and the defender aims to maximize the expected path length of the attacker. In order to solve the problem, the authors designed various fixed-parameter algorithms, including a tree decomposition-based algorithm and another one that considers a small number of splitting nodes in the graph. Besides, the authors proposed a graph convolutional neural network based approach that is scalable to larger AD graphs. In another study, Guo et al. \cite{guo2022scalable} investigated a Stackelberg game that involves one attacker and one defender. The attacker attempts to reach the domain admin via a path that has high success rate, and the defender attempts to minimize the attacker's success rate by blocking a constant number of edges. The authors exploited the tree-like structure of AD graphs and designed various scalable algorithms. The authors also designed reinforcement learning and mixed integer programming-based techniques to further improve the scalability of the proposed approaches. Guo et al. \cite{guo2023limited} proposed a limited query graph connectivity test model to determine the connectivity between two nodes in a graph. Each edge in the graph has a binary state that can be queried to reveal its status. The goal is to design a querying strategy that minimizes the required number of queries to determine a path between two nodes while operating under a limited query budget. The authors proposed an algorithm that scales well on larger graphs and evaluated its performance on numerous practical graphs, including AD graphs and power networks. Ngo et al. \cite{quang} designed various near-optimal policies for placing honeypots on computer nodes when the AD graph changes dynamically. The authors studied the problem of identifying bottleneck nodes in the network that can be used for monitoring purposes. Zhang et al. \cite{Yumeng23:Near} proposed a scalable double oracle algorithm for defending AD graphs and compared their solution against various industry solutions.

\textbf{Challenge:} Active directories have increasingly become an attractive target for cyber attackers. However, a very limited amount of work has been done focusing on blocking edges to defend the AD graphs. Consequently, this thesis designs effective and scalable edge-blocking defensive policies to safeguard active directory graphs.

\subsection{Discovering Most Vital Edges}
Bar-Noy et al. \cite{bar1998complexity} proved that identifying the $k$-most vital arcs or nodes in the graph is an NP-hard problem. The authors described the $k$ most vital arcs as the set of edges whose removal leads to the greatest increase in the shortest path between two nodes in network. Bazgan et al. \cite{bazgan2019more} studied the shortest path most vital edges problem for analyzing the network robustness. The authors proved that the problem is NP-hard and aim to remove the minimum number of edges so as to increase the shortest path length between two nodes in the network. The authors studied a few parameters that affect the computational tractability and underlined the key challenges. Khachiyan et al. \cite{khachiyan2008short} studied short path interdiction problem that involves deleting arcs from a directed graph to eliminate the shortest path from source to destination node. The authors analyzed two subproblems and showed that the short paths node interdiction problem can be solved efficiently; however, the short paths total interdiction problem is an NP-hard problem and is not approximable within certain bounds. Furini et al. \cite{furini2021branch} investigated the edge interdiction clique problem that aims to discover an edge set that minimizes the maximum clique size when removed from the graph. The authors developed an integer linear program formulation and a branch-and-cut algorithm to solve the problem. Nardelli et al. \cite{lin1993most} examined the most vital edge problem. The authors aim to discover a set of edges in a graph, which, when removed, maximizes the total weight of the minimum spanning tree. The authors proved that the problem is NP-hard, and proposed a branch and bound algorithm to address the problem.

\subsection{Stackelberg Game on Attack Graphs}
Milani et al. \cite{milani2020harnessing} studied the applicability of deception in Stackelberg security games using attack graphs. The authors designed an attack graph deception game where the defender can utilize three deceptive actions to modify the attack graph and protect crucial targets from attackers. The authors proved that computing the optimal deception and defensive strategy is NP-hard. The authors developed an approach using a mixed-integer linear program to efficiently solve the problem. Aziz et al. \cite{aziz2018defender} examined a  Stackelberg game on a network, where the defender aims to optimize the inverse geodesic length of network by protecting network components from the attacker, and the attacker seeks to weaken the network. The authors designed several algorithms to determine the defender's optimal policy. Aziz et al. \cite{aziz2017weakening} investigated the problem of removing nodes from the network to optimize the covert network's performance. The authors used inverse geodesic length to compute the network performance and designed various algorithms to address the problem. Durkota et al. \cite{durkota2019hardening} studied the problem of defending networks from strategic attackers with limited resources. The authors modelled the network's interaction between defender and attacker, and employed attack graphs to illustrate the attacker's possible actions. The authors proposed heuristic algorithms to obtain defence strategies against the attacker. \cite{zhang2021bayesian, durkota2015game, durkota2015approximate, letchford2013optimal, durkota2015approximate, anwar2020game} are other studies that explored Stackelberg games on attack graphs.

\section{Structural Hole Spanners}

Structural hole spanner nodes are the bottleneck nodes that connect otherwise disconnected groups of nodes in a network and are crucial for optimizing information flow in the network. Eliminating SHS nodes from the network results in a fragmented network and disrupted information flow \cite{lou2013mining}. Hence, it is essential to discover SHS nodes and enhance the network's resilience. The theory of SH \cite{burt2009structural} was introduced by Burt to discover essential individuals in the organizations and was further explored by \cite{ahuja2000collaboration, burt2007secondhand}. This section presents the current state-of-the-art techniques for discovering SHS nodes in the network and \textcolor{blue}{Table \ref{related_work_table}} presents the summary of SHS identification solutions. Existing work on SHSs identification can be categorized into the following three categories.

\begin{enumerate}
    \item Discovering structural hole spanners in static networks.
    \item Discovering structural hole spanners in diverse networks.
    \item Discovering structural hole spanners in dynamic networks.
\end{enumerate}
 
\subsection{Discovering Structural Hole Spanners in Static Networks} 
\noindent  There are several pioneering works on discovering SHS nodes in static networks. The work can be further categorized into three categories, i.e., information propagation-based solutions, centrality-based solutions and machine learning-based solutions. This section discusses the existing solutions for discovering SHS nodes in static networks.

\subsubsection{Information Propagation-based Solutions}
The solutions based on information propagation aim to discover the SHS nodes whose removal maximally disrupts the information flow in the network. Lou et al. \cite{lou2013mining} proposed an algorithm for discovering SHSs in the network, considering that the community information is given in advance. The authors argued that eliminating SHS nodes from the network decreases the minimal cut of the communities. The authors described the minimal cut as the minimum number of edges that disconnect a community from its connected communities. However, the proposed algorithm only works if the community information is given in advance. He et al. \cite{he2016joint} proposed a Harmonic Modularity (HAM) algorithm that discovers both SHSs and communities in the network. The authors utilized the harmonic function to estimate the smoothness of the community structure and examined the interaction type among the bridging nodes to distinguish SHS nodes from normal ones. The algorithm presumes that every node belongs to only one community, but a node may belong to many communities in the real world. \cite{burt2001closure, burt2011structural} reveal that the profits gained by SHS differ depending on the nature of their ties with the connected communities. Some SHSs may achieve high profits, while others may gain comparatively less. Inspired by this notion, Xu et al. \cite{xu2019identifying} developed a fast solution to discover SHS nodes that connect multiple communities and have strong relations with these communities. The authors argued that removing such SHS nodes blocks maximum information propagation in the network. The authors formulated the top-$k$ SHS problem as a set of $k$ nodes that maximally block the information propagation in the network. Li et al. \cite{li2019distributed} proposed an ESH algorithm for identifying SHS nodes in large-scale networks. ESH utilizes the distributed parallel graph processing frameworks to solve the scalability issue. The algorithm is based on the factor diffusion process, which allows it to efficiently determine structural holes without depending on the substructures in the network. 

\begin{table*}[!ht] 
\caption{Summary of SHSs identification solutions.}
\label{related_work_table}
\renewcommand{\arraystretch}{1}
\centering 
\footnotesize
\resizebox{\textwidth}{!}{\begin{tabular}{lcp{5.5cm}p{4cm}} \hlineB{1.5}  
\textbf{Author} & \textbf{Method} & \textbf{Main idea} & \textbf{Pros \& Cons} \\ \hlineB{1.5}

\multirow{2}{*}{Lou et al. \cite{lou2013mining}} & \specialcell[t]{HIS\\MaxD} & SHS connects opinion leaders of the various communities & Require prior community information \\ \hline
\multirow{2}{*}{He et al. \cite{he2016joint}} & \multirow{2}{*}{HAM} & The authors used harmonic function to identify SHSs & Jointly discover SHSs and communities \\ \hline 
\multirow{2}{*}{Xu et al. \cite{xu2019identifying}} & \specialcell[t]{maxBlock\\maxBlockFast} & SHSs are likely to connect multiple communities and have strong relations with these communities & Less computational cost \\ \hline
\multirow{2}{*}{Li et al. \cite{li2019distributed}} & \multirow{2}{*}{ESH} & The authors designed entropy-based mechanism that uses distributed parallel computing & Less computational cost \\ \hline
 
\multirow{2}{*}{Tang et al. \cite{tang2012inferring}} & \specialcell[t]{2-step algorithm} & The model considers the shortest path of length two that pass through the node & It fails to work in case a node is densely linked to many communities \\ \hline
\multirow{2}{*}{Rezvani et al. \cite{rezvani2015identifying}} & \specialcell[t]{ICC\\BICC\\AP\_BICC} & Eliminating SHSs from the network leads to an increase in the average shortest distance of the network & Only used topological network structure \\\hline
\multirow{2}{*}{Xu et al. \cite{xu2017efficient}} & \specialcell[t]{Greedy\\AP\_Greedy} & The authors used inverse closeness centrality to discover SHSs & Does not require community information\\ \hline
 
\multirow{2}{*}{Ding et al. \cite{ding2016method}} & \multirow{2}{*}{V-Constraint} & The authors used ego-network of the node to discover SHSs & Ego network may not capture the global importance of the node \\ \hline

\multirow{2}{*}{Zhang et al. \cite{zhang2020finding}} & \multirow{2}{*}{FSBCDM} & The author used community forest-based model utility to discover SHSs & Jointly discover SHSs and communities \\\hline
 
\multirow{2}{*}{Gong et al. \cite{gong2019identifying}} & \multirow{2}{*}{Machine learning model} & The authors used various cross-site and ego network features of the nodes & Achieves high accuracy \\\hlineB{1.5}
\end{tabular}}
\end{table*}

\subsubsection{Network Centrality based Solutions} 
The solutions based on network centrality aim to identify the nodes that are located at advantageous positions in the network. Tang et al. \cite{tang2012inferring} designed a two-step technique for identifying SHSs in the network. The authors considered the shortest path of length two for each node while ignoring the others. The drawback of this technique is that it is not able to identify the SHSs when a node is densely connected with two or more communities. The network's mean distance can also be utilized to discover SHS nodes. The shortest path between two nodes that belong to different communities is more likely to pass through the SHS node, and removing this SHS node increases the shortest path between the two nodes. Based on this notion, Rezvani et al. \cite{rezvani2015identifying} first proved that identifying the SHS node problem is NP-hard and then designed various scalable algorithms based on bounded inverse closeness centrality and articulation points in the network. The authors argued that removing SHS nodes from the network increases the shortest distance of network. Inspired by \cite{rezvani2015identifying}, Xu et al. \cite{xu2017efficient} designed solutions to estimate the quality of SHSs based on various properties. The authors proposed efficient algorithms that use filtering techniques to eliminate unlikely solutions. Ding et al. \cite{ding2016method} developed a V-Constraint method to enhance information transmission in the network by determining the SHS nodes. The authors argued that SHS nodes occupy advantageous positions in the network and also control the information diffusion between different communities. The authors used the susceptible-infected-recovery model to perform the experiments. The drawback of the proposed solution is that it only considers the ego network of nodes; therefore, fails to capture the global properties of the nodes. Zhang et al. \cite{zhang2020finding} demonstrated that the local features-based metrics are not sufficient for discovering SHS nodes in the network. The authors designed an algorithm that utilizes the community forest model and diminishing marginal utility to identify SHSs. The authors also emphasized the crucial role SHS nodes play in numerous real-world applications, such as information diffusion \cite{fatemi2022gcnfusion}, community detection \cite{zhang2022information}, epidemic diseases \cite{liu2020optimal} and viral marketing \cite{huang2019community}.

\subsubsection{Machine Learning-based Solutions} 
Machine learning-based solutions utilize the node features to discover SHSs. Gong et al. \cite{gong2019identifying} designed a supervised learning model to discover SHS nodes in online social networks. The authors used the user's profile and content generated by the user to decide if the user is SHS or not. Features such as descriptive, cross-site and ego network features represent each user. However, the solution depends on manually extracted features and may not understand the complex network structures; therefore, the solution may not be applicable to all networks.

\textbf{Challenge:} In the literature, several approaches exist for discovering SHS nodes in static graphs; however, those approaches fail to scale to large graphs. Therefore, there is a need to design approaches for discovering SHS nodes in large-scale networks. This thesis aims to design a graph neural network-based model for efficiently finding SHSs in large-scale networks.

\subsection{Discovering Structural Hole Spanners in Diverse Networks}
Discovering SHS nodes in diverse networks is a new research problem that needs to be addressed. We study different machine-learning approaches \cite{hospedales2021meta, zhuang2020comprehensive, crawshaw2020multi, kouw2019review, ren2021survey,  kaelbling1996reinforcement} and discovered that Meta-Learning techniques can be utilized for discovering SHSs in diverse networks. The meta-learning technique designs machine learning models that learn from the experiences and improve the learning process over time. The objective of meta-learning is to design models that can learn to learn and are able to adapt to new tasks very quickly. This section discusses several studies \cite{zhou2019meta, ryu2020metaperturb, frans2017meta, wang2020graph, sankar2019meta, qu2020few, zhang2022hg, suo2020tadanet, bose2019meta} that uses Meta-Learning \cite{vilalta2002perspective, vanschoren2018meta, finn2017model, nichol2018first, finn2019online, hospedales2021meta} to solve similar research problems. 

Zhou et al. \cite{zhou2019meta} developed Meta-GNN framework to solve the few-shot node classification problem. The authors trained one classifier on numerous similar few-shot learning tasks to gain prior knowledge. The acquired knowledge is then utilized to classify the nodes from new classes, given only a limited number of labelled instances. The authors claimed that the proposed framework can be incorporated with any graph neural network model, making it a universal framework for solving similar problems. Ryu et al. \cite{ryu2020metaperturb} proposed a technique to solve an essential problem in machine learning, i.e., to improve the model generalization capability on unseen data. The authors developed a meta-learning model that trains the perturbation function in parallel over many heterogeneous tasks to improve the model's generalization power across different tasks. Frans et al. \cite {frans2017meta}  proposed a meta-learning approach for leveraging the shared primitives to learn hierarchically structured policies in order to enhance the sampling efficiency on new tasks. Wang et al. \cite{wang2020graph}  proposed a meta-learning based model to address the problem of few-shot learning in an attributed network. The authors considered the distinctive characteristics of attributed networks that led to the model's exceptional performance in the meta-testing stage. Sankar et al. \cite{sankar2019meta} designed a semi-supervised learning technique for attributed heterogeneous networks. Despite limited supervision, the method facilitates node classification based on network structure and type of nodes.

\textbf{Challenge:} Discovering SHS nodes in diverse networks is a novel research problem that needs to be addressed. Traditional one-model-fit-all Machine Learning-based models fail to capture the inter-graph differences and may perform poorly for diverse graphs. Therefore, we need machine learning models that are aware of differences across graphs and can customize accordingly without the need to be retrained on every type of network. This thesis designs a meta-learning-based graph neural network model to effectively discover SHSs across diverse networks.

\subsection{Discovering Structural Hole Spanners in Dynamic Networks} 
\noindent Numerous solutions have been proposed for discovering SHSs in the steady-state behaviour of the network. Nevertheless, real-world networks are not static; they evolve continuously. Currently, there is no solution that discovers SHS nodes in dynamic networks. This section presents various proposed solutions that discover bottleneck nodes, such as influential nodes, blockers and critical nodes in dynamic networks. 

Several studies  \cite{chen2015influential, yang2017tracking, zhao2019tracking, song2016influential, yang2018tracking, ohsaka2016dynamic, yang2019tracking, basaras2013detecting, aggarwal2012influential} have investigated the problem of discovering influential nodes in dynamic networks. Chen et al. \cite{chen2015influential} studied the problem of tracking influential nodes in dynamic networks. The authors designed an upper-bound interchange greedy algorithm that exploits the network structure evolution smoothness to solve the problem. The algorithm utilizes the previously identified influential nodes and then executes node replacement in order to enhance influence coverage. Besides, the authors proposed an approach that updates the nodes quickly to retain an upper bound on the node replacing gain. Yang et al. \cite{yang2017tracking} designed an algorithm to track the influential nodes in dynamic networks. The algorithm works by updating the nodes incrementally as the network evolves. The algorithm is able to work under various scenarios, including insertions and deletions. Zhao et al. \cite{zhao2019tracking} designed an algorithm to identify the influential nodes in dynamic networks. The algorithm capture the interaction between nodes and discard the outdated interactions. Song et al. \cite{song2016influential} investigated the problem of tracking influential nodes as the network evolves and proposed an algorithm that updates the influential nodes using the influential nodes from the previous network. \cite{yen2013efficient, lerman2010centrality, mantzaris2013dynamic, kim2012temporal, bergamini2016approximating, bergamini2014approximating,  ghanem2018centrality} studied solutions for updating the centrality measures in dynamic networks. Yen et al. \cite{yen2013efficient}  proposed an algorithm named CENDY to efficiently update the closeness centrality and average path length in dynamic networks. The algorithm works for incremental and decremental edge updates in the network. The algorithm first discovers the nodes for which the shortest path changes due to network updates and then performs centrality updates for affected nodes only. Bergamini et al. \cite{bergamini2016approximating} designed an algorithm for calculating the betweenness centrality in dynamic networks. The proposed algorithm is able to perform in-memory computations of the centrality measure in dynamic networks with millions of edges. Lerman et al. \cite{lerman2010centrality} designed a centrality metric that considers the network's temporal dynamics. The metric computes the node's centrality as the number of paths connecting it to the other nodes. Bergamini et al. \cite{bergamini2014approximating} designed incremental algorithms for calculating the betweenness centrality in dynamic networks and provided a provable guarantee on absolute approximation error. The algorithms are able to achieve significant speedup due to the efficient update mechanism of shortest paths. Song et al. \cite{song2015mining} developed various heuristic solutions to identify broker nodes in dynamic networks. The author defined broker nodes as the set of nodes that occupy critical positions and control information flow in the network. The authors first proved that the problem is NP-hard and then proposed incremental algorithms based on the weak tie theory. Yu et al. \cite{yu2010finding} studied the problem of determining good blocker nodes in the network, which can hamper the spread of dynamic process. The authors evaluated the structural measures to determine the most effective blockers in static and dynamic networks. The authors demonstrated that the simple local measures could accurately predict an individual's capability to block the spread of process, regardless of network's dynamic nature.

\textbf{Challenge:} The SHS node identification problem in a static network is a well-studied problem, but no solution exists for discovering SHS nodes in dynamic networks. Traditional algorithms are time-consuming and may not perform well for dynamic networks. Therefore, this thesis designs efficient solutions for discovering SHS nodes in dynamic networks.


\chapter{Defending Active Directory by Combining Neural Network based Dynamic Program and Evolutionary Diversity Optimization} 

\label{Chapter_AD_NNDP} 

\textbf{\underline{Related publication:}} 
\vspace{0.06in}

\noindent This chapter is based on our paper titled “\textit{Defending Active Directory by Combining Neural Network based Dynamic Program and Evolutionary Diversity Optimization}” published in The Genetic and Evolutionary Computation Conference (GECCO), 2022 \cite{Goel2022defending}.

\vspace{0.1in}

\noindent Microsoft domain network comprises significant market shares among small as well as big organizations globally, due to which active directories have become an attractive target for the attackers. If an attacker gains unauthorized access to the most privileged account in AD, it can cause considerable financial and reputational damage to the organization. Therefore, defending AD is crucial to protect the organization’s assets and ensure its smooth functioning. This chapter presents an effective defensive technique for determining the bottleneck edges which can be blocked to prevent an attacker from reaching the most privileged account in AD. We study a Stackelberg game model between one attacker and one defender on an AD attack graph. The attacker aims to maximize their chances of successfully reaching the destination before getting detected, and the defender’s goal is to block a constant number of edges to minimize the attacker’s chances of success. We show that the problem is \#P-hard and, therefore, intractable to solve exactly. We train a Neural Network to approximate the attacker’s problem and propose Evolutionary Diversity Optimization based policy to solve the defender’s problem. We evaluate the performance of our proposed defensive policy on synthetic AD attack graphs, and our results demonstrate that the proposed defensive policy is highly effective.

\section{Introduction}\label{def_intro}
Cyber attackers utilize attack graphs to identify the possible ways to gain unauthorized access to the systems. Industrial practitioners are actively utilizing the AD attack graph, which is an attack graph model. Microsoft \textit{Active Directory} is a default security management system for Windows domain networks \cite{dias2002guide}.
Microsoft domain network constitutes significant market shares among organizations globally, due to which AD are promising targets for cyber attackers. \textsc{BloodHound} is the most popular tool to analyze the AD attack graphs. The \textsc{BloodHound} has significantly eased the process of attacking AD for potential cyber attackers. Due to the popularity of AD attack graphs, defenders also study AD graphs to devise defensive strategies. Dunagan et al. \cite{dunagan2009heat} developed a heuristic solution for blocking a few edges to partition the attack graph into disconnected components. The resulting disconnected components disable attackers from reaching the most privileged account, called Domain Admin, only if the attacker’s entry nodes and DA are in different components. Edge blocking in an AD environment can be attained by either monitoring the edges or revoking access. Evolutionary algorithms have traditionally been used to solve various attacker-defender problems \cite{zhang2019investigating, hu2020optimal, hemberg2018adversarial}. Ulrich et al. \cite{ulrich2010integrating, ulrich2011maximizing} studied Evolutionary Diversity Optimization (EDO) that aims to find a set of diverse solutions. EDO has gained considerable attention in the Evolutionary Computation community. A new solution deviating from its predecessors leads to less competitiveness and high evolvability \cite{lehman2013evolvability}. \textit{\textbf{In this chapter, the defender aims to block a set of edges to minimize the strategic attacker’s probability of reaching DA.}} To solve the defender’s problem, we consider the blocking plans of the defender as a population and propose EDO based defensive approach to generate a diverse set of blocking plans. To address the defender's problem, we treat the defender's blocking strategies as a population and introduce an Evolutionary Differential Optimization (EDO) based defensive approach. Our solution aims to generate a wide range of blocking plans, leveraging EDO's diversity capabilities. The diversity in blocking plans, in turn, improves the accuracy of modelling the attacker's problem, ultimately leading to the discovery of more effective defensive plans.

In the literature, Guo et al. \cite{guo2021practical} proposed several algorithms to address this problem. However, the authors have considered a scenario where once the attacker is detected, the attack ends, and when an attacker chooses a path, the attacker continues to move on to that path until the attacker gets caught or reaches DA. The attacker’s problem in \cite{guo2021practical} can be solved using Dijkstra’s algorithm, but in our model, the attacker’s problem is computationally hard. In our model, we assume that every edge $e$ has a different detection probability $p_{d(e)}$, i.e., if the attacker is detected, the whole attack ends and a failure probability $p_{f(e)}$, i.e., if the attacker fails and is not able to pass through an edge; the attacker can try another edge from the same node or different node. For example, an edge requires cracking a password to pass through it; if the password is weak, then the attacker can crack it. In this case, the probability of having a strong password is the failure probability. With every edge’s $p_{d(e)}$ and $p_{f(e)}$, the attacker can successfully go through an edge $e$ with a success probability $p_{s(e)}=(1-p_{d(e)} - p_{f(e)})$. 

Therefore, in our model, the attacker plays strategically in the sense that the attacker first starts an attack, and if, at some point, fails (instead of being detected), the attacker can try again until the attacker is detected or has tried all possible options. The attacker aims to design an attacking policy that maximizes their probability of reaching DA. Initially, the attacker has access only to a set of entry nodes, from where the attacker can start. While trying to reach DA, the attacker expands the set of accessible nodes by exploring edges and saves this information, i.e., set of successful edges the attacker has currently control of, set of edges that the attacker has lost (failed but not detected, for not being able to guess the password, etc.). All the strategies that attacker has tried are an “\textit{investment}” that the attacker can use later. In this way, the attacker has \textit{“secured”} a set of nodes at any point and can attempt an unattempted edge originating from any of the secured nodes or entry nodes, which we combinedly call checkpoints, \textit{$\text{Checkpoints} = \{\text{Entry nodes} \cup \text{Secured nodes}\}$}. Generally, the attacker prefers the attack paths with low detection and failure probability, and covers the valuable checkpoints along its way. The defender wants to deterministically block $k$ \textit{block-worthy edges}, where $k$ is the defender’s budget, in turn increasing the corresponding edge’s failure rate from original $p_f$ to 100\%. We follow a standard \textit{Stackelberg game model} \cite{yin2010stackelberg}, where 
the defender devises a strategy, and the attacker observes the defender’s strategy and plays his best to develop an attack strategy on the target. In practice, the attacker can scan the AD environment using \textsc{SharpHound}\footnote{https://github.com/BloodHoundAD/SharpHound} tool and get information about which edges are blocked. 

We aim to propose an effective approach for computing defender strategy (to block a set of edges that minimizes strategic attackers’ probability of reaching DA) that scales to large AD attack graphs. We have proved that the attacker’s problem of deriving an optimal attacking policy is $\#P$-hard. We proved that the defender’s problem is also $\#P$-hard, even if the blocking budget is one. Therefore, the problem is intractable to solve exactly with the current methods, so we propose an approach that involves training a \textbf{\textit{Neural Network (NN)}} to approximate the attacker’s problem and \textbf{\textit{Evolutionary Diversity Optimization}} to solve the defender’s problem. We can describe the attacker’s problem as Markov Decision Process (MDP), which can be solved using Dynamic Program (DP) \cite{bellman1966dynamic}. The size of state space is $3^{\text{|Edges|}}$, where an edge is either unattempted, attempted and failed, or attempted and successful. However, this state space is too large considering that practical AD graphs may have tens of thousands of edges. We use a fixed-parameter analysis technique that focuses on determining easy-to-solve instances of a problem. We design a kernelization technique so as to reduce the AD attack graphs to a much smaller \textit{condensed graph}. We then convert the attacker’s problem from a condensed graph to a DP. Using our kernelization technique, the state space becomes Fixed Parameter Tractable with respect to a parameter called the number of \textit{Non-Splitting Paths (NSP)}. Guo et al. \cite{guo2021practical, guo2022scalable} proposed NSP idea to describe tree-likeness of AD graphs. NSP is a path on which every node has only one successor except the last node. We can solve the DP directly for small AD graphs. For larger AD graphs, our FPT special parameter $\#NSP$ is practically too large, so we use NN to approximate the DP \cite{yang2018boosting}. Our main idea is to train NN to learn both the base cases and the DP recursive relationships. Considering that the state space is exponential and we do not have resources to train NN to learn the value of every state; however, not all states are useful. It is important to learn the values of the states that are referenced by the optimal decision path; therefore, we only consider the important states in the state space, which reduces the state space size to a large extent. With the strong flexibility power of NN, we aim to train the NN to approximate the attacker’s policy. NN serves as a fitness function for EDO (the exact fitness function is $\#$P-hard to compute). EDO provides a diverse set of blocking plans, i.e., diverse set of training samples for training NN, which prevents NN from getting stuck in local optimum. The blocking plans are given as input to the NN, and it outputs the attacker’s probability of reaching DA corresponding to the blocking plan. We go back and forth between the processes, generating blocking plans using EDO and training NN on blocking plans, to get a well trained NN that can act as an efficient fitness function for EDO. In this way, EDO and NN help each other in order to improve.

\noindent \textbf{\textit{Contributions:}} In this chapter, we make the following contributions.
\begin{itemize}
    \item \textbf{\#P-hard proof.} We first prove that the attacker’s problem of deriving an optimal attacking policy is $\#P$-hard, and the defender’s problem is also  $\#P$-hard. 
    \item \textbf{Kernelization technique.} We design a kernelization technique that converts the attacker’s problem to a dynamic program where the number of states is fixed-parameter tractable with respect to the number of non-splitting paths. 
    \item \textbf{Attacker-defender policy.} We train a Neural Network for approximating the attacker’s problem and propose Evolutionary Diversity Optimization to solve the defender’s edge blocking problem.
    \item \textbf{Extensive experiments.} Our experimental results on the synthetic R500 AD attack graph show that the proposed approach (attacker’s policy and defender’s policy) is highly effective, and is less than 1\%  away from the optimal solution.
\end{itemize}

\noindent \textbf{\textit{Chapter organization}:} \textcolor{blue}{Section \ref{model_description}} describes the model description. \textcolor{blue}{Section \ref{proposed}} discusses the proposed methodology in detail, including kernelization technique, attacker’s policy and defender’s policy. \textcolor{blue}{Section \ref{experiments}}  reports the experimental results, and finally, \textcolor{blue}{Section \ref{conclusion}} concludes the chapter.

\section{Background}
\noindent \textbf{\textit{Dynamic Programming (DP)}.} 
DP is a widely used technique for designing efficient optimisation algorithms. For a problem that can be reduced to a subproblem with similar structures, each corresponding to a decision-making stage, DP first finds the optimal solution for each subproblem and then finds the optimal global solution. One significant characteristic of DP is that the solution of a subproblem is frequently used numerous times for solving various large subproblems.

\noindent \textbf{\textit{Neural Networks (NN).}} NN are the foundation for various artificial intelligence applications such as speech recognition, image recognition, self-driving cars, detecting cancer etc. Currently, NNs are outperforming human accuracy on a range of tasks. The reason behind the exceptional performance of NNs is its ability to extract high-level features from data and obtain an effective representation of input space. Theoretically, NNs with a large number of parameters can fit any complex function. NN have enabled machine learning techniques to attain high accuracy in less time for various problems.

\noindent \textbf{\textit{Evolutionary Diversity Optimisation (EDO).}}  EDO aims to find a diverse, high-quality set of solutions that are maximally different. This research area was first studied by Ulrich et al. \cite{ulrich2010integrating} and has gained considerable attention in the community of evolutionary computation. A new solution deviating from its predecessors leads to less competitiveness and high evolvability. In addition, numerous interesting solutions are considered more valuable than a single good solution; therefore, EDO is a beneficial addition to conventional optimisation. Various applications that use a diverse set of solutions are air traffic control, personalised agent planning, etc.

\section{Problem Description}\label{model_description}
AD attack graph is a directed graph $G=(V,E)$, with $n=|V|$ nodes and $m=|E|$ edges. There are $s$ entry nodes, from where the attacker can enter the graph and one destination node DA (Domain Admin). AD attack graphs may have multiple admin nodes, but we merge all admin nodes into one node and call it DA\footnote{We merge all the DA into a single node to reduce the problem complexity.}. The attacker starts from any entry node and aims to devise a policy that maximizes their probability of reaching DA. The attacker initially has access only to a set of entry nodes and tries to expand this set by exploring more edges. Every edge $e \in E$ has a detection probability $p_{d(e)}$ that ends the attack, and a failure probability $p_{f(e)}$, which does not end the attack; on encountering a failed edge, the attacker can continue the attack by exploring one of the unexplored edges. The attacker can successfully pass through an edge with a probability of $(1-p_{d(e)} - p_{f(e)})$. In our model, the attacker plays strategically; the attacker initiates an attack and continues the attack by exploring unexplored edges until the attacker is detected, has explored all possible options or reached DA. The defender blocks $k$ block-worthy edges, where $k$ is the defender’s budget and aims to minimize the strategic attacker’s chances of reaching DA. Only a set of edges are blockable. We assume that the attacker can observe the defensive action and accordingly comes up with the best-response attacking policy. We first show that the attacker’s and defender’s problems are $\#P$-hard to calculate. Therefore, the problems are intractable to solve with the existing approaches. We propose an approach that trains a Neural Network to approximate the attacker’s problem and Evolutionary Diversity Optimization to solve the defender’s problem.

This chapter proposes Evolutionary Diversity Optimization (EDO) as a defender's policy that generates diverse plans to block a set of edges, minimizing the attacker's chances of successfully reaching the DA. The NN serves as the fitness function for EDO, evaluating the attacker's probability of reaching the DA based on different blocking plans. The diverse blocking plans generated using EDO are used as training data for the NN, and they prevent the NN from getting stuck in local optima and enhance the NN's accuracy in modelling attacker behaviour.

\section{Hardness Results}
Let $G$ be a graph with one entry node ENTRY and one destination node DA as shown in \textcolor{blue}{Figure \ref{ch3_hard1}(a)}. For both the edges ENTRY $\rightarrow$ 1 and ENTRY $\rightarrow$ 2 in $G$, $p_d = 0.1$ ($p_d$ can be any arbitrary value greater than 0) and $p_f = 0$. Our hardness results are discussed below:

\begin{figure}[h!]
 \centering
 \includegraphics[width=0.55\paperwidth,height=4.5cm]{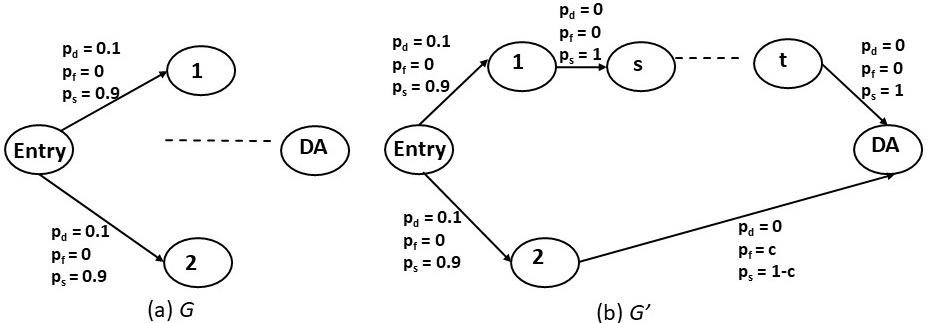}
 \caption[Graph instances to prove $\#$P-hardness of problem.]{$G'$ is constructed from $G$.}
 \label{ch3_hard1}
\end{figure}

\begin{theorem}
\label{chapter3_t_1}
\normalfont The attacker’s value is $\#$P-hard to calculate.
\end{theorem}
\begin{proof} The attacker’s value is the attacker’s probability for reaching DA before getting detected under the attacker’s optimal attacking policy, facing a given defense (\textcolor{blue}{Figure \ref{ch3_hard1}(a)}). A known $\#P$-complete result for the $s-t$ connectedness problem  \cite{valiant1979complexity} states that, given nodes $s$ and $t$ in a directed graph, if every edge’s probability of existence is 0.5, then it is $\#$P-complete to calculate the probability that s and t are connected. We can apply this result to our problem by assuming an edge’s $p_f =0.5$ and $p_d =0$ for all edges in the attack graph. Therefore, the attacker’s value is $\#$P-hard to compute.
\end{proof}

\begin{theorem}
\label{Chapter3_t_2}
\normalfont The attacker’s optimal policy is $\#$P-hard to calculate.
\end{theorem}
\begin{proof} We construct a graph $G'$ from $G$ as shown in \textcolor{blue}{Figure \ref{ch3_hard1}(b)}. Let $G'$ be a directed graph with a source node $s$ and destination node $t$ . All the edges in graph $G'$ have $p_d=0$ and $p_f = 0.5$. We add two edges 1 $\rightarrow s$ and t $\rightarrow DA$ to $G'$, each with $p_d = p_f=0$. In addition, we add another edge 2 $\rightarrow DA$ with $p_d=0$, $p_f =c$ where $c$ is any constant. 
In \textcolor{blue}{Figure \ref{ch3_hard1}(b)}, the attacker needs to determine whether the attacker should go up (Entry$\rightarrow$1) or down first (Entry$\rightarrow$2). Let us assume that the total failure rate for going from $s$ to $t$ is $p_{up}$, and the failure rate of going from 2 to DA is $p_{dn}$. If the attacker go up (Entry$\rightarrow$1) first then,
\begin{equation*}
    \text{Attacker’s value} = 0.9(p_{up} + (1-p_{up}) 0.9p_{dn}).
\end{equation*}
\begin{equation*}
    \text{Attacker’s value}  = 0.9p_{up} + 0.81p_{dn} - 0.81p_{up}\times p_{dn}  .
\end{equation*}
If the attacker go down first (Entry$\rightarrow$2) then,
\begin{equation*}
    \text{Attacker’s value} = 0.9(p_{dn} + (1-p_{dn}) 0.9 p_{up}).
\end{equation*}
\begin{equation*}
    \text{Attacker’s value}  = 0.9p_{dn} + 0.81p_{up} - 0.81p_{up}\times p_{dn}.
\end{equation*}

Therefore, the attacker needs to determine higher failure rate from two, $p_{up}$ or $p_{dn}$. We call this problem an up-down deciding problem, and the attacker’s optimal policy is at least as difficult as the up-down deciding problem. Let’s assume we have an oracle that solves the up-down deciding problem. We can compute the exact value of $p_{up}$ by calling the oracle polynomial number of times (Value of $p_{up}$ is from a finite set, $\{\frac{0}{2^m}, \frac{1}{2^m}, ... \frac{2^m}{2^m}\}$, where $m$ is the number of edges in the directed graph containing $s$ and $t$). We can perform binary search on $p_{up}$ using our up-down oracle, i.e., every time we construct one attack graph, we can reduce the range of possible values for $p_{up}$ by a factor of 2, and we need to call the up-down oracle only $m$ times. This is Turing reduction as we use the up-down oracle polynomial number of times to get the exact value of $p_{up}$, and construction of $m$ attack graphs is also polynomial. By Turing reduction, the up-down deciding problem behind the oracle is at least as hard as calculating $p_{up}$; however, $p_{up}$ is $\#$P-hard to compute (\textcolor{blue}{Theorem \ref{chapter3_t_1}}), therefore, the up-down deciding problem is $\#$P-hard. Since the attacker’s optimal policy is at least as hard as the up-down problem, therefore, attacker’s optimal policy is also $\#$P-hard to calculate.
\end{proof}

\begin{Corollary}
\label{Chapter3_t_3}
\normalfont The defender’s value and policy are both $\#$P-hard to calculate.
\end{Corollary}
\noindent \begin{proof} 
The defender’s value is the attacker’s chance of reaching DA (assuming the attacker plays the optimal policy) under the best defense. In order to compute the defender’s value, we have to solve the up-down deciding problem, which is $\#$P-hard. Let us assume that the defender’s budget is one for defender’s policy. Now, the defender has 2 paths to block, 1$\rightarrow$DA or 2$\rightarrow$DA, and assuming only two edges t$\rightarrow$DA and 2$\rightarrow$DA are blockable. Therefore, the defender’s policy is also at least as hard as the up-down deciding problem, which is $\#$P-hard and hence, the defender’s is also $\#$P-hard to calculate.
\end{proof}

\section{Proposed Approach}\label{proposed}
This section first discusses the proposed kernelization technique that reduces the original graph to a condensed graph and converts the attacker’s problem to a dynamic program. We then discuss the proposed Neural Network that approximates the attacker’s problem. Lastly, we discuss the proposed Evolutionary Diversity Optimization to design the defensive policy.

\subsection{Fixed-Parameter Tractable Kernelization}
\noindent \textit{Kernelization} is a common technique for fixed-parameter analysis that processes a problem and reduces it to a smaller equivalent problem, called \textit{kernel}. 
We can consider an attack graph as a tree with $h$ extra edges, known as \textit{feedback edges (h)} and $h = m-(n-1)$. \textit{Splitting nodes (t)} are the nodes with more than one outgoing edge. DA does not have any successor in the graph, and even if DA have any successors, that can be ignored as once the attacker reaches DA, the attack ends. When $t = 0$, the attack graph is exactly a tree. Also,  $t\leq h$. We denote the \textit{Number of entry nodes} using $s$. For kernelization, we consider one parameter of AD attack graphs, which is the number of Non-Splitting Paths (paths from entry nodes, or other splitting nodes).
\begin{definition}
\label{def_NSP}
\noindent \textbf{Non-Splitting Path (NSP).}  Non-Splitting Path NSP$(i, j)$ is a path that goes from node $i$ to $j$, where $j$ is $ i’s$ successor and then continually moves to $ j’s$ sole successor, until we reach DA or another splitting node \cite{guo2021practical, guo2022scalable}.
\end{definition}
\begin{equation*}
NSP = \{NSP(i,j)|\, i \in \text{SPLIT} \cup \text{ENTRY}\},
\end{equation*}

where SPLIT represents the set of splitting nodes and ENTRY represents the set of entry nodes. A NSP is \textit{blockable} only if at least one of its edge is blockable. In addition, once the attacker chooses to move onto a NSP, it is without the loss of generality to assume that under the attacker’s optimal policy, the attacker has to complete the NSP until the attacker 1) gets caught; 2) fails; 3) reaches DA; or 4) reaches any splitting node. Otherwise, the attacker risks getting detected without securing any new checkpoints (splitting nodes).

\begin{definition}
\noindent \textbf{Block-worthy (BW).} A block-worthy $bw(i, j)$ is any furthermost blockable edge on path NSP$(i,j)$. Two NSPs may share the same block-worthy edge. 
\end{definition}
\begin{equation*}
BW = \{bw(i,j)|\, i \in \text{SPLIT} \cup \text{ENTRY}, j \in \text{Successor$(i)$}\}.
\end{equation*}

Notably, for the edge blocking strategy, we only need to spend one unit of budget on $NSP(i,j)$; otherwise, we could simply block $bw(i,j)$ to eliminate this NSP from the attacker’s consideration. 
Size of block-worthy edge set, $|BW|$ can be bounded as:
\begin{align*}
\begin{split}
    |BW|& \leq s+t+h\\
    |BW|& \leq s + 2h,\,\,\,\,\,\,\,\,\, \text{since}\,\, t \leq h
\end{split}
\end{align*}

\begin{figure}[t!]
\centering
 \includegraphics[width=0.4\paperwidth]{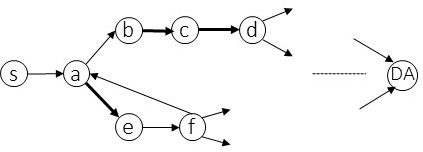}
 \caption{Example of attack graph.}
 \label{fig:attack_graph_nodes}
\end{figure}
\textcolor{blue}{Figure \ref{fig:attack_graph_nodes}} illustrates an attack graph, where the entry node is $s$, destination node is $DA$, split nodes are $\{a, d, f\}$,  non-splitting paths are $\{(s, a), (a, b, c, d), (a, e, f)\}$. The thick edges are blockable, so blockable edges are $\{(b, c), (c, d), (a, e)\}$ and block-worthy set is $\{(c, d), (a, e)\}$. {In the {original AD attack graph}, there are $n$ nodes and $m$ edges. The kerneization technique converts the original graph into a {condensed AD attack graph} with only (|ENTRY|+|SPLIT| + 1) nodes and |NSP| edges.}

\noindent \textbf{Converting Attacker’s Problem to Dynamic Program.} 
\noindent We describe attacker’s problem of devising an attacking policy that maximizes the probability of reaching DA as Markov Decision Process, where the state $s$ is a vector of size $|NSP|$ and can be represented as:

\begin{equation}
\label{attacker_state_vector}
\underbrace{< S, F, ?, ?, ?, S, F, ?, ?, ?, ?, ?, S, F, F >}_\text{Length of state vector = Number of NSP} 
\end{equation}
where:
\begin{description}
\item `S’ represents that the attacker has tried NSP and it is successful (the attacker has reached at the end of NSP)
\item `?’ represents that the attacker has not yet attempted the NSP
\item `F’ represents that attacker has tried NSP and failed (not detected) 
\end{description}

Given a state vector as shown in \textcolor{blue}{Equation \ref{attacker_state_vector}}, the attacker tries one of the NSP (action for a state) with status `?’ (not attempted) in order to reach DA. The realization of the NSP that the attacker tries can turn out to be either successful, fail or detected. Accordingly, the status of that NSP changes to `S’ or `F’ and the attacker gets a new state. However, if the attacker is detected, the attack ends. All the NSPs that the attacker has tried and are successful act as checkpoints for the attacker. The attacker can explore an unexplored edge from these checkpoints in the future. In this way, the attacker tries unexplored NSPs until the attacker reaches DA or is detected.

\noindent \textbf{State Transition.} For a state and action (actions for a state s are the unexplored NSPs outgoing from the end node of successful NSPs in state s), we may have a distribution of future states. We have described the state transition process using \textcolor{blue}{Figure} \ref{fig:state}. For simplicity, assume that every edge in \textcolor{blue}{Figure \ref{fig:state}} has $p_d$ = 0.1 and $p_f$ = 0.2. We have two NSPs, $<$NSP(ABCD), NSP(ECD)$>$ visible in our figure, so we focus on these two NSPs for presentation purposes. The initial state vector is $\,\,<?, ?>$. We can try one of the NSPs. Let us try NSP(ABCD), and we go through each edge on this NSP. 
\begin{figure}[h!]
 \centering
 \includegraphics[width=0.3\paperwidth]{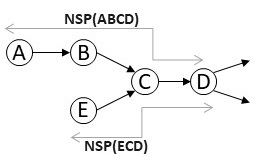}
 \caption{Example of state transition process.}
 \label{fig:state}
\end{figure}
If AB fails: the status of its NSP changes to `F’, and the new state is <F, ?>. This state is added to the future state set with 0.2 transition probability. There is 0.1 probability the attack ends and 0.7 probability of successfully passing through edge. If edge AB succeeds but BC fails: <F,?> state is already in set; therefore, its transition probability is updated to (0.7$\times$0.2+0.2)=0.34. There is (0.7$\times$0.1)=0.07 probability that the attack ends and the probability of successfully passing through edge is (0.7$\times$0.7)=0.49. If edge AB succeeds, BC succeeds, but CD fails: CD is a common block-worthy edge between NSP(ABCD) and NSP(ECD). If CD fails, both the NSPs sharing CD fail, and the new state is <F, F>. The state is added to the future state set with (0.7$\times$0.7$\times$0.2)= 0.098 transition probability. There is (0.7$\times$0.7$\times$0.1)=0.049 probability that the attack ends, and the probability of successfully passing through edge is (0.7$\times$0.7$\times$0.7)=0.343. If AB succeeds, BC succeeds, CD succeeds: All the edges on NSP(ABCD) have been explored, and we can successfully go through this NSP. Therefore, we have a new state <S, ?>. This state is added to the future state set with 0.343 transition probability. On trying NSP(ABCD), we get three future states and their corresponding transition probabilities, $\{(<F, ? >,\;0.34),(<F, F >,\;0.098), (<S, ? >,\;0.343)\}$.

\textit{{State Transition  Rules.}} Two or more NSPs may share the same block-worthy edge, making the state transition process complicated. In \textcolor{blue}{Figure \ref{fig:state}}, there is one splitting node $D$ and two NSPs, i.e., $ABCD$ and $ECD$. However, the two NSPs are not independent as they share the same block-worthy edge $CD$. We design the following rules for the state transition:
\begin{enumerate}
    \item If $NSP(ABCD)$ is successful, then we change the status of all unexplored NSPs that ends at node $D$ to successful. Since we already got access to node $D$, we do not have to try these NSP(ECD). We can directly change the status of $NSP(ECD)$ to `S’. 
    \item If all the NSPs that ends on a split node $D$ fails, then the status of all the NSPs from the split node $D$ is changed to `F’.
    \item If a $NSP(ABCD)$ fails, then it may or may not affect the status of other NSPs ending or originating from node $D$. It depends on which edge on the $NSP(ABCD)$ fails. The failure on $NSP(ABCD)$ may be due to any of the edge $AB$, $BC$ or $CD$. If edge $AB$ or $BC$ fails, then it will not affect the status of $NSP(ECD)$, and we can try $NSP(ECD)$ later; however, if the common block-worthy edge $CD$ fails, then the status of all the NSPs sharing this block-worthy edge is changed to `F’. Therefore, we need to determine where exactly the failure occurs and if this failure will affect the other NSPs or not.
\end{enumerate}

During state transition, for each new state $s$, we first check if it is already a determined state or not. Determined states are the states which are already present in future state set and for each determined state, we have identified set of actions that can be performed on that state. If the state is not determined yet, we compute \textbf{\textit{admissible set of actions}} $A(s)$, for this state. Admissible set of actions are the actions available for the state, i.e., set of unexplored NSPs that can be explored from the checkpoints of this state. Notably, admissible actions only include unexplored NSPs and ignore those NSPs for which we do not have access to their source node. The maximum size of the attacker’s state space can be $3^{|NSP|}$, which is very large; however, not all state vectors are possible or relevant for the attacker. We only consider the state vectors that are relevant for the attacker by following the state transition process; given an initial state and admissible set of action, we determine the future states that an attacker can encounter on the way to DA, and we call these states as important states. We can solve the attacker’s problem using the DP technique. 
For a given state $s$ and action $a$ from \textit{admissible set of actions} $a \in A(s)$, the attacker problem of maximizing the probability of reaching DA can be written as:
\begin{equation}
\label{DP}
V(s) = \max_{\substack{a\in A(s)}} \,\bigg(\sum_{\substack{s'}}{Pr(s'\,|\,s,a)\, V(s')}\bigg),
\end{equation}
where $V(s)$ denotes the value function for the current state $s$ i.e., probability of reaching DA when the attacker is in state $s$, $s'$ denotes the distribution of future states that follows after choosing action $a$ and $Pr(s’ \,|\,s,a)$ denotes the state transition  probability of $s'$, when an action $a$ is performed on $s$. \textcolor{blue}{Equation \ref{DP}} can be solved by computing the value functions for smaller problems and the overall value function step by step. 
Nevertheless, backward induction can be computationally challenging for large state spaces. Therefore, we use NNs to approximate the dynamic programming function.

\subsection{Attacker Policy: Neural Network based Dynamic Program}
We use NN to approximate the attacker’s problem, and it acts as a fitness function for the Evolutionary Diversity Optimization. Given a blocking plan, NN outputs a value that indicates the attacker’s chances of reaching DA. For approximating the attacker’s value function, we first supervise NN to learn the DP base states. Base states are the states in which the status of NSPs that end at DA is either `S’ or `F’. NSPs that ends at DA with status `S’ are always 100\% successful, and therefore, the value of these states is 1; the attacker will reach DA. If all the NSPs leading to DA have status `F’, the attack fails 100\%; therefore, the value for these states is 0. Once the attacker reaches DA, there is no state transition as the attack ends immediately. For other states, we train the NN to learn the recursive relationship of \textcolor{blue}{Equation \ref{DP}}. 
We select an initial state vector (actual state vector corresponding to a blocking plan; the state where some coordinates are `F’, which are the NSPs blocked by the blocking plan and rest of the coordinates are `?’) as shown in \textcolor{blue}{Equation \ref{attacker_state_vector}} and generate a batch of future states to train NN. Given an initial state, NN makes an optimal decision (according to the NN model) with $0.5$ probability or makes a random decision with $0.5$ probability to explore other possibilities. It is possible that the action performed by NN is not optimal; therefore, we employ randomness to explore other actions as well. After performing an action, we may have a distribution of future states and their transition probabilities. We select one of the future states weighted by its probability. Similarly, we keep moving to other states until we reach the base state. We train the NN to learn the recursive relationship between states and minimize the mean squared error (MSE) of estimated results of the value function.
Let $V(s;\theta)$ be the value predicted by the NN, where $s$ represents the input state vector, and $\theta$ be the model parameter. Ideally, $V(s;\theta)$ is exactly a DP function. The loss is computed as follow:
\begin{equation}
\label{loss}
Loss = \sum_{s \in S} \bigg(V(s;\theta) - \max_{\substack{a\in A(s)}} \,\Big(\sum_{\substack{s'}}{Pr(s'\,|\,s,a)\, V(s')}\Big)\bigg)^2,
\end{equation}
where $S$ is the set of all states. It is impractical to compute the gradient for all states in one iteration; therefore, we adopt a common approach of performing gradient descent on a mini-batch of data to train the NN. As the NN value function gets better, there are higher chances that the generated states are optimal or near-optimal, which indicates that the attacker may go via these states to DA. In addition, given a deep enough neural network model and unlimited time, this will give us optimal results.

\subsection{Defender Policy: Evolutionary Diversity Optimization}
The defender uses Evolutionary Diversity Optimization to block $k$ block-worthy edges in order to minimize the strategic attacker probability of reaching DA. NN acts as a fitness function for EDO and the fitness function computes the attacker’s probability of reaching DA, for a given blocking plan. The defender uses EDO to obtain a diverse set of defensive blocking plans to train NN, with an aim to improve the accuracy of trained NN for modelling the attacker. We have only considered the block-worthy edges for the defensive policy. The defensive state vector can be represented as: 
\begin{equation}
\label{def_state_vec}
\underbrace{< 0, 1, 0, 0, 1, 1, 0, 0>}_\text{Length of defensive state vector = Number of block-worthy edges} 
\end{equation}
where:
\begin{description}
\item `1’ represents blocked edges
\item `0’ represents non-blocked edges
\end{description}
\textbf{Evolutionary Diversity Optimization.} We initially generate a random population $P$ of defensive blocking plans, and each blocking plan is a vector of size $|BW|$. The values in a state vector are either $1$ or $0$, and the number of $1$s is always equal to $k$ (defensive budget). We randomly select individuals from $P$ for mutation and crossover. We also draw a random number $x$ from Poisson distribution with a mean value equal to 1 and perform either mutation or crossover with a probability of $0.5$ to create new offspring. Our mutation and cross-over operators ensure that the number of blocked edges in an offspring does not exceed $k$; the mutation and crossover operations are discussed below:

\noindent \textit{\textbf{Mutation.}} We pick a random individual $p'$ from the population $P$, and flip $x$ 0s to 1s and $x$ 1s to 0s. For example, we pick a random individual $p'$ = <1,0,0,0,1,0,0,1> from the population $P$. To mutate $p'$, we flip two 0s to 1s and two 1s to 0s. New individual obtain is: <0,1,0,0,1,0,1,0>.

\noindent \textit{\textbf{Crossover.}} We pick two random individual state vector $p'$ and $p''$ from the population $P$ to crossover. We find $x$ coordinates where $p'$ has 0s on those coordinates, and $p''$ has 1s on those coordinates. For these coordinates, we change 0s in $p'$ to 1s and change 1s in $p''$ to 0s. We then find $x$ coordinates where $p'$ has 1s on those coordinates and $p''$ has 0s on those coordinates. We again change 1s to 0s and 0s to 1s.

\noindent \textbf{\textit{Diversity Measure.}} After mutation and crossover, we add the new individual to $P$ if its fitness value is close to optimal (within an absolute difference of 0.1); otherwise, we reject the individual even if it is good for the diversity of population. We consider the diversity in terms of equal representation of block-worthy edges in the population. Let us say there are $\mu$ individuals in the population $P$ and each individual $p_j$ is represented as: 
\begin{equation*}
  p_j =  \big((bw_1;j), (bw_2;j), ..., (bw_{|BW|};j)\big), \;\;\;\;\;\;j \in \{1,...,\mu\}
\end{equation*}

For each block-worthy edge $bw_i$, $\,\,i \in \{1,...,|BW|\}$, let $c(bw_i)$ denotes the block-worthy edge count as the number of individuals out of $\mu$ who have blocked this edge. Therefore, we get block-worthy edge count vector $C(bw)$ as:
\begin{equation*}
   C(bw) = (c(bw_1), c(bw_2), ..., c(bw_{|BW|})).
\end{equation*}

We now compute the diversity vector $D$ of the population $P$ without individual $p_{j}$ as: 
\begin{equation*}
  D(C(bw)\backslash{p_j}) = C(bw) - p_j,
\end{equation*}
\begin{equation*}
   D(C(bw)\backslash{p_j}) = \Big(c(bw_1)-(bw_1;j),..., c(bw_{|BW|})-(bw_{|BW|};j)\Big),
\end{equation*}
where $D(C(bw)\backslash{p_j})$ represents the diversity of population without individual $p_{j}$. In order to maximize the blocked edge diversity, we aim to minimize the $SortedD(C(bw)\backslash{p_j})$ in the lexicographic order, where sorting is done in descending order.

\begin{flalign*}
\text{\textit{Sorted}} D(C(bw)\backslash{p_j})=  sort \Big(\big(c(bw_1)-(bw_1;j)\big),  .......,\big(c(bw_{|BW|})-(bw_{|BW|};j)\big)\Big).
\end{flalign*}
In this way, we determine the diversity \textit{SortedD}, of the population without each individual in the population. We then select the individual $r$, removal of which leads to maximum diversity (minimum $SortedD(C(bw)\backslash{p_r})$) in population. We reject $r$ if it contributes least to diversity and its fitness value is not close to optimal. One exception is that if the newly introduced individual is the best-performing individual, then we remove the individual with the worst fitness value instead. In this way, we get a diverse set of defensive blocking plans via EDO.

\noindent \textit{\textbf{Converting defensive blocking plan to actual state vector.}} For each defensive blocking plan vector (\textcolor{blue}{Equation \ref{def_state_vec}}), we first need to convert it to the \textit{actual state vector} (\textcolor{blue}{Equation \ref{attacker_state_vector}}) in order to determine the attacker’s chances of reaching DA (corresponding to this blocking plan). “Actual state vector” is the state vector used in the attacker’s MDP. For each block-worthy blocked edge $bw$ in the defensive state vector, we first determine the number of NSPs blocked by this $bw$ edge and then change the status of those NSPs to `F’, keeping the status of rest to `?’.

\noindent \textbf{Overall approach: NNDP-EDO.} Initially, we have an NN which is highly inaccurate (attacker’s policy). We generate a diverse set of blocking plans as training data for NN using EDO (Defender’s policy), and NN acts as an efficient fitness function for EDO. Notably, the exact fitness function is $\#P$-hard to compute, therefore, we use NN as a fitness function. We convert the blocking plans to the actual state vectors and in each training epoch select a random blocking plan out of the population to train NN, which we call an initial state vector. We use the state transition process to determine future states set, transition probabilities and admissible action set. We then train NN using \textcolor{blue}{Equation \ref{loss}} on the future state set to approximate the value function, i.e., NN outputs attacker’s probability of reaching DA with the given a blocking plan. We repeatedly train NN on diverse blocking plans aiming to improve the accuracy of NN for modelling the attacker. We repeatedly go back and forth between training NN and EDO processes for many rounds to get a well trained NN that can act as an efficient fitness function for EDO. In the last round, we train NN once again on the population obtained from the last round of EDO. EDO prevents the NN model from getting stuck into local optima too early, especially in the early phases of the training when the value function is highly inaccurate. In addition, not all states are important for the attacker; therefore, we let the value of states that are referenced by the attacker’s optimal decision path, be more accurate by training NN on these states.

\section{Experimental Results}\label{experiments}
We discuss the effectiveness of the proposed approach by conducting exhaustive experiments on various AD attack graphs. 
We conducted all the experiments on a high-performance cluster server with Intel Gold 6148/6248 CPUs. All trials are executed using 1 CPU and 1 Core. Notably, we conducted the experiments on a very large scale, and it took us 141.47 days of computing hours to run all the experiments. We allocated 180 CPUs from the high performance cluster to run our 180 trials in parallel, therefore, our experiment completed in 1 day. The main reason for the high computational cost is training the neural network to learn the dynamic programming's recursive relation. The EDO process is not computationally expensive; it's rather efficient. Therefore, the time-intensive aspect lies in the neural network training process. We implemented the code in PyTorch.

\subsection{Dataset} 
We generated synthetic R500, R1000, R2000 AD attack graphs using DBCREATOR, where 500, 1000 and 2000 are the number of computers. In addition, we only consider three types of edges by default present in \textsc{BloodHound}; HasSession, AdminTo, and MemberOf. \textcolor{blue}{Table \ref{dataset}} presents the summary of original attack graphs, where nodes in original graphs include computers, user accounts and security groups. We pick 40 nodes that are farthest away from DA (in terms of number of hops) and randomly select 20 nodes as entry nodes. An edge$(i,j)$ is set to be blockable with probability $L$:
\begin{equation*}
  L =  \frac{\text{Min number of hops between edge$(i,j)$ and DA}}{\text{Max number of hops between edge$(i,j)$ and DA}}.
\end{equation*}
In this way, the farthest edges from DA are more likely to be blockable, which are less important. 
We pre-process the original attack graphs to obtain the condensed graphs, which only contains splitting nodes. For instance, in R500 original graph, there are 7 DA, 1493 nodes and 3456 edges. We merge the 7 DA into one destination node and remove all the outgoing edges from DA (once the attacker reaches DA, the attack ends). As the attacker will never use the incoming edges to entry nodes, we remove them. We also remove the nodes with no incoming edges. Out of 1493 nodes, only 105 can reach DA; we remove all the nodes that can not reach DA. In addition, splitting nodes are connected via non-splitting paths, therefore, we consider a non-splitting path as a single edge. In this way, we obtain a condensed graph, which is much smaller than the original AD graph.

\begin{table}[ht!] 
\caption{Summary of synthetic original AD attack graphs.}
\label{dataset}
\renewcommand{\arraystretch}{1}
\centering
\footnotesize
\begin{tabular}{ccc} \hlineB{2} 
\textbf{{AD attack graph} }& \textbf{{Nodes}} & \textbf{{Edges}}\\ \hlineB{2}
R500 & 1493 & 3456 \\
R1000 & 2996 & 8814\\
R2000 & 5997 & 18795\\\hlineB{2}
\end{tabular} 
\end{table}

\subsection{Training Parameters}
In our approach, we use a simple fully connected Neural Network, and a ReLU activation function follows each NN layer. For the R500 graph, a small NN with 4 fully connected layers is used, whereas, for R1000 and R2000, we use an NN with 10 layers. The size of each layer is 256 and the last layer of NN is followed by a sigmoid activation function that maps the model output to a value between 0 and 1 (the attacker’s chances of reaching DA). We train the NN in a batch of 16 states. We use the mean squared error to compute the loss, and train the parameters using Adam Optimizer with a learning rate of 0.001. The model is trained for 500 epochs in each round. 
We set the defender budget to 5. We generate a population of 100 defensive blocking plans in 10000 iterations and perform mutation or crossover with a probability of 0.5. We repeat the combined process (training NN and generating blocking plans using EDO) for 100 rounds. We perform experiments with a seed from 0 to 9 for 10 separate trials (we conduct experiments on 10 different AD graphs with different entry nodes and different blockable edges).

\subsection{Baselines}
We compare our proposed approach with a combination of various attacking and defensive policies, and the details are described below:
\begin{enumerate}[leftmargin=*]
    \item \textit{\textbf{NNDP-EDO (Proposed).}} We proposed NNDP approach to approximate the attacking policy and EDO for defensive policy. The defensive plan that contributes least to the diversity is rejected, and the fitness value is evaluated using NNDP.
    \item \textit{\textbf{NNDP-EDO+DP.}} NNDP is used to approximate the attacking policy and EDO for defensive policy. However, instead of using NNDP to evaluate the success rate of best blocking plan (as in our proposed approach), we use accurate DP to determine the value of the defensive plan.
    \item \textit{\textbf{Exact solution.}} Dynamic Program is used as attacker’s policy, and defensive blocking plan is obtained by exhaustively trying each in order to get the best plan.
    \item \textit{\textbf{NNDP-VEC.}} In this approach, NNDP is used to approximate the attacking policy, and Value-based Evolutionary Computation (VEC) is used to design the defensive policy. In VEC, the blocking plans with the worst fitness values are rejected, and the fitness value is evaluated using NNDP.
    \item \textit{\textbf{NNDP-Greedy.}} The attacker uses NNDP to approximate the attacking policy and defender adopts a greedy approach to design the defensive policy. The defender greedily blocks single best block-worthy edge to minimize the attacker’s chances of reaching DA using NNDP. The defender repeats this for $k$ times.
\end{enumerate}

Notably, NNDP-EDO+DP and Exact solution use DP to evaluate the blocking plan; however, it is infeasible for DP to process large graphs. Therefore, for R500 graph, we compare our proposed approach NNDP-EDO, with NNDP-EDO+DP and Exact solution. For larger graphs R1000 and R2000, we compare our approach with the NNDP-VEC and NNDP-Greedy to determine our defensive strategy’s effectiveness. The trained NNDP may not give us the accurate values of success rate for the defensive plan; therefore, we use Monte Carlo simulations to determine the effectiveness of the defensive plan. Moreover, in order to investigate the impact of correlation between failure rate $p_f$ and detection rate $p_d$ of edges on the success rate of attacker, we have considered three types of distribution; Independent (I), positive correlation (P) and negative correlation (N). In independent distribution, $p_f$ and $p_d$ are mutually independent, and we set $p_f$ and $p_d$ based on independent uniform distribution from 0 to 0.2. Positive correlation between $p_f$ and $p_d$ indicates that both have a steady relationship in the same direction. Negative correlation between $p_f$ and $p_d$ indicates an inverse relationship; one decreases as the other increases. We use multivariate normal distribution to get $p_f$ and $p_d$ for positive and negative correlation between an edge $e$ as:
\begin{equation*}
p_{d(e)}, p_{f(e)} = \text{multivariate normal (mean, cov)}.
\end{equation*}

For positive correlation, mean = $[0.1, 0.1]$ and cov = $[[0.05^2, 0.5\times0.05^2], [0.5 \times 0.05^2, 0.05^2]]$. For negative correlation, mean = $[0.1, 0.1]$ and cov = $[[0.05^2, -0.5\times0.05^2], [-0.5 \times 0.05^2, 0.05^2]]$.

\subsection{Results}
In our results, “Success Rate” represents the attacker’s chances of successfully reaching DA before getting detected (given attacker’s policy) under the defensive blocking plan from the defender’s policy. “Time (s)” is the number of seconds per trial. 
We simulate the attacker’s policy (NNDP) on the best defensive blocking plan (predicted by NNDP) using Monte Carlo simulations over 100000 runs to determine attacker’s chances of reaching DA. We performed all the experiments ten times with a seed from 0 to 9, and avg. represents the average results over all seeds are presented as final results.

\begin{table*}[t!] 
\caption{Comparison of success rate on R500 AD attack graph.}
\label{r500}
\renewcommand{\arraystretch}{1.2}
\centering 
\footnotesize	
\begin{tabular}{llllll} \hlineB{2} 
\textbf{Graph} & \textbf{Approach} &\textbf{I} & \textbf{P} & \textbf{N} & \textbf{Avg.}\\ \hlineB{2} 
      & NNDP-EDO (Proposed) & 87.96\% & 84.77\% & 84.5\% & 85.74\%\\
R500 & NNDP-EDO + DP & 90.49\% & 84.8\% & 84.46\%  & 86.58\%  \\
      & Exact solution & 89.73\%  & 84.78\% & 84.45\% & 86.32\% \\\hlineB{2} 
\end{tabular}
\end{table*}

\noindent \textit{\textbf{Results on R500.}} \textcolor{blue}{Table \ref{r500}} presents the success rate of the proposed approach NNDP-EDO and other baselines on synthetic R500 graph under various distributions. The attacker’s average success rate is  87.96\% when simulated using Monte Carlo over 100000 runs with our proposed defense EDO and NNDP (under independent distribution). Moreover, the exact actual success rate for our proposed defense when evaluated using NNDP-EDO+DP is 90.49\%. This indicates that for a blocking plan from EDO, our trained NNDP generates an error of 2.53\% in the success rate. For exact optimal defense, the attacker’s success rate is 89.73\%, which indicates that our proposed defense is 0.76\% (less than 1\%) away from the optimal. Similarly, the proposed defense NNDP-EDO performs near-optimal under positive correlation; NNDP-EDO believes the best defense has a success rate of 84.77\%, but in reality, the accurate success rate of the proposed defense is 84.8\%, which is slightly worse than the exact solution 84.78\%. In negative correlation as well, NNDP-EDO defense is nearly as effective as optimal.

\noindent \textit{\textbf{Results on R1000.}} \textcolor{blue}{Table \ref{r1000}} presents the results for R1000 AD attack graph. It is impossible to determine the exact attacker’s success rate for R1000 graph using DP; therefore, we use NNDP to approximate the attacker’s policy. The trained NNDP may not provide us with accurate success rate values for the defensive plan. Consequently, we employ Monte Carlo simulations to assess the effectiveness of the defensive plan. \textcolor{blue}{Table \ref{r1000}} shows that, on an average, the proposed EDO based defense NNDP-EDO, leads to the best defensive policy as compared to the other defense (NNDP-VEC and NNDP-Greedy). For instance, under independent distribution, the attacker’s success rate is 42.69\% with EDO defensive policy; however, the success rate increases to 43.19\% and 53.89\% when facing VEC defense and Greedy defense, respectively. Under positive correlation, NNDP-EDO is the best defense among the three and leads to minimum attacker’s success rate. However, under negative correlation, the VEC defense is slightly better than EDO. The results show that overall EDO’s best defense has an average success rate of 42.14\%; VEC best defense has a success rate of 43.04\%, which is slightly worse than EDO. The greedy defense has an average success rate of 52.35\%, which is far worse than VEC.

\noindent \textit{\textbf{Results on R2000.}} The results for R2000 AD attack graphs in \textcolor{blue}{Table \ref{r2000}} show that EDO based defence outperforms VEC and Greedy based defense in terms of average success rate. With the EDO defense (under independent distribution), there are only 31.07\% chances of attacker’s reaching DA; however, with the Greedy defense, the attacker’s success rate increases to 38.32\%. EDO performs better than VEC and Greedy in both independent and positive correlation; however, VEC performs slightly better than EDO in negative correlation. On an average best defense from EDO has an average success rate of 33.51\%; VEC has success rate of 34.62\%, slightly worse than EDO and Greedy has an average success rate of 40.01\%, which is very high as compared to EDO defense.

\begin{table*}[t!] 
\caption{Comparison of success rate on R1000 AD attack graph.}
\label{r1000}
\renewcommand{\arraystretch}{1.2}
\centering 
\footnotesize
\resizebox{\textwidth}{!}{\begin{tabular}{p{1.2cm}llllllll} \hlineB{2} 
\multicolumn{1}{c}{\textbf{{}}}& \multicolumn{1}{c}{\textbf{{}}}&  \multicolumn{4}{c}{\textbf{{Success Rate}}} & \multicolumn{3}{c}{\textbf{{Time(s)}}} \\  \cmidrule(lr){3-6}  \cmidrule(lr){7-9}

\textbf{Graph} & \textbf{Approach} &\textbf{I} & \textbf{P} & \textbf{N} &  \textbf{Avg.} & \textbf{I} & \textbf{P} & \textbf{N}\\ \hlineB{2} 

      & NNDP-EDO (Proposed) & \textbf{42.69\%} & \textbf{42.44\%} & 41.29\% & \textbf{42.14\% } & 39410.85 & 60263.28  & 64093.08 \\
R1000 & NNDP-VEC & 43.19\% & 44.7\% &\textbf{41.24\%} & 43.04\% & 39815.72 & 61572.04  & 57708.75 \\
      & NNDP-Greedy & 53.89\%  & 52.86\% & 50.31\% & 52.35\% & 38413.97 & 64555.83 & 67436.33 \\\hlineB{2} 
\end{tabular}}
\end{table*}

\begin{table*}[t!] 
\caption{Comparison of success rate on R2000 AD attack graph.}
\label{r2000}
\renewcommand{\arraystretch}{1.2}
\centering
\footnotesize
\resizebox{\textwidth}{!}{\begin{tabular}{lllllllll} \hlineB{2} 
\multicolumn{1}{c}{\textbf{{}}}& \multicolumn{1}{c}{\textbf{{}}}&  \multicolumn{4}{c}{\textbf{{Success Rate}}} & \multicolumn{3}{c}{\textbf{{Time(s)}}} \\  \cmidrule(lr){3-6}  \cmidrule(lr){7-9}

\textbf{Graph} & \textbf{Approach} &\textbf{I} & \textbf{P} & \textbf{N} &  \textbf{Avg.} & \textbf{I} & \textbf{P} & \textbf{N}\\ \hlineB{2} 

      & NNDP-EDO (Proposed) & \textbf{31.07\%} & \textbf{33.9\%} & 35.56\% & \textbf{33.51\% } & 27291.01 & 65488.9 & 59843.9 \\
R2000 & NNDP-VEC & 33.23\% & 36.45\% & \textbf{34.19\%} & 34.62\% & 26838.6 & 67079.8 & 57160.1 \\
      & NNDP-Greedy & 38.32\%  & 42.69\% & 39.02\% & 40.01\% & 24478.4 & 65620.3 & 57412.8 \\\hlineB{2} 
\end{tabular}}
\end{table*}

\subsection{Discussion}
The results show that the proposed defense NNDP-EDO is highly effective. We have exact optimal results for R500 attack graph, and on average, our proposed approach is less than 1\% away from the optimal defense. In addition, the proposed attacker’s policy NNDP approximates the success rate for the defense with high accuracy and incurs a small error of 2.53\%. This shows that EDO trains NNDP very effectively, and trained NNDP acts as a very efficient fitness function for EDO. It is impossible for large R1000 and R2000 graphs to run the exact DP evaluation function; therefore, we simulate the defense using Monte Carlo simulation to get attacker’s success rate. The results show that the proposed approach NNDP-EDO is better than others, and overall, the best defense from EDO has an average success rate of 42.14\% for R1000 and 33.51\% for R2000, which is far less than other defensive approaches. In addition, the results show that the Value-based Evolutionary Computation proves to be a better defense than the greedy defense. For R1000 and R2000 graphs, 6/180 trails (around 3\% ) ended up with unconverged NN training, i.e., after 500$\times$100 epochs, the cost function remains to be large.

\section{Chapter Summary}\label{conclusion}
This chapter investigated a Stackelberg game model on an Active Directory attack graph between an attacker and a defender. The defender aims to block a number of edges to minimize the attacker’s probability of reaching DA; however, the attacker aims to maximize their chances of reaching DA. We first proved that both the attacker’s and defender’s problems are $\#P$-hard. We proposed Evolutionary Diversity Optimization to solve the defender’s problem, and train the Neural Network to solve the attacker's problem. The experimental results showed that the proposed Evolutionary Diversity Optimization based defensive policy is highly effective, and our proposed approach can solve AD graph problems, which are intractable to solve using a conventional Dynamic Program. Moreover, for R500 AD attack graph, our proposed approach is less than 1\% away from the optimal defense.






\chapter{Evolving Reinforcement Learning Environment to Minimize Learner’s Achievable Reward: An Application on Hardening Active Directory Systems} 

\label{AD_RL} 

\textbf{\underline{Related publication:}} 
\vspace{0.06in}

\noindent This chapter is based on our paper titled “\textit{Evolving Reinforcement Learning Environment to Minimize Learner’s Achievable Reward: An Application on Hardening Active Directory Systems}” published in The Genetic and Evolutionary Computation Conference (GECCO), 2023 \cite{goel2023evolving}.

\vspace{0.1in}

\noindent \textcolor{blue}{Chapter 3} studied defensive techniques for identifying bottleneck edges that can be blocked to defend organization’s active directory graphs. However, devising defensive policies for large-scale AD graphs can be challenging. Therefore, this chapter presents another approach aiming to defend large-scale AD graphs. In this chapter, we study a Stackelberg game between one attacker and one defender in configurable environment settings to defend AD graphs. The defender picks a specific environment configuration that represents an edge-blocking plan. The attacker observes the configuration and attacks via Reinforcement Learning (RL). The defender aims to find the environment configuration with the minimum achievable reward for the attacker, i.e., minimum success rate of reaching the DA. We propose an Evolutionary Diversity Optimization based defensive policy to generate a diverse population of defensive environment configurations, and these environments are used for training the attacker’s policy. The diversity not only improves training quality but also fits well with our RL scenario, i.e.,  RL agents tend to improve gradually, so a slightly worse environment earlier on may become better later.

Overall, this thesis proposes two approaches, i.e., neural network based dynamic program and reinforcement learning based approach for defending AD graphs. Our experimental results demonstrate that the reinforcement learning based approach is scalable to larger AD graphs than the neural network based dynamic program approach in \textcolor{blue}{Chapter 3}.

\section{Introduction}
In an adversarial environment, the attacker and defender play against each other \cite{creswell2018generative}. The attacker intends to devise a technique to successfully carry out an attack, while the defender’s objective is to protect the system from being compromised. In such environments, the strategies and actions of the attacker and defender are interdependent and affect each other \cite{deldjoo2021survey}. This chapter studies an attacker-defender Stackelberg game in configurable environment settings, where the defender (leader) tries various environment configurations to protect the system \cite{bruckner2011stackelberg}. In contrast, the attacker (follower) observes the environment configurations and designs an attacking policy using Reinforcement Learning (RL) to maximize their rewards. The defender aims to find an environmental configuration where the attacker’s attainable reward is minimum. We consider a specific application scenario, “\textit{\textbf{hardening active directory systems}}”, to discuss the problem in detail.

Active Directory is a directory service developed by Microsoft for managing and securing network resources in Windows domain networks. It is considered as the default security management system for Windows domain networks \cite{dias2002guide} and has become a popular target for cyber attackers. Given the popularity of organizational AD graphs and the large number of cyber attackers targeting AD  graphs, security professionals are devising various solutions to protect AD \cite{dunagan2009heat, guo2021practical, guo2022scalable, Yumeng23:Near, quang}. One solution is to selectively block certain edges from the attack graph to prevent attackers from reaching DA \cite{dunagan2009heat}. Edge blocking in AD graphs can be performed by either abrogating access or monitoring to stop the attacker from reaching the DA.

This chapter studies a Stackelberg game between an attacker and one defender on AD graphs, where the attacker aims to design a strategy to maximize their chances of successfully reaching the DA, and the defender aims to design a strategy to minimize the attacker’s success rate. The attacker in our model is strategic and adopts a proactive approach while performing the attack (For details, refer \textcolor{blue}{Section \ref{def_intro}}). Each edge in AD graph is associated with a failure rate, detection rate and success rate. The attacker starts the attack from one of the entry nodes and attempts to traverse an edge to reach new nodes. If the attacker fails to traverse the edge, the attacker tries another edge until gets detected, has tried all possibilities or reaches the DA. Notably, if the attacker previously failed to pass through an edge, then they do not try this edge again. The strategic attacker maintains a set called \textbf{\textit{checkpoints}}, which are nodes that the attacker can use as starting points or continue an attack from. Initially, the checkpoint consists of only the entry nodes. 
The optimal attacking policy is the one that maximizes attacker’s chances of reaching the DA without getting detected. Notably, it is essential to design a sophisticated attacker’s policy, as we can not have an effective defense without having an accurate attacker’s policy. Furthermore, {the defender’s goal is to design a defensive policy to minimize the attacker’s success rate.} The defender blocks a set of $k$ edges by increasing the edge’s failure rate from the original to 100\%. 


\textit{\textbf{This chapter aims to design the defender’s policy to block a set of edges with the goal of minimizing the attacker’s chances of successfully reaching the DA.}}  For the attacker problem, we propose a \textbf{\textit{Reinforcement Learning (RL)}} based policy to maximize attacker’s chances of successfully reaching the DA (maximize attacker’s achievable reward). We propose a \textbf{\textit{ Critic network assisted Evolutionary Diversity Optimization (C-EDO)}} based defensive policy to find the defensive plan configurations that minimize the attacker’s success rate. The attacker’s problem of maximizing their chances of successfully reaching DA can be modelled as a Markov Decision Process (MDP) \cite{puterman1990markov}. We propose RL based policy to approximate the attacker’s problem. Specifically, we use {Proximal Policy Optimization} RL algorithm, an {Actor-Critic based approach} to train the attacker policy \cite{schulman2017proximal}. The RL agent interacts with the environment by suggesting actions with the goal of maximizing the overall reward. In our proposed approach, the RL agent interacts and learns from “\textit{multiple independent environments}” simultaneously. The RL training process is continuous and not interrupted by any other process. We propose a Critic network assisted Evolutionary Diversity Optimization based policy to solve the defender problem. C-EDO generates numerous environment configurations (defensive plans). 
Our approach uses the trained RL critic network to estimate the fitness of environment configurations. The defender continuously monitors the RL training process and after regular intervals, the defender uses the trained critic network to evaluate the current configurations and uses C-EDO to generate better ones. The defender follows the approach of discarding the environment configurations favourable for the attacker and replacing them with better ones. In our proposed RL+C-EDO approach, the attacker and defender play against each other in parallel. On the other hand, in \textcolor{blue}{Chapter \ref{Chapter_AD_NNDP}}, we propose Evolutionary Diversity Optimization (EDO) as the defender's policy. In EDO, the Neural Network acts as a fitness function that computes the attacker’s probability
of reaching DA. The defender uses EDO to obtain a diverse
set of defensive blocking plans to train the Neural Network, with an aim to improve the accuracy of the trained Neural Network for modelling the attacker. 

Overall, the defender’s C-EDO generates numerous high-quality, diverse environment configurations (defensive plan) worth learning for the RL agent to train a better attacking policy. We emphasize on high-quality, diverse environments to prevent the RL agent from spending time on learning environments that are irrelevant to an effective defense. Besides, we train our attacker’s policy using actor-critic based algorithm, so we inherently have a critic network that gives us a reasonable estimation of the state values and can be used as a fitness function for defender’s C-EDO. 
In this way, these two problems are interconnected, and the solution strategies complement each other. We conduct extensive experimental analysis and compare the results with our previously proposed NNDP-EDO approach \cite{Goel2022defending}. Our experimental results show that the RL+C-EDO approach achieves better results than our NNDP-EDO approach and the reason is that we have used RL to approximate the attacker’s problem. In NNDP-EDO approach, NNDP attacker’s policy trains the model against one defensive plan at a time, due to which it forgets the previous plan. This way, it keeps learning and forgetting the plans. However, in RL+C-EDO approach, we train our RL based attacker’s policy against multiple defensive plans (environment configurations) at a time, due to which it learns shared experience and performs better. For the RL agent, diverse environment configurations are only different in the “opening games”, whereas the “end games” or “mid games” are likely to be similar across different environments. The similarity in later stages can be utilized in parallel training, where the agent is trained against multiple environments simultaneously and gains shared experience, leading to faster convergence and improved performance. Besides, the NNDP approach is value iteration-based RL algorithm, whereas our approach is policy iteration-based RL algorithm. In general, the policy iteration-based algorithms converge faster than value iteration-based algorithms \cite{kaelbling1996reinforcement}, which is another reason for the superior performance of our approach. Due to the reasons mentioned above, the performance of our proposed  RL+C-EDO approach is better than our NNDP-EDO approach.

\noindent \textbf{\textit{Contributions:}} In this chapter, we make the following contributions.
\begin{itemize}
    \item \textbf{Attacker policy. }We propose a Reinforcement Learning  based policy to solve the attacker problem and train the RL agent in parallel on multiple environments, to accelerate the training process.
    
    \item  \textbf{Defender policy.} We propose Critic network assisted Evolutionary Diversity Optimisation based policy to address the defender problem. The defender generate diverse blocking plans and replicates those plans that perform well for defender and diminishes the ones that do not. 
    
    \item \textbf{Extensive experiments.} Our experimental results demonstrate that the proposed approach is highly effective and scalable to R4000\footnote{R4000 AD graph is an AD graph containing 4000 computers.} AD graph. The proposed RL based attacker policy approximates attacker’s problem more accurately and the proposed defensive approach generates better defense. 
\end{itemize}

\noindent \textbf{\textit{Chapter organization:}} \textcolor{blue}{Section \ref{RL_problem}} discusses the problem description in detail. \textcolor{blue}{Section \ref{RL_approach}} discusses the proposed approach, including pre-processing AD graphs, reinforcement learning based attacker’s policy, critic network assisted evolutionary diversity optimization based defender’s policy and attacker-defender overall approach. \textcolor{blue}{Section \ref{RL_experiments}}  reports the experimental results and finally, \textcolor{blue}{Section \ref{RL_conclusion}} concludes the chapter.

\section{Problem Description}\label{RL_problem}
Active directory graph can be represented as a directed graph $G = (V, E)$, where $V$ is nodes set, and $E$ is edges set. The highest privilege accounts in AD graphs are called \textsc{Domain Admin} (DA). This chapter considers a two-player Stackelberg game between one defender and one attacker to defend AD graphs. There are $s$ entry nodes. The attacker can start from one of the entry nodes and aims to design a strategy to maximize their chances of successfully reaching DA. The defender seeks to block a set of edges so as to minimize the attacker’s success rate. The defender has a limited budget and can only block $k$ edges. Not all edges are blockable; therefore, the attacker can only block `\textit{blockable}’ edges. Edge blocking is costly and requires efforts (auditing access logs) to safely remove an edge; due to this, we cannot remove too many edges and have assumed a limited budget. In our model, every edge in the AD graph has a detection rate, failure rate and success rate. For more details, please refer to \textcolor{blue}{Section \ref{model_description}}. 
The strategic attacker starts from one of the entry nodes and tries unexplored edges in order to reach DA. At any time during the attack, the checkpoint indicates the set of nodes that the attacker has control of and can continue the attack from (in case the attacker fails to pass through the edge). \textcolor{blue}{Chapter 3} proved that computing defender’s and attacker’s optimal policy (and value) is $\#$P-hard. Therefore, to approximate the attacker problem, we design a reinforcement learning based policy where the agent learns from multiple environment configurations (defensive plans) at a time, in turn accelerating the training process. We propose a critic network assisted evolutionary diversity optimization policy to address the defender problem.

\section{Proposed Approach}\label{RL_approach}
This section describes our proposed approach for defending active directory graphs. We first discuss our proposed pre-processing procedure that converts the original AD graph to a smaller graph. We then discuss our proposed reinforcement learning based attacker’s policy and critic network assisted evolutionary diversity optimization based defender’s policy. Later, we describe our overall attacker-defender approach.

\subsection{Pre-processing AD graphs}
We first pre-process our AD graph by exploiting its structural features to get a smaller graph. In an AD graph, \textit{Splitting nodes} are the nodes with multiple edges originating from them. \textit{Entry nodes} are the starting points from where the attacker can initiate an attack. The set of splitting nodes and entry nodes is represented by \textsc{Split} and \textsc{Entry}, respectively. 
If at least one of the edges on an NSP \textcolor{blue}{(Definition \ref{def_NSP})} is blockable, only then we say that the NSP is blockable.

\begin{definition}
\label{def_bw}
\noindent \textbf{Block-worthy edge (bw).}  Any blockable edge farthest away from node $i$ on $NSP(i,j)$ is known as block-worthy edge $bw(i,j)$. The block-worthy edge set is defined as:
\end{definition}
\begin{equation*}
BW = \{bw(i,j)|\, i \in \textsc{Split}\, \cup \,\textsc{Entry}, j \in \textsc{Successors}(i)\}.
\end{equation*}
A block-worthy edge can be shared among two NSPs. We only spent one budget unit on blocking $NSP(i,j)$. 
\textit{If the original graph contains $n$ nodes and $m$ edges; after pre-processing, the resulting graph consists of ($|\textsc{Entry}|+|\textsc{Split}| + 1)$ nodes and $|NSP|$ edges.}

\subsection{Attacker Policy: Reinforcement Learning }
The attacker’s goal is to design a strategy that maximizes their chances of successfully reaching the DA. We describe the attacker’s problem as a Markov Decision Process and propose a reinforcement learning based policy to address the attacker’s problem. Our proposed RL based attacker policy uses \textit{Proximal Policy Optimization (PPO)} algorithm, an Actor-Critic based approach, to train the RL agent. 
We train our RL agent in parallel against multiple instances of environment configurations at a time, and each environment contains a defensive plan from the defender. Training the RL agent in parallel on numerous environments speeds up the training process as the agent is able to learn from various defensive plans at a time. Moreover, the shared experience that the RL agent gains from the different environments can be used to make better decisions and achieve a higher reward in the final stages of the game. This section discusses our proposed RL based attacker’s policy in detail.

\subsubsection{Environment}
We model the attacker’s problem of designing a policy to maximize their chances of successfully reaching the DA as MDP. We call the attacker’s MDP $M = (S, A, R, T)$  as an environment, where $S$ denotes the state space, $A$ denotes the action space, $R$ is the reward function, and $T$ represents the transition function. The description of the environment is discussed below.

\noindent \textbf{\textit{State space (S):}} State space $S$ is a finite set of attacker’s states, and state $`s$’ is a vector of size equal to the number of NSPs in AD graph. Each coordinate in state $s$ represents one NSP, and the status of each NSP is either `S’, `F’ or `?’. We represent the attacker state $s$ as described in \textcolor{blue}{Equation \ref{attacker_state_vector}}. 
Given a state $s$, the attacker explores one of the NSPs with status  `?’ and the status of tried NSP changes to either `S’ or `F’ \footnote{Status of some other NSPs may also change, in case the destination is already reached, or joint block-worthy edge is failed.}, and the attacker reaches a new state. The attacker continues to explore one of the unexplored NSPs at a time to reach a new state till the attacker reaches DA or gets detected. At any time $t$ during the attack, the attacker’s current state $s_t$ acts as a  knowledge base and conveys the following information: the set of NSPs that the attacker has control of (NSPs with status `S’), NSPs that the attacker has failed on (NSPs with status `F’) and NSPs that the attacker can try in future (NSPs with status `?’). The attacker has two base states: 1) When the attacker reaches DA, the attack is successful and terminates; 2) When the attacker is not able to reach the DA; the reason can be that they got detected or has tried all possible NSPs, and there is no option left to explore; in this case, the attack fails and ends.

\noindent \textbf{\textit{Action space (A):}} Action space $A$ is the action set available for state $s$, which are the outgoing NSPs from the successful NSPs in state $s$. The attacker’s action space is discrete. Action $a$ is one of the NSPs from the action space of state $s$ and indicates that the attacker tries this NSP to reach the DA. 

\noindent \textbf{\textit{Reward function (R):}} The reward $r(s, a)$ for state $s$ on performing an action $a$ is $1$ if the attacker is able to reach the DA without getting detected. Otherwise, the reward is $0$.

\noindent \textbf{\textit{Transition function (T):}} For a given state $s$ and action $a$, the transition function performs action $a$ on state $s$, and may have a set of future states. Each future state is associated with its transition probability, and one state is selected as the next state (weighted by its transition probability). 

\subsubsection{Policy training}
We propose to utilize \textit{Proximal Policy Optimization (PPO)} RL algorithm to train the attacker’s policy so as to maximize attacker’s success rate. PPO follows an \textit{{Actor-Critic approach}} that exploits the advantages of policy based and value based approaches while eliminating their disadvantages. In this approach, the policy and value networks help each other improve. In our approach, we train two networks: actor network and critic network. \textbf{\textit{Actor network}}, also referred to as policy network, takes the current state $s$ as input and outputs the action $a$ to be performed on $s$. Actor network updates the policy parameters in the direction implied by the critic network. \textbf{\textit{Critic network}}, also known as value network, takes the state as input and outputs the value of state. {For an attacker problem, value of the state is the attacker’s success rate.}

Our RL agent uses the actor-critic based PPO algorithm to train the attacker’s policy by interacting with the environment (each environment is associated with a configuration, i.e., defensive plan). The RL agent does not possess any prior knowledge of the environment. Instead of training the RL agent against a single environment, we train the agent against \textit{\textbf{multiple parallel environments}}.  Each environment is initialized with the attacker’s initial state (defensive plan is converted to obtain attacker’s initial state). The following process is executed in all environments simultaneously. At each timestep $t$ of an episode, the RL agent receives state $s_t$. The trained actor network issues an action $a_t$ to be performed on the current state $s_t$ and action $a_t$ is sent to the environment. The environment performs action $a_t$ on state $s_t$ and reaches a new state $s_{t+1}$ (following the transition function), and receives a reward $r_{t+1}$ (following the reward function). The process repeats until we reach the base state, i.e., the attacker reaches DA or gets detected. In this manner, we obtain a sequence of states, actions, and rewards that terminates at the base state. The designed policy intends to maximize the total reward received during an episode. The final reward is the attacker’s success rate. This way, the attacker designs RL based policy to maximize their achievable reward (success rate).

\subsection{Defender Policy: Critic Network Assisted Evolutionary Diversity Optimization}
The defender’s goal is to block $k$ block-worthy edges to minimize the attacker’s chances of successfully reaching the DA. We propose a \textit{\textbf{Critic network assisted Evolutionary Diversity Optimization (C-EDO)}} based defensive policy that computes high quality and diverse environmental configurations (defensive plan). We aim to identify the valuable environments, i.e., the environment configurations that correspond to potentially good defense. Our main idea is to let the RL agent play against the environment configuration and if, after training for some time, the configuration proves to be unfavourable for the defender (i.e., the attacker has a high success rate against the configuration), we discard this environment configuration. We do not waste our computational effort on this environment and allocate our limited computational resources to other challenging environments. The high-quality and diverse characteristics of environments enhance the  accuracy of modelling the attacker. The trained RL critic network serves as a fitness function for the defender’s C-EDO. For every environmental configuration, the fitness function is used to evaluate that environment, i.e., calculate the attacker’s success chances against that configuration (defensive plan). The defender only blocks block-worthy edges to generate environment configuration. The defender’s environment configuration/defensive plan can be represented as:

\begin{algorithm*}[t!]
\caption{Critic network assisted Evolutionary Diversity Optimization based Defender’s Policy}
 \label{edo_algo}
 \begin{algorithmic}[1]
 \STATE Initialise population P with $\mu$ environment configurations
 \STATE Select individual $p’$ uniformly at random from $P$ and create an offspring $p’_{new}$ by mutation or crossover
 \STATE If $(OPT - 0.1) \leq fitness(p’_{new}) \leq (OPT + 0.1)$, add $p’_{new}$ to $P$
 \STATE If $|P| = \mu + 1$, remove  individual $r$ from $P$ that contributes least to diversity, i.e., minimum $SortedDiver(C(bw)\backslash{p_r})$
 \STATE Repeat steps 2 to 4 till the termination condition is met
\end{algorithmic}
\end{algorithm*}

\begin{equation}
\label{def_state_vec_ch4}
\text{Environment configuration} \,\, = \,\,< N, B, \,.\, .\, .\, ,\,B, N, B>,
\end{equation}
where `B’ denotes the blocked edges and `N’ denotes the non-blocked edges. Notably, length of the configuration (defensive plan) is equal to the number of block-worthy edges in AD graph.

\textcolor{blue}{Algorithm \ref{edo_algo}} outlines the defender’s policy. The process starts by generating arbitrary population $P$ of $\mu$ configurations as shown in \textcolor{blue}{Equation \ref{def_state_vec_ch4}}. In every configuration, the total number of $Bs$ is $k$ (defender’s budget). To create new offspring (environmental configuration), we perform mutation or crossover operation, each with $0.5$ probability on the randomly selected configuration $p'$ from the population $P$. We randomly select a variable $x$ from a Poisson distribution with mean value of 1. For \textit{\textbf{Mutation}}, we select an individual $p'$ randomly from $P$ and flip $x$ $Bs$ to $Ns$ and $x$ $Ns$ to $Bs$. \textcolor{blue}{Algorithm  \ref{mutationalgo}} presents the mutation operator for C-EDO. For \textit{\textbf{Crossover}}, we randomly select two individual $p'$ and  $p''$ from $P$ and look for $x$ indices such that $p'$ has $N$ and $p''$ has $B$ on those indices. Now, for these indices, we change $Ns$ to $Bs$ in $p'$ and $Bs$ to $Ns$ in $p''$. Similarly, we look for $x$ indices where $p'$ has $B$ and $p''$ has $N$ on those indices and change $Bs$ to $Ns$ in $p'$ and $Ns$ to $Bs$ in $p''$. The mutation and crossover operation make sure that exactly $k$ edges are blocked in the environmental configuration. We add the newly created offspring to the population only if its fitness score lies within the range of $(OPT \pm 0.1)$; otherwise, the offspring is rejected even though it is beneficial for diversity. If the new offspring is added to the population, we aim to maximize the blocked edge diversity and remove the individual that contributes least to the diversity. \textcolor{blue}{Algorithm \ref{crossoveralgo}} presents the crossover operator for C-EDO.


\begin{algorithm}[t!]
\caption{Mutation operator for C-EDO}
 \label{mutationalgo}
 \begin{algorithmic}[1]
 \renewcommand{\algorithmicrequire}{\textbf{Input:}}
 \renewcommand{\algorithmicensure}{\textbf{Output:}}
 \REQUIRE Population $P$, bits to mutate $x$
 \STATE $p' \leftarrow RandomSelect(P)$
 \STATE $p'' \leftarrow $Flip $x$ $Bs$ to $Ns$ and $x$ $Ns$ to $Bs$ in $p'$
 \RETURN $p''$
\end{algorithmic}
\end{algorithm}

\begin{algorithm}[ht!]
\caption{Crossover operator for C-EDO}
 \label{crossoveralgo}
 \begin{algorithmic}[1]
 \renewcommand{\algorithmicrequire}{\textbf{Input:}}
 \renewcommand{\algorithmicensure}{\textbf{Output:}}
 \REQUIRE Population $P$
 \STATE $p', p'' \leftarrow RandomSelect(P)$
 \FOR{$i = 1,...,len(p')$}
 \IF{$(p'[i]==0 \land p''[i]==1)$}
 \STATE $C.append(i)$
 \ENDIF
 \ENDFOR
 \STATE $C' = $ Randomly select $x$ bits from $C$
 \FORALL{$i\in C'$}
 \STATE set $p'[i] =1$ and $p''[i] =0$
 \ENDFOR
 \FOR{$i = 1,...,len(p')$}
 \IF{$(p'[i]==1 \land p''[i]==0)$}
 \STATE $D.append(i)$
 \ENDIF
 \ENDFOR
 \STATE $D' = $ Randomly select $x$ bits from $D$
 \FORALL{$i\in D'$}
 \STATE set $p'[i] =0$ and $p''[i] =1$
 \ENDFOR
 \RETURN $p'$
\end{algorithmic}
\end{algorithm}

\noindent \textbf{Maximizing blocked edges diversity.} We define the diversity measure as ``\textit{all block-worthy edges are equally blocked}''. Our proposed diversity measure aims to maximize the diversity of blocked edges in the environment configurations. It calculates how often each block-worthy edge is blocked in the configuration population, with the aim of making this frequency equal. When a new offspring is created using mutation or crossover, the offspring is added to the population only if its fitness score is close to the best fitness score and rejects the individual that contributes least to the diversity. Let us assume there are $\mu$ configurations (we call these configurations as individuals) in $P$. Each individual $p_i$ can be described as:

\begin{equation*}
  p_i =  \big((B/N, bw_1), (B/N, bw_2), ..., (B/N,bw_{|BW|})\big),
\end{equation*}
where B/N denotes the status of block-worthy edge; `B’ indicates blocked, `N’ indicates non-blocked, and $i \in \{1, . . .,\mu\}$. We then compute the block-worthy edge count vector $C(bw)$, which is defined as “for each block-worthy edge $bw_j, \,\, j \in \{1, ..., |BW|\}$, the number of individuals who have blocked $bw_j$ edge”.
\begin{equation*}
   C(bw) = (c(bw_1), c(bw_2), ..., c(bw_{|BW|})),
\end{equation*}
where $c(bw_1)$ denotes the total individuals out of $\mu$ who have blocked $bw_1$ edge. For each individual $p_i$, we then determine the vector $Diver(C(bw)\backslash{p_i})$, which computes the diversity of the population without individual $p_i$ as:
\begin{equation*}
\label{diversity_eq}
  Diver(C(bw)\backslash{p_i}) = C(bw) - p_i.
\end{equation*}

Our goal is to maximize the blocked edge diversity; therefore, we calculate $SortedDiver(C(bw)\backslash{p_i})$ as:
\begin{flalign*}
\text{\textit{SortedDiver}} (C(bw)\backslash{p_i})=  sort\Big(Diver(C(bw)\backslash{p_i})\Big).
\end{flalign*}

To maximize the diversity of blocked edges, our goal is to minimize the $SortedDiver(C(bw)\backslash{p_i})$ in lexicographic order where sorting is carried out in descending order. The individual $l$ with minimum $SortedDiver(C(bw)\backslash{p_l})$ is the one, removal of which maximizes the diversity. Therefore, the individual $l$ is removed from the population, if its removal maximizes the diversity and its fitness score is not close to the best. In special case, when the newly created offspring has the best fitness score, we add it to $P$ (even though it is worst in terms of diversity) and the individual with the lowest fitness score is discarded. Using this process, the defender creates diverse environment configurations. \textcolor{blue}{Figure \ref{fig:edge_diversity_eg}} presents an example of maximizing the blocked edge diversity in population.

\begin{figure}[t!]
\centering
 \includegraphics[width=0.5\paperwidth]{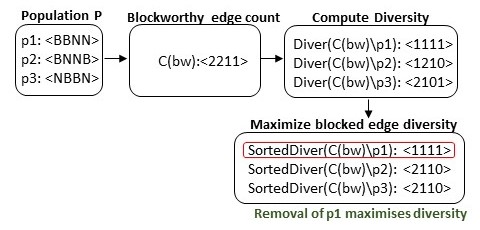}
 \caption{Example of maximizing blocked edge diversity.}
 \label{fig:edge_diversity_eg}
\end{figure}

\subsection{Attacker-Defender Overall Approach}
The defender employs C-EDO to generate high-quality and diverse environment configurations, each containing a defensive plan worth learning for the attacker policy. The attacker uses an actor-critic based RL algorithm to train their policy against the defender’s environment configurations. The attacker’s trained RL critic network serves as a fitness function for the defender C-EDO. The attacker first converts the defensive plans in environment configurations (\textcolor{blue}{Equation \ref{def_state_vec_ch4}}) to the attacker’s initial state\footnote{In attacker’s state, the status of NSPs corresponding to the blocked block-worthy edges in the defender’s environment configuration is changed from `?’ to `F’ to obtain attacker’s initial state.} (\textcolor{blue}{Equation \ref{attacker_state_vector}})). The RL agent then interacts and learns from the environments in parallel by issuing actions according to the trained policy. The environments perform the action and return a new state and reward. The quality of the trained policy is determined based on the total reward collected by the agent during an episode. Initially, the policy is not trained, and the action suggested might result in low rewards, but with training, it results in high rewards (attacker’s success rate). In this way, the RL agent trains the policy to maximize the reward. 

\begin{figure*}[t!]
\centering
\includegraphics[width=10.3cm,height=4.7cm]{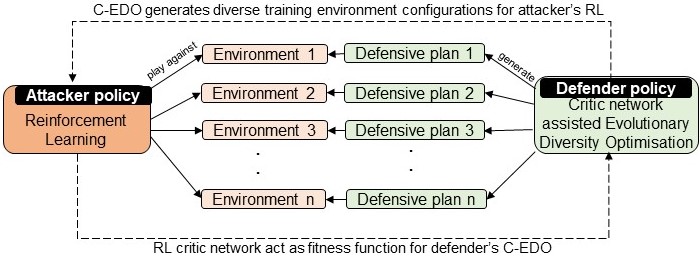}
 \caption{Proposed attacker-defender approach.}
 \label{fig:overall_approach}
\end{figure*}

Besides, our defensive policy is based on the idea of “\textit{replicating the environment configurations that perform well for the defender and diminishing the ones that do not}”. The defender process continuously monitors the RL training process and, after every regular interval, uses the trained critic network to evaluate the current set of environments. The defender discards those environments that are good for the attacker and replaces them with better ones. This way, at the beginning of every episode, the RL agent checks if the defensive plan configuration corresponding to each environment is changed or not. If changed, the attacker initializes the environment with the new defense configuration and starts training the agent against it. In this manner, the attacker and defender play against in parallel, where the attacker’s RL process continuously learns, and the defender process evaluates the current environment configurations and generates better ones. C-EDO’s diversity characteristic helps RL not get stuck in the local optimum. The RL agent experiences differences in environmental configurations only in the early stages, but the later stages tend to be similar across environments. We took advantage of this similarity and trained the agent on multiple environments in parallel, due to which our agent gained shared experience, resulting in faster convergence and enhanced performance. Overall, the trained attacker’s RL policy improves as the defender generates better environmental configurations using C-EDO. The defender’s policy generates better environments as the attacker’s critic network is trained. This way, these two processes assist each other to perform better. \textcolor{blue}{Figure \ref{fig:overall_approach}} illustrates our overall proposed approach.

\section{Experimental Results}\label{RL_experiments}
We executed experiments on a high-performance cluster server with Intel Gold 6148/6248 CPUs, utilizing 1 CPU and 20 cores for each trial. We used OpenAI Gym \cite{brockman2016openai} to implement the RL environments and trained the RL model using the PPO algorithm from the \textsc{Tianshou} library \cite{weng2021tianshou}. \textit{Success rate/Chances of success} indicates the attacker’s probability of reaching the DA without getting detected when the defender has blocked certain edges.

\subsection{Dataset}
Real-world AD graphs are highly sensitive; therefore, we used \textsc{BloodHound} team’s synthetic graph generator \textsc{DBCreator} to generate synthetic AD graphs. We generate AD graphs of four sizes, i.e., R500, R1000, R2000 and R4000, where 500, 1000, 2000 and 4000 indicate the number of computers in the AD graph. R500 AD graph contains 1493 nodes and 3456 edges; R1000 AD graph contains 2996 nodes and 8814 edges; R2000 AD graph contains 5997 nodes and 18795 edges; R4000 AD graph contains 12001 nodes and  45780 edges. Our experiments consider only 3 specific kinds of edges, by default present in \textsc{BloodHound}: \textsc{HasSession}, \textsc{MemberOf} and \textsc{AdminTo}. To set the attacker’s starting nodes, we first find 40 faraway nodes from DA and then arbitrarily set 20 nodes from them as the starting nodes.

\vspace{0.15in}
Each edge $e$ is blockable with a probability =  $\frac{\text{Min \#hops between e and DA}}{\text{Max \#hops between any e and DA}}$.
\vspace{0.15in}

This way, the edges farthest from the DA, i.e., less significant edges, are more likely to be blocked. It is challenging to perform computations on a large AD graph; therefore, we pre-process it to obtain a smaller graph. Pre-processing includes merging all DA into 1 DA, removing nodes and edges irrelevant for the attacker (outgoing edges from DA,  all incoming edges to the entry nodes and the nodes without any incoming edges). We also consider an NSP (\textcolor{blue}{Definition \ref{def_NSP}}) as one edge. 
Furthermore, to investigate the relationship between the failure rate $p_{f(e)}$ and detection rate $p_{d(e)}$ of an edge $e$ on attacker’s chances of success, we use three different distributions: Independent distribution (I), Positive correlation (P), and Negative correlation (N). In \textit{independent distribution}, we set the values of $p_{d(e)}$ and $p_{f(e)}$  as:
\begin{equation*}
p_{d(e)}, p_{f(e)}  = \text{Independent uniform (0, 0.2)}.
\end{equation*}
\noindent In \textit{positive correlation}, we set the values of $p_{d(e)}$ and $p_{f(e)}$ as:
\begin{equation*}
p_{d(e)}, p_{f(e)} = \mathcal{N}  (\mu, \Sigma),
\end{equation*}
where
\begin{equation*}
\mu = [0.1, 0.1] \text{\,\,and\,\,} \Sigma = [[0.05^2, 0.5\times0.05^2], [0.5 \times 0.05^2, 0.05^2]].
\end{equation*}
Here, $\mathcal{N}$ represents the multivariate normal distribution, $\mu$ represents the mean, and $\Sigma$ represents the covariance matrix.

\noindent In \textit{negative correlation}, we set the values of $p_{d(e)}$ and $p_{f(e)}$ as:

\begin{equation*}
p_{d(e)}, p_{f(e)} = \mathcal{N}  (\mu, \Sigma),
\end{equation*}
where
\begin{equation*}
\mu = [0.1, 0.1] \text{\,\,and\,\,} \Sigma = [[0.05^2, -0.5\times0.05^2], [-0.5 \times 0.05^2, 0.05^2]].
\end{equation*}

\subsection{Training Parameters}
\textit{\textbf{Reinforcement Learning:}}
We used a simple multi-layer perceptron neural network to implement the actor and critic network. 
The hidden layer size is 128 for R500 and R1000 AD graphs,  and 256 for R2000 and R4000 AD graphs. The parameters are trained using adam optimizer, learning rate of $1e^{-4}$ and batch size of $800$ states. For PPO-specific hyper-parameters, we used the standard hyper-parameters as specified in the original paper \cite{schulman2017proximal}. We created $20$ environments. For experimental setup 1, we parallelly train the RL policy for a total of $700$ epochs (1200 epochs for R2000 and R4000 AD graph) on $20$ environments. After every $20$ training epochs, the defender evaluates and resets the environments. When the terminating condition is met (number of epochs), the defender chooses the defensive plan with the lowest attacker success rate as their \textit{best environment configuration}. For experimental setup 2 and 3, we train the RL policy for $150$ epochs on $20$ environments parallelly. The trained attacker’s policy is then simulated on the best environment configuration for $5000$ episodes, and the average reward over $5000$ episodes is the attacker’s success rate. Notably, we intensively train the RL based attacker’s policy to obtain a good attacking strategy, as the effectiveness of defensive plans depends on the attacker’s model accuracy.

\noindent \textit{\textbf{Critic network assisted Evolutionary Diversity Optimization:}} The defender can only block $5$ edges. In $20000$ iterations, defender creates a population of $20$ environment configurations (defense). \textcolor{blue}{Table \ref{Ch4_hyperparametertable}} presents the summary of hyper-parameters configurations.

\begin{table}[t!] 
\caption{Summary of hyper-parameters configurations.}
\label{Ch4_hyperparametertable}
\renewcommand{\arraystretch}{1}
\centering
\footnotesize
\begin{tabular}{ll} \hlineB{2} 
\textbf{{Hyper-parameter} }& \textbf{{Value}} \\ \hlineB{2}
Buffer size & 30,000  \\
Learning rate & $1e^{-4}$ \\
Gamma &  0.99\\
Step per epoch &  80,000\\
Step per collect &  1,000\\
Repeat per collect &  10\\
Number of environment configurations & 20 \\
Batch size &  800\\\hlineB{1}
Defender’s budget & 5\\
Population size & 20\\
Number of EC iterations & 20000\\
Mutation, Crossover probability & 0.5\\
Absolute difference   & 0.1\\\hlineB{2}
\end{tabular} 
\end{table}

\subsection{Experimental Setup 1}
In this experimental setup, we determine the effectiveness of our overall proposed approach. 

\subsubsection{\textbf{Baseline}}
We combine the RL based attacker’s policy with various defender’s policies to compare the effectiveness of our proposed defensive policy. 
\begin{itemize}[leftmargin=*]
    \item \textit{\textbf{RL+C-EDO (proposed)}}: In this approach, RL is utilized as attacker’s strategy, whereas C-EDO is used as defender’s policy. In C-EDO, the defender rejects those environment configurations that contribute least to diversity. 
    \item \textit{\textbf{RL+EC}}: In RL+EC approach, RL is utilized as attacker’s policy, and Evolutionary Computation (EC) is used as defender’s policy. In EC, the configurations with the lowest fitness score are discarded.
    \item \textit{\textbf{RL+Greedy}}: In RL+Greedy approach, RL is utilized as attacker’s strategy. The defender uses a greedy technique to generate environment configurations. The defender greedily blocks one edge that minimizes attacker’s success rate. This way, the defender iteratively discovers $k$ edges to be blocked.
\end{itemize}

\begin{table*}[t!] 
\caption[Comparison of attacker’s chances of success under various attacker-defender policies.]{Comparison of attacker's chances of success under various attacker-defender policies (smaller number represents better performance). Results show that the proposed C-EDO defensive policy leads to the best defense.}
\label{R1000}
\renewcommand{\arraystretch}{1.2}
\centering
\resizebox{\textwidth}{!}{\begin{tabular}{p{1.2cm}llllllll} \hlineB{2}
\multicolumn{1}{c}{\textbf{{}}}& \multicolumn{1}{c}{\textbf{{}}}&  \multicolumn{4}{c}{\textbf{{Chances of success}}} & \multicolumn{3}{c}{\textbf{{Time (hour)}}} \\  \cmidrule(lr){3-6}  \cmidrule(lr){7-9}

\textbf{Graph} & \textbf{Policy} & \textbf{I} & \textbf{P} & \textbf{N} &  \textbf{Avg.} &\textbf{I} & \textbf{P} & \textbf{N}\\ \hlineB{2} 
      & RL+C-EDO (Proposed) & \textbf{40.16\%} & \textbf{41.36\%} & \textbf{41.51\% }& \textbf{41.01\%} & 43.82 & 44.73  &  47.27\\
R1000 & RL+EC & 41.69\% & 48.45\% & 42.89\% & 44.34\% & 44.43 & 43.88  & 48.43\\
      & RL+Greedy & 56.45\%  & 47.86\% & 47.88\% & 50.73\% & 43.32 & 44.95 & 47.38 \\\hlineB{2} 
      & RL+C-EDO (Proposed) & \textbf{25.09\%} & \textbf{32.45\%} & \textbf{30.44\%} & \textbf{29.33\%} & 112.74 & 118.58  &  119.53\\
R2000 & RL+EC & 28.13\% & 35.51\% & 37.28\% & 33.64\% & 114.22 &  112.57 & 118.51\\
      & RL+Greedy & 33.43\%  & 34.30\% & 40.06\% & 35.93\% & 113.84 & 116.43 &  113.76\\\hlineB{2} 
      & RL+C-EDO (Proposed) & \textbf{22.02\%} & \textbf{17.29\%} & \textbf{21.81\%} & \textbf{20.37\% }& 127.54 & 121.11  &  137.73\\
R4000 & RL+EC & 22.78\% & 21.87\% & 24.71\% & 23.12\% & 132.55  &  120.67 & 135.31\\
      & RL+Greedy & 25.16\%  & 24.18\% & 22.34\% & 23.89\% & 125.61 & 120.73 & 127.29 \\\hlineB{2} 
\end{tabular}}
\end{table*}

\subsubsection{\textbf{Results}} We perform experiments on $R1000$, $R2000$ and $R4000$ AD graphs. We first train the RL based attacker’s policy on the environment configurations (defensive plan) generated by the defender’s policies, i.e., C-EDO, EC and Greedy. We then test the effectiveness of the attacker’s policy against the defender’s best environmental configurations. We report the average reward (success rate) by simulating the attacker’s strategy on the best environment for $5000$ episodes. For each AD graph, we perform experiments on $5$ seeds from $0$ to $4$, and report the average success rate over $5$ seeds. \textcolor{blue}{Table \ref{R1000}} reports the results obtained for R1000, R2000 and R4000 AD graph. For R1000 AD graphs, results from the table show that under independent distribution, the attacker’s chances of success under C-EDO based defensive policy is 40.16\%. In contrast, the attacker’s chances of success increase to 41.69\% and 56.45\% under EC based and greedy defense, respectively. 
For R2000 AD graphs, our results show that under independent distribution, the attacker chances of success are minimum, i.e., 25.09\% when the defender uses C-EDO based policy; the success rate increases to 28.13\% and 33.43\% under EC based defense and Greedy defense, respectively. Our proposed approach is scalable to R4000 AD graph. The results on R4000 AD graphs show that under a positive correlation, the attacker success rate is minimum, i.e., 17.29\% under C-EDO based defence; however, the success rate increases to 21.87\% and 24.18\% under EC based and Greedy defense, respectively. Overall, our results demonstrate that for all three AD graphs, i.e., R1000, R2000 and R4000, on average C-EDO based defence is the best defense where the attacker success rate is minimum. Also, EC based defense outperforms Greedy defense. {Notably, our proposed RL+C-EDO approach is scalable to R4000 graphs.}

\subsection{Experimental Setup 2}
In this experimental setup, we determine the effectiveness of our proposed RL based attacker’s policy. 

\subsubsection{\textbf{Baseline}}
We compare our proposed RL based attacker’s strategy with our previously proposed Neural Network based Dynamic Program (NNDP) attacker policy \cite{Goel2022defending}. In NNDP approach, we trained neural network to address the attacker’s problem.

\begin{table*}[ht!] 
\caption[Comparison of attacker’s chances of success under various attacker’s policies.]{Comparison of attacker's chances of success under various attacker's policies (larger number represents better performance).}
\label{R500}
\renewcommand{\arraystretch}{1.2}
\centering 
\footnotesize
\begin{tabular}{p{1.5cm}p{3.2cm}p{2cm}p{2cm}p{2cm}p{1.8cm}}\hlineB{2} 
\textbf{Graph} & \textbf{Policy} &\textbf{I} & \textbf{P} & \textbf{N} & \textbf{Avg.}\\ \hlineB{2} 

R500 & RL (Proposed) & \textbf{88.01\%} & \textbf{86.58\%} &  \textbf{89.37\%}& \textbf{87.98\%}\\
 & NNDP & 87.57\% & 86.08\% & 89.28\%  & 87.64\%  \\ \hlineB{2}
R1000 & RL (Proposed) & \textbf{54.99\%} & 48.21\%$^{*}$ & \textbf{52.69\%} & \textbf{51.96\%}\\
 & NNDP & 53.52\% & \textbf{48.32\%} & 52.15\%  & 51.33\%  \\ \hlineB{2} 
R2000 & RL (Proposed) & \textbf{45.28\%}$^{*}$ & \textbf{56.41\%} & \textbf{42.43\%}$^{*}$  & \textbf{48.04\%}\\
 & NNDP & 45.11\% & 56.29\% & 42.39\%  & 47.93\%\\ \hlineB{2} 
\end{tabular}
\end{table*}

\subsubsection{\textbf{Results}}
Our baseline NNDP approach is scalable to R2000 graph; therefore, we perform experiments on $R500$, $R1000$ and $R2000$ AD graphs. We randomly generate $10$ environmental configurations for each AD graph. 
We first train NNDP based attacker’s strategy on $10$ environments for 2000 epochs and perform Monte Carlo simulations for $5000$ runs to compute the attacker’s success rate on each environment. We then train our proposed RL based attacker’s policy on the same set of $10$ environments for 150 epochs and then evaluate the trained policy for $5000$ episodes to compute attacker’s chances of success. We reported attacker’s average chances of success over $10$ environments in \textcolor{blue}{Table \ref{R500}}. For a given environmental configuration, the attacker’s policy that results in a higher success rate indicates that the corresponding policy is able to approximate the attacker’s problem more accurately than others. We opted for a higher number of epochs in NNDP training because RL training is time-consuming per epoch, whereas NNDP training is much faster. This allowed us to balance the overall training time between the two models.

\textcolor{blue}{Table \ref{R500}} shows that for R500 graph, the attacker average success rate is 87.98\% under the RL policy, which is slightly higher than NNDP based policy. For the R1000 graph, attacker’s average chance of success is 51.96\%, which is again higher than the NNDP policy. Our results show that the under our proposed RL based strategy, the attacker success rate is higher compared to NNDP based strategy, implying that RL policy is more effective at countering defense than NNDP policy.

\subsection{Experimental Setup 3}
In this experimental setup, we determine the effectiveness of our proposed C-EDO based defense\blfootnote{* indicates that with our general parameter settings, RL policy results were slightly bad than baseline. Therefore, we train the RL policy for 300 epochs instead of 150 epochs. Given enough time, the RL policy outperforms the baseline.}. 

\subsubsection{\textbf{Baseline}}
We compare our proposed approach’s final environmental configuration (defensive plan) with the final configuration from NNDP-EDO approach. 

\begin{table*}[t!] 
\caption[Comparison of attacker’s success rate on best defense from various attacker-defender policies.]{Comparison of attacker's chances of success on best defense from various attacker-defender policies (smaller number represents better performance).}
\label{e3_setting}
\renewcommand{\arraystretch}{1.2}
\centering
\footnotesize
\resizebox{\textwidth}{!}{\begin{tabular}{p{1.2cm}lllll}\hlineB{2} 

\textbf{Graph} & \textbf{Policy} &\textbf{I} & \textbf{P} & \textbf{N} & \textbf{Avg.}\\ \hlineB{2} 
R1000 & Best defense from RL+C-EDO (Proposed) & \textbf{40.16\%} & \textbf{41.36\%} & \textbf{41.51\%} & \textbf{ 41.01\%}\\
 & Best defense from NNDP-EDO & 42.02\% & 44.76\% & 41.53\%$^{*}$ & 42.77\%  \\ \hlineB{2} 
R2000 & Best defense from RL+C-EDO (Proposed) & \textbf{25.09\%} & 32.45\% &  \textbf{30.44\%} & \textbf{29.32\%}\\
 & Best defense from NNDP-EDO & 30.31\% & \textbf{30.17\%}$^{*}$ & 30.85\%  & 30.44\%  \\ \hlineB{2} 
\end{tabular}}
\end{table*}

\subsubsection{\textbf{Results}} We run RL+C-EDO and NNDP-EDO approaches on 5 seeds from $0$ to $4$ to obtain the defender’s best environment. We train the RL attacker policy for 150 epochs on the best environmental configurations from both approaches (on 5 seeds). We then evaluate the trained policy for 5000 episodes to compute the attacker success rate. We reported the results in \textcolor{blue}{Table \ref{e3_setting}}. An environmental configuration against which the RL based attacker policy is able to achieve a lower success rate is considered as the best environmental configuration. Results in \textcolor{blue}{Table \ref{e3_setting}} show that the average attacker success rate for R1000 AD graph is 42.77\% on the best configuration from NNDP-EDO. However, the attacker success rate is 41.01\% on the best configuration from RL+C-EDO, which is 1.76\%  less than the former approach. Similarly, for R2000 AD graph, the attacker success rate is minimal under RL+C-EDO based defensive plan. The results demonstrate that our approach RL+C-EDO is able to generate better environmental configurations and minimizes the attacker’s success rate.


\section{Chapter Summary}\label{RL_conclusion}
We studied a Stackelberg game model in a configurable environment, where the attacker’s goal is to devise a strategy to maximize their achievable rewards. The defender seeks to identify the environment configuration where the attacker’s attainable reward is minimum. We proposed a reinforcement learning based approach to address the attacker problem and critic network assisted evolutionary diversity optimization based policy to address the defender problem. We trained the attacker policy against numerous environments simultaneously. We leverage the trained RL critic network to evaluate the fitness of the environment configurations. 

Overall, we proposed two approaches for defending AD graphs; first is neural network based dynamic program and evolutionary diversity optimization approach in \textcolor{blue}{Chapter \ref{Chapter_AD_NNDP}} and second is reinforcement learning and evolutionary diversity optimization based approach in \textcolor{blue}{Chapter \ref{AD_RL}}. Our experimental results showed that the proposed reinforcement learning and evolutionary diversity optimization based approach is more effective and scalable than our proposed neural network based dynamic program approach.
\chapter{Effective Graph-Neural-Network based Models for Discovering Structural Hole Spanners in Large-Scale and Diverse Networks} 

\label{SHS_GNN} 

\textbf{\underline{Related publication:}} 
\vspace{0.06in}

\noindent This chapter is based on our paper titled “\textit{Effective Graph-Neural-Network based Models for Discovering Structural Hole Spanners in Large-Scale and Diverse Networks}” \cite{goel2023effective}.

\vspace{0.1in}

\noindent Structural hole spanners are the bottleneck nodes essential for information diffusion in the network. Although numerous solutions have been developed to discover SHS nodes in the network; however, these solutions require high run time on large-scale networks. Another limitation is discovering SHSs across different networks, for which conventional approaches fail to work. To address these limitations, this chapter aims to design effective and efficient solutions for discovering SHS nodes in large-scale and diverse networks. We transform the problem of discovering SHSs into a learning problem and propose an efficient GraphSHS model that exploits both the network structure and node features to discover SHS nodes in large-scale networks, endeavouring to lessen the computational cost while maintaining high accuracy. To effectively discover SHSs across diverse networks, we propose another model, Meta-GraphSHS, based on meta-learning that learns generalizable knowledge from diverse training graphs and utilizes the learned knowledge to create a customized model to identify SHSs in each new graph. We evaluate the performance of our proposed models through extensive experimentation on synthetic and real-world datasets, and our results demonstrate that the proposed models are highly effective and efficient.

\section{Introduction}
In recent years, various large-scale networks have emerged, such as biological, collaboration and social networks. These networks exhibit a community structure where the nodes within the community share similar characteristics, behaviour, and opinions \cite{chen2019contextual}. The absence of connection between different communities in a network is referred to as Structural Holes (SH) \cite{burt2009structural}, and the presence of SHs in the network restricts the flow of novel information \cite{rinia2001citation} between communities. The individuals who bridge multiple communities obtain considerable benefits in the network over others who belong to one community only \cite{lou2013mining}, and these individuals are known as \textbf{\textit{Structural Hole Spanners (SHSs)}}. SHSs have various real-world applications, such as information diffusion, identifying central nodes and discovering communities \cite{lin2021efficient, amelkin2019fighting, zareie2019influential, yu2021modeling, zhao2021community, ahmad2023review}. A number of centrality measures such as Closeness Centrality \cite{rezvani2015identifying}, Constraint \cite{burt1992structural}, Betweenness Centrality (BC) \cite{freeman1977set} exist in the literature to define SHSs. SHS nodes lie on the maximum number of shortest paths between the communities \cite{rezvani2015identifying}; removal of the SHS nodes will disconnect multiple communities and block information flow among the nodes of the communities \cite{lou2013mining}. Based on this, we have two implications about the properties of SHSs; 1) SHSs bridge multiple communities; 2) SHSs control information diffusion in the network. \textcolor{blue}{Figure \ref{fig:cent}} illustrates the comparison of various centrality measures in a network. The figure shows that node $i$ holds a vital position in the network, and the shortest path between the nodes of three communities passing-through node $i$, and removing node $i$ will block the information propagation between the nodes of these communities. In contrast, the impact of the removal of other nodes is comparatively less significant. Since the removal of a node with the highest betweenness centrality disconnects the maximum number of communities and blocks information propagation between the nodes of the communities, therefore, we adopt the betweenness centrality measure for defining SHSs in the network. Goyal et al. \cite{goyal2007structural} defined the node that lies on a large number of shortest paths as SHS, which is similar to the betweenness centrality. BC quantifies a node’s control on the information flow in the network and discovers those nodes that act as a bridge between different communities. 
Brandes algorithm \cite{brandes2001faster} is the best-known method for calculating the BC scores of the nodes and has a run time of $\mathcal{O}(nm)$.

\begin{figure*}[t!]
 \centering
 \includegraphics[width=0.6\paperwidth]{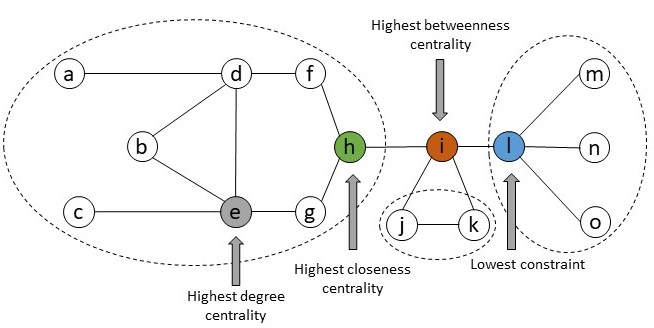}
 \caption{Comparison of various centrality measures.}
 \label{fig:cent}
\end{figure*}

\noindent \textbf{\textit{Challenges}:} Several studies \cite{lou2013mining, he2016joint, xu2019identifying, li2019distributed} have been conducted for discovering SHSs in the network. Lou et al. \cite{lou2013mining} developed an algorithm for finding SHSs by assuming that community information is given in advance. However, discovering communities in a large network is a challenging task. He et al. \cite{he2016joint} designed a harmonic modularity solution that discovers both SHSs and communities in the network. The authors assume that every node belongs to one community only, but a node may belong to many communities in the real world. Although there are numerous solutions that address the SHSs identification problem; however, there are still challenges that need to be addressed, such as:

\begin{enumerate}
    \item \textit{{Discovering SHS nodes efficiently in large scale networks:}} For small networks, we can discover SHSs by computing the BC score of the nodes using Brandes algorithm \cite{brandes2001faster}. However, real-world networks are much larger, with millions of nodes and edges, making it challenging to perform tasks like SHS or community detection efficiently. The Brandes algorithm's time complexity of $O(nm)$ becomes a significant bottleneck in such cases, where we need quick solutions for very large networks. Therefore, we need efficient solutions for discovering SHSs in large networks.

    \item \textit{{Discovering SHS nodes effectively in diverse networks:}} For discovering SHSs in different types of network, traditional learning techniques fail to work because their one-model-fit-all approach neglect the inter-graph differences, especially when the graphs belong to diverse domains. Besides, re-training the model again on different types of large networks is a time-consuming process. Therefore, it is crucial to have a model which is aware of differences across the graphs and customizes accordingly, avoiding the requirement of re-training the model on every type of network individually.
\end{enumerate}

To address the challenges mentioned above and inspired by the recent advancements of Graph Neural Network (GNN), we propose message-passing GNN based models to discover SHS nodes. GNNs \cite{thekumparampil2018attention, kipf2016semi} are Neural Network architectures designed for graph structured data. GNNs are used as graph representation learning models \cite{joshi2019efficient} and learn node representations by aggregating feature information from the local graph neighbourhood \cite{djenouri2022hybrid}. GNNs have shown exceptional results on various graph mining problems \cite{horta2021extracting, ji2021temporal}; therefore, we investigate the power of GNNs for solving SHS identification problem. In this chapter, we aim to discover SHS nodes in large-scale networks, endeavouring to reduce the computational cost while maintaining high accuracy, and in different types of networks effectively without the need of re-training the model on individual network datasets to adapt to cross-network property changes. We transform the SHS discovery problem into a \textit{learning problem} and propose two GNN based models, GraphSHS and Meta-GraphSHS. In order to address the first challenge mentioned above, we propose \textbf{\textit{GraphSHS}} \textit{\textbf{(\underline{Graph} neural network for \underline{S}tructural \underline{H}ole \underline{S}panners)}}, a graph neural network-based model for efficiently discovering SHSs in large scale networks. GraphSHS exploits both the network structure and features of nodes to learn the low-dimensional node embeddings. 
In addition, unlike traditional Deep Learning approaches that assume a transductive setting, GraphSHS assumes an inductive setting. GraphSHS is generalizable to new nodes of the same graph or even to the new graphs from the same network domain. Our experimental results demonstrate that the idea of designing graph neural network based model to discover SHSs in large-scale networks provides a \textit{{significant run time advantage}} over other algorithms. Apart from the run time efficiency, GraphSHS achieves competitive or better accuracy in most cases than the baseline algorithms. To address the second challenge, we propose \textbf{\textit{Meta-GraphSHS}} \textit{\textbf{(\underline{Meta-}learning based \underline{Graph} neural network for \underline{S}tructural \underline{H}ole \underline{S}panners)}} to effectively discover SHSs across diverse networks. In the case of diverse graphs, there exist inter-graph differences due to which GraphSHS can not effectively discover SHSs across diverse networks. Therefore, instead of directly learning the model, we learn the generalizable knowledge (parameters) from diverse training graphs and utilize the learned knowledge to create a customized model by fine-tuning the parameters according to each new graph\footnote{The generalizable knowledge act as a good initialization point (good set of parameters) for the new customized model. The generalized parameters are fine-tuned using the labelled nodes of new unseen graphs.}. Meta-GraphSHS uses meta-learning \cite{vilalta2002perspective} to learn generalizable parameters from the training graphs that are different from the testing graphs we are considering, and the goal is to reach an {\textit{“almost trained”}} model that can be quickly adapted to create a customized model for the new graph under consideration within a few gradient steps. The goal of Meta-GraphSHS is to observe many graphs from different domains and use the learned knowledge to identify SHS on any new graphs,\textit{ enabling quick adaptation and higher accuracy}. Once our proposed model is trained, it can be applied repeatedly for future arriving data; therefore, we consider primarily the run time of applying the model and regard the training process is done offline, as the common practice in machine learning literature. 

Our experimental results show that both the proposed graph neural network models GraphSHS and Meta-GraphSHS are highly efficient and effective in discovering SHSs in large scale networks and diverse networks, respectively. We evaluate the performance of GraphSHS on synthetic datasets, and the results show that GraphSHS is at least 58 times faster than baselines and achieves higher or competitive accuracy than baselines. In addition, GraphSHS is at least 167.1 times faster than the baselines on real-world networks, illustrating the efficiency advantage of the proposed GraphSHS model on large-scale networks. We evaluate the performance of Meta-GraphSHS on a diverse set of synthetic and real-world graphs, and the results show that Meta-GraphSHS identifies SHSs with high accuracy, i.e., 96.2\% on synthetic graphs and outperforms GraphSHS by 2.7\% accuracy, demonstrating the importance of designing separate model for discovering SHSs in diverse networks. Additionally, we also conduct parameter sensitivity analysis to analyze the impact of parameters on the performance of proposed models. In order to determine the applicability of the proposed model GraphSHS in the dynamic network, we perform experiments on synthetic graphs and found that our model is at least 89.8 times faster than the existing baseline.

\noindent \textbf{\textit{Contributions:}} In this chapter, we make the following contributions.
\begin{itemize}
    \item \textbf{GraphSHS model.} We propose an efficient graph neural network-based model GraphSHS that discovers SHSs in large scale networks and achieves considerable efficiency advantage while maintaining high accuracy compared to existing baselines. 
    
    \item \textbf{Meta-GraphSHS model.} We propose another model Meta-GraphSHS that combines meta-learning with graph neural network to discover SHS nodes across diverse networks effectively. This model learns a generalized knowledge from diverse graphs that can be utilized to create a customized inductive model for each new graph, in turn avoiding the requirement of repeated model training on every type of diverse graph.
    
    \item \textbf{Inductive setting.} We use an inductive setting, where our GraphSHS model is generalizable to new nodes of the same graph or even to the new graphs from the same network. In addition, the proposed Meta-GraphSHS model is generalizable to unseen graphs from diverse networks.
    
    \item \textbf{Theoretical analysis.} We theoretically show that our message-passing architecture of GraphSHS is sufficient to solve the SHSs identification problem under sufficient conditions on its node attributes, expressiveness of layer, architecture’s depth and width. In addition, we show that the depth of the model should be at least  $\Omega(\sqrt{n}/\log n)$ to accurately solve the SHSs identification problem.
    
    \item \textbf{Extensive experiments.} We conduct extensive experiments on synthetic networks and real-world networks of varying scales. The results show that the proposed model GraphSHS is at least 167.1 times faster than the baselines on real-world networks and at least 58 times faster on synthetic networks. In addition, Meta-GraphSHS discovers SHSs across diverse networks with an accuracy of 96.2\%.
    
\end{itemize}

\noindent \textbf{\textit{Chapter organization}:} \textcolor{blue}{Section \ref{c6_Preliminaries}} discusses the preliminaries and background of the problem. \textcolor{blue}{Section \ref{c6_Problem}} presents the problem description. \textcolor{blue}{Section \ref{c6_model}} discusses the proposed GraphSHS model that discovers SHSs in large-scale networks and \textcolor{blue}{Section \ref{c6_Meta_model}} describes the  Meta-GraphSHS model that discovers SHSs in diverse networks. \textcolor{blue}{Section \ref{c6_results}} reports and discusses the experimental results. Finally, \textcolor{blue}{Section \ref{c6_summary}} concludes the chapter.

\section{Preliminaries and Background}\label{c6_Preliminaries}
\noindent This section discusses the preliminaries and background of the problem.
\vspace{0.15in}

\noindent \textbf{\textit{Notations.}} A network can be represented as an undirected graph $G = (V, E)$, where $V$ is the set of nodes (users), and $E$ is the set of edges (the relationship between users). Let $n=|V|$ and $m=|E|$. We use $\vec{x}(i)$ to represent the feature vector of node $i$ and $h^{(l)}(i)$  to denote the embedding of node $i$ at the $l^{th}$ layer of the model, where $l = (1,2,...,L)$. The neighbors of node $i$ are represented by $N(i)$, and the degree of node $i$ is represented by $d(i)$.

\noindent \textbf{\textit{Graph Neural Networks.}}
Graph Neural Networks (GNNs) are designed by extending Deep Learning approaches for the graph-structured data and are used in diverse fields, including computer vision, graph problems etc. GNNs are used to learn the graph data representations. Motivated by the success of Convolution Neural Network, various Graph Neural Network architectures are designed. One such architecture is Graph Convolutional Network \cite{kipf2016semi}, which uses an aggregation mechanism similar to the mean pooling. Graph Attention Network \cite{velivckovic2017graph} is another Graph Neural Network architecture that uses an attention mechanism for aggregating features from the neighbors. Existing GNN architectures mostly follow message-passing mechanism. These GNNs execute graph convolution by aggregating features from neighbours, and stacking many layers of GNN to capture far-off node dependencies.

\begin{figure}[ht!]
 \centering
 \includegraphics[width=0.4\paperwidth]{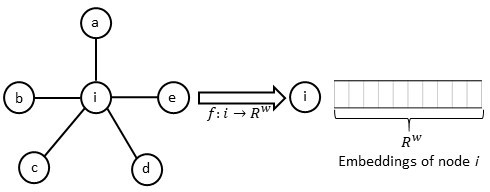}
 \caption{Network embedding: embedding of node $i$ in $R^w$ space.}
 \label{fig:embed}
\end{figure}

\noindent \textbf{\textit{Network Embedding.}} Network embedding \cite{cui2018survey} is a mechanism that maps the nodes of the network to a low-dimensional vector representation. It aims to encode the nodes in such a way that the resemblance in the embedding space approximates the resemblance in the network \cite{aguilar2021novel}. These embeddings can then be utilized for various graph problems such as classification, regression etc. \textcolor{blue}{Figure \ref{fig:embed}} illustrates an example of node embedding.

\noindent \textbf{\textit{Meta-Learning.}}
Meta-Learning 
aims to learn efficiently and generalize the learned knowledge to the new tasks. There are various meta-learning approaches such as black-box methods \cite{andrychowicz2016learning}, gradient-based methods \cite{finn2017model} and non-parametric learning methods \cite{chen2019closer}. Meta-Learning assumes that the prior learned knowledge is transferable among the tasks. The model trained on the training tasks can be adjusted to the new task using a small amount of labelled data or in the absence of any supervised knowledge. Meta-learning significantly improves the performance of the tasks that suffers from data deficiency problem. It learns the shared learning from the various tasks and adapts this knowledge to the unseen tasks, speeding up the learning process on new tasks.

\section{Problem Description} \label{c6_Problem}
In this section, we formally state the structural hole spanner problem for large-scale networks and diverse networks.

\begin{definition}
\noindent \textbf{Betweenness Centrality.} The betweenness centrality $BC(v)$ of a node $v \in V$ is defined as \cite{freeman1977set}:
\begin{equation}
\begin{aligned}
BC(v) = \sum_{\substack{s\neq v\neq t\\v \in V}}{\frac{\sigma_{st}{(v)}}{\sigma_{st}}},
\end{aligned}
\end{equation}
\end{definition}
\noindent where $\sigma_{st}$ denotes the total number of shortest paths from node $s$ to $t$ and $\sigma_{st}{(v)}$ denotes the number of shortest paths from node $s$ to $t$ that pass through node $v$. We will use the term SHS score of a node and BC of a node interchangeably. \textit{We label $k$ nodes with the highest BC in the graph as \textit{Structural Hole Spanner nodes} and the rest as normal nodes.}

In theory, the computation of Betweenness Centrality (discovering SHSs) is tractable as polynomial-time solutions exist; however, in practice, the solutions are computationally expensive. Currently, Brandes algorithm \cite{brandes2001faster} is the best-known technique for calculating the BC of the nodes with a run time of $\mathcal{O}(nm)$. However, this run time is not practically applicable, considering that even mid-size networks may have tens of thousands of edges. Computing the exact BC for a large scale network is not practically possible with traditional algorithms; consequently, we convert the SHS identification problem into a learning problem and then solve the problem. We formally define both the structural hole spanner discovering problems as follows:

\noindent \textbf{{Problem 1:}} \textbf{Discover SHS nodes in large scale networks.}

\noindent \textbf{\textit{Given:}} Training graph $G_{train}$, features and labels\footnote{Label of a node can either be SHS or normal.} of nodes in $G_{train}$, and test graph $G_{test}$.

\noindent \textbf{\textit{Goal:}} Design an inductive model GraphSHS (by training the model on $G_{train}$) to discover SHSs in new unseen large scale graph $G_{test}$. GraphSHS aims to achieve a considerable efficiency advantage while maintaining high accuracy. 

\noindent \textbf{{Problem 2:}} \textbf{Discover SHS nodes in diverse networks.}

\noindent \textbf{\textit{Given:}} A set of training graphs $G_{train}=\{G_1, G_2, . . . , G_M\}$ from diverse domains, features and labels of nodes in $G_{train}$ and test graph $G_{test}$ in which the nodes are partially labeled.

\noindent \textbf{\textit{Goal:}} Design a model Meta-GraphSHS to discover SHS nodes across diverse networks effectively by learning generalized knowledge from diverse training graphs $G_{train}$. The generalized knowledge (parameters) is fine-tuned using labelled nodes from $G_{test}$ in order to obtain updated parameters that can be used to discover SHSs in $G_{test}$.
\vspace{0.15in}

We address the above-discussed two problems by transforming them into learning problems and proposing two message-passing GNN-based models. Once the models are trained, the inductive setting of the models enables them to discover SHS nodes. The identified $k$ SHSs are the nodes with the highest SHS score (BC) in the network.

\section[Proposed Model: GraphSHS]{Proposed Model: GraphSHS - Discovering SHSs in Large-Scale Networks}\label{c6_model}

In this section, we discuss our proposed message-passing graph neural network-based model {\textit{GraphSHS}} that aims to discover SHS nodes in large scale networks. We first discuss the network features that we extracted to characterize each node. We then discuss the proposed model GraphSHS in detail. To discover SHSs, GraphSHS first maps each node to an embedding vector (low dimensional node representation) using the aggregation mechanism. GraphSHS then uses the embedding vector of each node to determine the labels of the node. \textcolor{blue}{Figure \ref{fig:gnn_arch}}  illustrates the overall architecture of the proposed model GraphSHS. The network features, aggregation mechanism and the training procedure of GraphSHS model are discussed below. 

\subsection{Network Features} \label{features}
\begin{definition}

\textbf{\textit{r}-ego network}. The \textit{r}-ego network of a node $v \in V$ is the subgraph induced from $N_{r}(v)$ where $N_{r}(v) = \{u :dist_{uv}^{G} \leq r \}$ is $v's$ $r$-hop neighbors and $dist_{uv}^{G}$ denotes the distance between node $u$ and $v$ in graph $G$ \cite{qiu2020gcc}.
\end{definition}

\noindent We use three network features; effective size, efficiency and degree computed from the one-hop ego network of each node to characterize the node.

\noindent \textit{\textbf{Effective Size.}} The effective size is a measure of non-redundant neighbors of a node \cite{burt1992structural}. Effective size determines the extent to which neighbor $j$ is redundant with the other neighbors of node $i$.

\noindent \textit{\textbf{Efficiency.}} The efficiency is the ratio of the effective size of ego network of the node to its actual size \cite{burt1992structural}.

\noindent \textit{\textbf{Degree.}} The degree of a node is the number of connections it has with the other nodes of the network.


\subsection{Aggregation Mechanism} 
Our proposed aggregation mechanism computes the low dimensional node embeddings in two phases: 1) Neighborhood aggregation phase, where a node aggregates embeddings from its neighbors; 2) Combine function phase, where a node combines its own embedding to the aggregated neighbors embeddings. The procedure for generating embeddings of the nodes is presented in \textcolor{blue}{Algorithm \ref{gnn-algo}}.

\begin{algorithm}[ht!]
\caption{Generating node embedding using GraphSHS}
 \label{gnn-algo}
 \begin{algorithmic}[1]
 \renewcommand{\algorithmicrequire}{\textbf{Input:}}
 \renewcommand{\algorithmicensure}{\textbf{Output:}}
 \REQUIRE Graph: $G(V,E)$, Input features: $\vec{x}(i),\,\, \forall i \in V$, Depth: $L$, Weight matrices: $W^{(l)},\,\, \forall l \in \{1,..,L\}$, Non-linearity: $\sigma$
 \ENSURE Node embedding: $z{(i)}, \,\, \forall i \in V$ 
 \STATE $h^{(0)}(i) \leftarrow \vec{x}(i),\,\,\forall i \in V$ 
 \FOR{$l = 1$ to $L$}
 \FOR{$i \in V$}
 \STATE Compute $h^{(l)}{(N(i))}$ using \textcolor{blue}{Equation} \ref{agg}
 \STATE Compute $h^{(l)}{(i)}$ using \textcolor{blue}{Equation} \ref{comb}
 \ENDFOR
 \ENDFOR
 \STATE $z{(i)} = h^{(L)}{(i)}$
\end{algorithmic}
\end{algorithm}

\noindent \textbf{Neighborhood Aggregation.}
For generating the node embeddings, GraphSHS first performs neighborhood aggregation by capturing feature information (embeddings) from the neighbors of the node. This process is similar to the \textit{message passing mechanism} of GNNs. Due to the distinctive properties exhibited by the SHS node (i.e., the SHS node act as a bridge, and its removal disconnects the network), we aggregate embeddings from all one-hop neighbors of the node. We describe the neighborhood aggregation as a weighted sum of embedding vectors and is given by:
\begin{equation}
\label{agg}
h^{(l)}{(N(i))} = \sum_{j\in N(i)}{\frac{h^{(l-1)}{(j)}}{d(i)}},
\end{equation}

\noindent where $h^{(l)}{(N(i))}$ denotes the embedding vectors aggregated from the neighbors $N(i)$ of node $i$ at the $l^{th}$ layer. The aggregated embedding from the neighbors of node $i$ is used to update node $ i$’s embeddings. During the aggregation process, we utilize the degree $d$ of the node as a weight. We use the features of the nodes (as discussed in \textcolor{blue}{Section \ref{features}}) to compute the initial embedding $h^{(0)}$ of the nodes. Let $\vec{x}(i)$ represents the feature vector of node $i$; GraphSHS initialize the initial embedding of node $i$ as:
\begin{equation}
h^{(0)}(i) = \vec{x}(i). 
\end{equation}
Therefore, given a network structure and initial node features, neighborhood aggregation phase computes the embedding of each node by aggregating features from the neighbors of the nodes.

\begin{figure*}[t!]
 \centering
 \includegraphics[width=0.72\paperwidth]{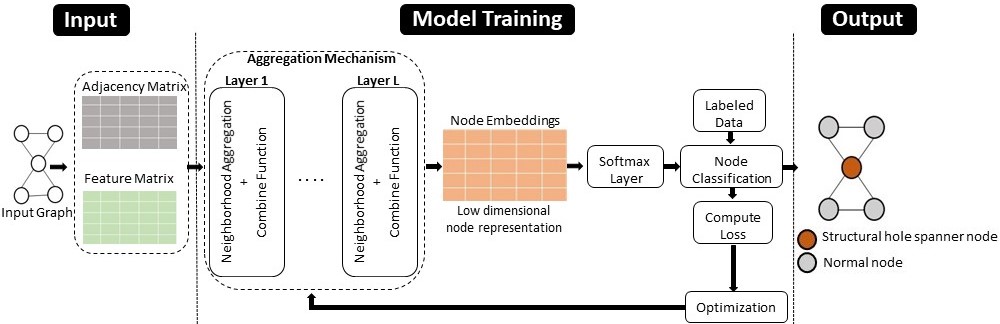}
 \caption{Architecture of proposed GraphSHS model.}
 \label{fig:gnn_arch}
\end{figure*}

\noindent \textbf{Combine Function.}
In the neighborhood aggregation phase, we describe the representation of a node in terms of its neighbors. Moreover, to retain the knowledge of each node’s original features, we propose to use the combine function. Combine function concatenates the aggregated embeddings of the neighbors from the current layer with the self-embedding of the node from the previous layer and is given by:
\begin{equation}
\label{comb}
h^{(l)}{(i)} = \sigma{\Bigg(W^{(l)}\bigg(h^{(l-1)}{(i)} \mathbin\Vert h^{(l)}{\Big(N(i)\Big)}\bigg) \Bigg)},
\end{equation}

\noindent where $h^{(l-1)}{(i)}$ represents embedding of node $i$ from layer $(l-1)$ and $h^{(l)}{(N(i))}$ represents aggregated embedding of the neighbors of node $i$. $W^{(l)}$ is the trainable parameters, $\mathbin\Vert$ denotes the concatenation operator, and $\sigma$ represents the non-linearity ReLU.

\noindent \textbf{High Order Propagation.}
GraphSHS stacks multiple layers (Neighborhood Aggregation phase and Combine Function phase) to capture information from the $l$-hop neighbors of a node. The output of layer $(l-1)$ acts as an input for layer $l$, whereas the embeddings at layer $0$ are initialized with the initial features of the nodes. Stacking $l$ layers will recursively formulate the embeddings $h^{(l)}(i)$ for node $i$ at the end of $l^{th}$ layer as:
\begin{equation}
z{(i)} = h^{(l)}{(i)}, \,\,\,\, \forall i \in V
\end{equation}

\noindent where $z(i)$ denotes the final embedding of node $i$ at the end of $l^{th}$ layer ($l$ = $1,...,L$). For the purpose of node classification, we pass the final embeddings $z(i)$ of all the nodes through the Softmax Layer. The softmax layer maps the embeddings of the nodes to the probabilities of two classes, i.e., SHS and normal node. The model is then supervised to learn to differentiate between SHS and normal nodes using the labelled data available.

\subsection{Model Training} 
In order to differentiate between SHSs and normal nodes, we train GraphSHS using \textit{Binary Cross-Entropy Loss} with the actual labels known for a set of nodes. The loss function $\mathcal{L}$ is computed as:
\begin{equation}
\label{eq_loss}
\mathcal{L}(\theta) = {-}\frac{1}{t}{\sum_{i=1}^{t}\bigg(y(i)\log{\hat{y}(i)} + (1-y(i)) \log{(1-\hat{y}(i))\bigg)}},
\end{equation}
\noindent where $y$ is the actual label of a node and $\hat{y}$ is the label predicted by GraphSHS, $t$ is the number of nodes in the training data for which the labels are known, and $\theta$ are the set of model parameters.

\begin{theorem}[\textbf{Loukas \protect{\cite{loukas2019graph}}}] \label{thm1}
\textbf{A simple message passing architecture of GraphSHS is sufficient to solve the SHSs discovery problem if it satisfies the following conditions: each node in the graph is distinctively identified; functions (Neighborhood aggregation and Combine function) computed within each layer $l$ are Turing-Complete; the architecture is deep enough, and the width is unbounded.}
\end{theorem}

Here, depth indicates the number of layers in the architecture and width is the number of hidden units. Simple message passing graph neural networks are proven to be \textit{universal} if the four conditions mentioned above are satisfied \cite{loukas2019graph}. Therefore, we adopt a simple message passing graph neural network architecture to solve the SHSs discovery problem, and our architecture satisfies these conditions. We believe that the universal characteristic of graph neural networks enables our model to discover SHS nodes with high accuracy. This argument is confirmed by our experimental results, as reported in Section 5. Notably, we choose not to include the unique identifiers (node ids) in our node features as SHSs are {\em equivariant} to node permutation. In other words, we can interpret our graph neural network GraphSHS as a function that maps a graph with $n$ nodes to an output vector of size $n$, where the $i^{th}$ coordinate of the output specifies whether node $i$ is a SHS or not. Since any permutation on the graph nodes would also permute the output exactly in the same way, and thus, what our model is trying to learn is an \textit{equivariant function} (by definition). Keriven et al. \cite{keriven2019universal} proposed a simple graph neural network architecture that does not require unique identifiers and shows that the network is a universal approximator of equivariant functions.
It should be noted that the theoretical results of Keriven et al. \cite{keriven2019universal} do not apply directly to message-passing graph neural networks that are more often used in practice. We do not have proof that unique identifiers are not necessary for our model, as we are using message-passing graph neural networks. We do not include unique identifiers as a design choice.

\begin{figure}[ht!]
 \centering
 \includegraphics[width=0.6\columnwidth]{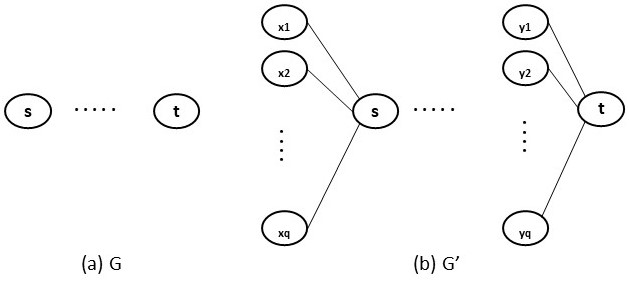}
 \caption[Graphs instances to derive depth of GraphSHS model.]{$G'$ is derived from $G$.}
 \label{fig:hard}
\end{figure}

\begin{theorem} \label{thm2}
\textbf{To calculate the SHSs discovery problem (discovering high betweenness centrality nodes), the depth of GraphSHS (with constant width) should be at least $\Omega(\sqrt{n}/\log n)$.}
\end{theorem}

\begin{proof} Let $G=(V,E)$ be an instance of shortest $s$-$t$ path problem \cite{loukas2019graph} in an undirected graph with source node $s$, destination node $t$ and $|V|=n$, as shown in \textcolor{blue}{Figure \ref{fig:hard}(a)}. The shortest $s$-$t$ path problem aims to find the nodes that lie on the shortest path from node $s$ to $t$. We construct an instance of discovering high betweenness centrality nodes (SHSs) problem in another undirected graph $G'$ from $G$, as illustrated in \textcolor{blue}{Figure \ref{fig:hard}(b)}. 
We add a set of nodes $X = \{x1, x2, ...,xq\}$ which are connected to node $s$ via undirected edges $\{x1\mathdash s, x2\mathdash s, ..., xq\mathdash s\}$. Similarly, we add another set of nodes $Y = \{y1, y2, ...,yq\}$  which are connected to node $t$ via undirected edges $\{y1\mathdash t, y2\mathdash t, ..., yq \mathdash t\}$. Our goal is to discover high betweenness centrality nodes (SHSs) in graph $G'$. 

\vspace{0.1in}
Let us assume that the value of $q$ is $cn$, where $c$ is a constant $\geq 3$ . For computation, we assume $q=3n$; then, for every node that lies on the shortest $s$-$t$ path, there are $9n^2$ shortest paths that go through these nodes. For the rest of the nodes that do not lie on the shortest $s$-$t$ path, the shortest paths in $G'$ that go through these nodes are:
\begin{enumerate}[noitemsep,topsep=0pt]
    \item The shortest paths between the nodes of the original graph $G$ in $G'$. For this case, there are at most $n^2$ shortest paths passing through the nodes that do not lie on the shortest $s$-$t$ path in $G'$.
    \item The shortest paths between the nodes of set $X$ to the nodes of the original graph $G$ in $G'$. For this case, there are at most $(3n \times n)$, i.e., $3n^2$ shortest paths passing through the nodes that do not lie on the shortest $s$-$t$ path.
    \item The shortest paths between the nodes of set $Y$ to the nodes of the original graph $G$ in $G'$. There are at most $(3n \times n)$, i.e., $3n^2$ shortest paths passing through the nodes that do not lie on the shortest $s$-$t$ path.
\end{enumerate}
\vspace{0.1in}

There are $9n^2$ shortest paths going through the nodes that lie on the shortest $s$-$t$ path, which is greater than the total number of shortest paths, i.e., at most $(n^2 + 3n^2 + 3n^2)$ going through the nodes that do not lie on the shortest $s$-$t$ path, i.e.,  $(n^2 + 3n^2 + 3n^2 \leq 9n^2)$. According to the definition of betweenness centrality, a node would have a high betweenness centrality if it appears on many shortest paths. Our analysis shows that more number of shortest paths go through those nodes that lie on the shortest $s$-$t$ path; therefore, high betweenness centrality nodes must also lie on the shortest $s$-$t$ path. In this way, if we can find the high betweenness centrality nodes (SHSs) in the graph, then we can solve the shortest $s$-$t$ path problem. {Corollary 4.3} of \cite{loukas2019graph} already showed that for approximating (to a constant factor) the shortest $s$-$t$ path problem, a message-passing graph neural network must have a depth that is at least $\Omega(\sqrt{n}/\log n)$ assuming constant model width. Hence, this depth lower bound also applies to our SHSs discovery problem.
\end{proof}

\subsection{Complexity Analysis}
\textit{\textbf{Training time.}} To train the GraphSHS model on a network of 5000 nodes, the convergence time is around 15 minutes, which includes the time to compute the ground truth labels and features of the nodes for the training graph. Notably, we train the model only once and then utilize the trained model to predict the nodes’ labels for any input graph.
\vspace{0.1in}

\noindent \textit{\textbf{Inference complexity.}} In the application step of GraphSHS, we apply the trained GraphSHS model to a given network for discovering SHSs. To determine the labels of the nodes, the model computes embeddings for each node. \textcolor{blue}{Algorithm \ref{gnn-algo}} shows that computing the nodes’ embedding takes $\mathcal{O}(LnN)$ time, where $L$ is the depth (number of layers) of the network, $n$ is the number of nodes, and $N$ is the average number of node neighbors. In practice, adjacency matrix multiplication is used for Line 3-6 in \textcolor{blue}{Algorithm \ref{gnn-algo}}, and if the graph is sparsely connected, then the complexity for Line 3-6 is $\mathcal{O}(m)$. Theoretically, we showed that the lower bound on depth $L$ is $\Omega(\sqrt{n}/\log n)$; therefore, the \textit{theoretical lower-bound complexity} for application step of GraphSHS is $\mathcal{O}(m \sqrt{n}/\log n)$. On the other hand, we experimentally showed that the depth $L$ of
the GraphSHS is a small constant $(L = 4)$, and most of the
real-world networks are sparse; therefore, the practical time
complexity for the application step of GraphSHS turns out to be
$O(m)$, i.e., linear in the number of edges.

\section[Proposed Model: Meta-GraphSHS]{Proposed Model: Meta-GraphSHS - Discovering SHSs in Diverse Networks}\label{c6_Meta_model}

In this section, we discuss our proposed meta-learning based model \textbf{\textit{Meta-GraphSHS}} that aims to discover SHS nodes across diverse networks. The crucial challenge in this task is to capture the inter-graph differences and customize the model according to the new diverse graph (test graph). Meta-GraphSHS discovers SHSs in the new test graph $G_{test}$ (called meta-testing graph) by training the model on a set of diverse training graphs $G_{train}$ = $\{G_1, G_2,...,G_M\}$ (called meta-training graphs). The distribution over graphs is considered as a distribution over tasks and we consider the training task corresponding to each training graph in $G_{train}$ as $\tau$ = $\{\tau_1, \tau_2,.., \tau_M\}$. Similarly, $\tau_{test}$ is the task corresponding to test graph $G_{test}$. We further refer to the training and testing node set in all tasks $\tau$ as support set $S$ and query set $Q$. Let $S_i$ represent the support set, and $Q_i$ represent the query set for $G_i$, and $f_{\theta}$ represent the model, where $\theta$ is a set of model parameters.

Meta-GraphSHS addresses above mentioned challenge by first learning the general parameters from diverse training graphs $G_{train}$ and utilizing these parameters as a good initialization point for the test graph $G_{test}$. The learned general parameters are fine-tuned using the small number of available labelled data of the test graph\footnote{Fine-tune aims to precisely adjust the learned general model parameters in order to fit with the test graph.} (support set of $\tau_{test}$) and the obtained updated parameters are used to determine the labels of unlabeled nodes in the test graph (query set of $\tau_{test}$). In this way, Meta-GraphSHS avoids the need for repeated model training on each type of different graph (which is a time-consuming task) by designing a customized model that can be quickly adapted to the test graph under consideration in a few gradient steps, given only a few labelled nodes in the test graph. Our goal is to reach an \textit{“almost trained model”} that quickly adapts to the new graph. The performance of the Meta-GraphSHS is determined via meta-testing on the testing task $\tau_{test}$, by fine-tuning the model on the support set of $\tau_{test}$ and evaluating on the query set of $\tau_{test}$. Meta-GraphSHS uses MAML \cite{finn2017model} for updating the gradients during training. The procedure for meta-training and meta-testing are discussed below:

\begin{figure*}[t!]
 \centering
 \includegraphics[width=0.7\paperwidth]{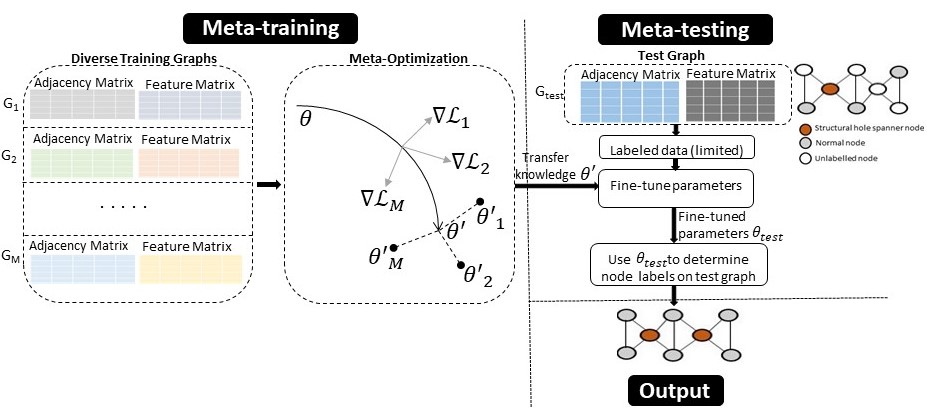}
 \caption{Architecture of proposed Meta-GraphSHS model.}
 \label{fig:meta-gnn}
\end{figure*}

\subsection{Meta Training}
During training, we intend to learn a set of generalizable parameters that act as a good initialization point for Meta-GraphSHS with the aim that the model rapidly adapts to the new task (test graph $G_{test}$) within a few gradient steps. For $M$ learning tasks $\{G_1, G_2,...,G_M\}$, we first adapt the model’s initial parameters to every learning task individually. We use $\mathcal{L}_{\tau_i}(\theta)$ to represent the loss function for task $\tau_i$. We utilize the same procedure as that of GraphSHS to train the model on task $\tau_i$ and compute the loss function for the same using \textcolor{blue}{Equation} \ref{eq_loss}. After computing the loss, updated model parameter $\theta_{i}$’ is computed using gradient descent. We update the parameters as follows:
\begin{equation}
\label{12}
\theta_{i}' = \theta - \alpha \nabla_\theta \mathcal{L}_{\tau_i}(\theta),
\end{equation}

\noindent where $\alpha$ represents learning rate and $\theta$ becomes $\theta_{i}'$ when adapting to the task $\tau_i$. We just describe 1 gradient step in \textcolor{blue}{Equation} \ref{12}, considering many gradient steps as a simple extension \cite{finn2017model}.
Since there are $M$ learning tasks, $M$ different variants of the initial model are constructed (i.e., $f_{\theta_1’}, ...,f_{\ theta_M’}$). We train the model parameters by optimizing the performance of $f_{\theta_i’}$ on all tasks. Precisely, the meta-objective is given by:
\begin{equation}
\label{35}
\theta = \mathop{\arg \min}\limits_{\theta} \sum_{i=1}^{M} \mathcal{L}_{\tau_i}(\theta_{i}').
\end{equation}

\begin{algorithm}[t!]
\caption{Training procedure for Meta-GraphSHS}
 \label{meta-algo}
 \begin{algorithmic}[1]
 \renewcommand{\algorithmicrequire}{\textbf{Input:}}
 \renewcommand{\algorithmicensure}{\textbf{Output:}}
 \REQUIRE $G_{train}$=$\{G_1, G_2,...,G_M\}$ and $G_{test}$
 \ENSURE Parameters: $\theta$
 \STATE Initialize $\theta$ randomly
 \WHILE{not early-stop}
 \FOR{$G_i$ = $G_1, G_2,...,G_M$}
 \STATE Split labelled nodes of $G_i$ into $S_i$ and $Q_i$
 \STATE Evaluate $\nabla_\theta \mathcal{L}_{\tau_i}(\theta)$ for $S_i$ 
 \STATE Compute parameter $\theta_{i}'$ using \textcolor{blue}{Equation} \ref{12}
 \ENDFOR
 \STATE Update $\theta$ on $\{Q_1, Q_2,...,Q_M\}$ using \textcolor{blue}{Equation} \ref{45}
 \ENDWHILE
 \STATE Fine-tune $\theta$ on $G_{test}$ using loss function
\end{algorithmic}
\end{algorithm}

Notably, optimization is performed over $\theta$, and the objective function is calculated using the updated parameters ${\ theta_i’}$. The model parameters are optimized in such a way that only a few gradient steps are needed to adjust to the new task, maximizing the model’s prediction performance on the new task. We use stochastic gradient descent to perform optimization across tasks. The parameter $\theta$ is updated as below:
\begin{equation}
\label{45}
\theta \leftarrow \theta - \gamma\nabla_\theta \sum_{i=1}^{M}\mathcal{L}_{\tau_{i}}(\theta_{i}'),
\end{equation}

\noindent where $\gamma$ represents meta-learning rate. The learned general parameter is then transferred to the meta-testing phase. The training procedure for Meta-GraphSHS is presented in \textcolor{blue}{Algorithm \ref{meta-algo}}.

\subsection{Meta Testing} The model in meta-testing phase is initialized with the learned parameters from meta-training phase, due to which the model is already almost trained. We then feed the support set $S_{test}$ of test graph $G_{test}$ as input to the model and fine-tune the learned model parameters precisely to fit with $G_{test}$. Since the model is already almost trained, it can be fine-tuned by just a few gradient steps. After fine-tuning, the model performance is assessed on query set $Q_{test}$ of test graph $G_{test}$. \textcolor{blue}{Figure \ref{fig:meta-gnn}} illustrates the overall architecture of the proposed model Meta-GraphSHS. The left side illustrates the meta-training phase that outputs a generalized parameter $\theta'$, which is transferred to the meta-testing phase. Meta-testing phase utilizes $\theta'$ as a good initialization point and fine tune $\theta'$ using labelled data from the test graph to obtain $\theta_{test}$, which can be used to obtain labels of unlabeled nodes in the test graph.

\section{Experimental Results}\label{c6_results}
We discuss the performance of the proposed models GraphSHS and Meta-GraphSHS by performing exhaustive experiments on widely used datasets. We first discuss the experimental setup. We then report the performance of GraphSHS on various synthetic and real-world datasets, followed by the performance of Meta-GraphSHS. Lastly, we present the parameter sensitivity analysis and application improvement. In our results, Top-$5$\%, Top-$10$\% and Top-$20$\% indicate the percentage of nodes labelled as SHSs.

\subsection{Experimental Setup for GraphSHS}

\subsubsection{\textit{Datasets}} 
We report the effectiveness and efficiency of GraphSHS on various datasets. The details of synthetic and real-world datasets are discussed below.
\vspace{0.15in}

\noindent \textbf{\textit{Synthetic Datasets}}. Considering the features of the Python NetworkX library, we used this library to create two types of synthetic graphs, namely Erdos-Renyi graphs (ER) \cite{erdHos1959random} and Scale-Free graphs (SF) \cite{onnela2007structure}. For each type, we generate test graphs of six different scales: $5000$, $10000$, $20000$, $50000$, {$100000$ and $150000$} nodes by keeping the parameter settings the same. In addition to these test graphs, we generate two graphs of 5000 nodes, one of each type (ER and SF) for training GraphSHS. Notably, for each type of graph (ER and SF), we train the model on a graph of 5000 nodes and test the model on all scales of graphs ($5000$, $10000$, $20000$, $50000$, $100000$ and $150000$ nodes). \textcolor{blue}{Table \ref{summary_syn}} presents the summary of graph generating parameters for synthetic datasets\footnote{Due to the computational challenges in computing ground truths for large-scale graphs, we limit the maximum number of edges to 200000 for ER and SF graphs with 100000 and 150000 nodes.}.

\begin{table*}[ht!] 
\caption{Summary of synthetic datasets.}
\label{summary_syn}
\renewcommand{\arraystretch}{1.1}
\centering 
\resizebox{\textwidth}{!}{\begin{tabular}{lll}\hlineB{1.5} 
\textbf{Graph type} & \multicolumn{2}{c}{\textbf{Graph generating parameters}} \\ \hlineB{1.5}
\multirow{2}{*}{\textbf{Erdos-Renyi Graphs}} & Number of nodes & \specialcell[t]{5000, 10000, 20000, 50000} \\ \cline{2-3}
 &{Probability of adding a random edge} & 0.001\\\cline{2-3}
 
 & Number of nodes & \specialcell[t]{{100000, 150000}} \\ \cline{2-3}
 &{Probability of adding a random edge} & 0.0001\\\hlineB{1.5}
\multirow{4}{*}{\textbf{Scale-Free Graphs}} & Number of nodes & \specialcell[t]{5000, 10000, 20000, 50000, {100000, 150000}}\\ \cline{2-3}
 & Alpha & 0.4 \\\cline{2-3}
 & Beta & 0.05 \\\cline{2-3}
 & Gamma & 0.55 \\ \hlineB{1.5} 
\end{tabular}}
\end{table*}

\noindent \textbf{\textit{Real-World Datasets}}. We use five real-world datasets to determine GraphSHS performance. \textcolor{blue}{Table \ref{summary_real}} presents the summary of these datasets, and the details are discussed below:
\begin{itemize}[leftmargin=*]
\item \textbf{{ca-CondMat \cite{leskovec2007graph}}} is a scientific collaboration network from arXiv. This network covers collaborations between the authors who have submitted papers in condensed matter category.
\item \textbf{{email-Enron \cite{leskovec2009community}}} is a communication network of emails where nodes denote the addresses, and edge connects two nodes if they have communicated via email. 
\item \textbf{{coauthor \cite{lou2013mining}}} is an author-coauthor relationship network. It consists of coauthor relationships obtained from papers published in major computer science conferences. 
\item \textbf{{com-DBLP \cite{yang2012defining}}} is a coauthor network. Nodes represent the authors, and an edge connects the authors if they have published at least one paper together.
\item \textbf{{com-Amazon \cite{yang2012defining}}} is a customer-product network obtained from amazon website. Nodes represent the customers, and edge connects the customers who have purchased the same product.
\end{itemize}

\begin{table}[t!] 
\caption{Summary of real-world datasets.}
\label{summary_real}
\renewcommand{\arraystretch}{1.1}
\centering
\footnotesize
\begin{tabular}{llll}\hlineB{1.5}
\textbf{Dataset} & \textbf{Nodes} & \textbf{Edges} & \textbf{Avg degree} \\ \hlineB{1.5}
ca-CondMat & 21,363 & 91342 & 8.55 \\
email-Enron & 33,696 & 180,811 & 10.73 \\
coauthor & 53,442 & 255,936 & 4.8 \\
com-DBLP & 317,080 & 1,049,866 & 6.62 \\
com-Amazon & 334,863 & 925,872 & 5.53 \\ \hlineB{1.5}
\end{tabular} 
\end{table}

\subsubsection{\textit{{Evaluation Metrics}}} 
For baselines and GraphSHS, we measure the effectiveness and efficiency in terms of accuracy and running time, respectively.

\subsubsection{\textit{Baselines}} 
We compare GraphSHS with the two representative SHS identification algorithms:

\begin{itemize}[leftmargin=*]
\item \textbf{{Constraint.}} Constraint is a heuristic solution to discover SHSs in the network \cite{burt1992structural}. It measures the degree of redundancy among the neighbors of the node. Constraint $C$ of a node $i$ is defined as:
\begin{equation*}
\begin{aligned}
C(i) = \sum_{j\in N(i)}{\left({p_{ij}}+\sum_{q}{p_{iq} p_{qj}}\right)}^2, \quad q\neq i,j
\end{aligned}
\end{equation*}
where $N(i)$ is neighbors of node $i$, $q$ is the node in the ego network other than node $i$ and $j$, and $p_{ij}$ represents the weight of edge $(i,j)$.

\item \textbf{{Closeness Centrality.}} The closeness centrality of a node is the reciprocal of sum of length of the shortest paths from the node to all other nodes in the graph  \cite{bavelas1950communication}. Rezvani et al. \cite{rezvani2015identifying} used closeness centrality as a base to propose an algorithm Inverse Closeness Centrality (ICC), for discovering SHSs in the network. Closeness Centrality (CC) of node $i$ is calculated as:
\begin{equation*}
\begin{aligned}
CC(i) = \frac{1}{\sum_{j \in V}\text{SP}(i,j)},
\end{aligned}
\end{equation*}
where $\text{SP}(i,j)$ is the shortest path between node $i$ and $j$.
\item {\textbf{Vote Rank Algorithm.} Vote Rank is an iterative algorithm to identify top-$k$ decentralized spreaders with the best spreading ability. This algorithm uses a voting scheme to rank nodes in a graph, where each node votes for its in-neighbors, and the node with the highest number of votes is selected in each iteration \cite{zhang2016identifying}.}
\end{itemize}

\subsubsection{\textit{Ground Truth Computation}} \label{ground}
For all the datasets under consideration, we used the Python library NetworkX to calculate nodes’ SHS score (BC). Besides, for large scale graphs, i.e., com-DBLP and com-Amazon, we used the SHS score (BC) reported by AlGhamdi et al. \cite{alghamdi2017benchmark}. {The authors \cite{alghamdi2017benchmark} performed parallel implementation of the Brandes algorithm, utilizing 96,000 CPU cores on a supercomputer to compute exact BC values for large graphs. We were not able to perform experiments on very large synthetic networks, as it is} {computationally challenging to compute the ground truth BC for larger graphs using normal system configurations; therefore, we limit the synthetic network size to 150000 nodes.} After computing the SHS score of the nodes, we sort the nodes in descending order of their score values. We label the high score $k$\% nodes as SHS nodes and the rest as normal ones. We evaluate the performance of GraphSHS for three different values of $k$, i.e., $5$, $10$ and $20$. Labelled graphs are used to train GraphSHS, and we assess the performance of GraphSHS on the test graphs.

\subsubsection{\textit{Training Details}}
We perform all the experiments on a Windows 10 PC with a CPU of 3.20 GHz and 16 GB RAM. We implement the code in PyTorch. We fix the number of layers to 4 and the embedding dimension to 128. Parameters are trained using Adam optimizer with a learning rate of 0.01 and weight decay $5e-4$. We train the GraphSHS for 200 epochs on ER graph of 5000 nodes and evaluate the performance on test ER graphs of all scales. We adopted the same training and testing procedure for SF graphs. Since real-world networks demonstrate attributes similar to SF graphs; therefore, we train our model on an SF graph of 5000 nodes and test the model on real-world datasets. Besides, we used an inductive setting where test graphs are invisible to the model during the training phase.

\subsection{Performance of GraphSHS on Synthetic Datasets}
\textcolor{blue}{Tables} \ref{result_syn} and \ref{time_syn} report the accuracy and run time of the comparative algorithms and GraphSHS on synthetic graphs. \textcolor{blue}{Table \ref{result_syn}} shows that GraphSHS achieves higher classification accuracy than the baselines. For example, in the SF graph of 5000 nodes, GraphSHS performs better than the baselines, closeness centrality, constraint and vote rank by achieving Top-$5$\% accuracy of 96.78\%, whereas the best accuracy achieved by the baseline is 96.12\%. Besides, GraphSHS is $190.9$ times faster than the best result for the same scale and type of graph, as reported in \textcolor{blue}{Table \ref{time_syn}}. For the ER graph of 10000 nodes, although GraphSHS sacrifices 2.93\% in Top-$10$\% accuracy in contrast to the top accuracy (vote rank); however, it is over $134$ times faster.

\begin{table*}[t!] 
\caption[Classification accuracy on synthetic datasets of different scales.]{Classification accuracy (\%) on synthetic datasets of different scales. Top-$5$\%, Top-$10$\% and Top-$20$\% indicate the percentage of nodes labelled as SHSs. Bold results indicate the best results among the proposed and all baselines.}
\label{result_syn}
\renewcommand{\arraystretch}{1.1}
\centering
\footnotesize
\resizebox{\textwidth}{!}{\begin{tabular}{llllllllllll} \hlineB{1.4} 
\textbf{Scale $\downarrow$} & \multirow{4}{*}{\textbf{Method}} & \multicolumn{2}{c}{\textbf{{Top-$5$\%}}} & \multicolumn{2}{c}{\textbf{{Top-$10$\%}}}& \multicolumn{2}{c}{\textbf{{Top-$20$\%}}} \\ \cmidrule(lr){3-4} \cmidrule(lr){5-6} \cmidrule(lr){7-8}
\textbf{Dataset $\rightarrow$} & & \textbf{SF} & \textbf{ER} & \textbf{SF} & \textbf{ER} & \textbf{SF} & \textbf{ER} \\ \hlineB{1.5}
 & Constraint & 94.23 & 93.98 & 91.57 & 91.05 & {88.26} & {84.35}\\ 
5,000 & Closeness centrality & 94.58 & 93.38 & 88.66 & 90.7 & 86.54 & 82.26\\
& {Vote Rank} & {96.12} & {95.43} & {92.41} & {91.13} & {\textbf{92.79}} & {\textbf{85.51}}\\
 & GraphSHS (Proposed) & \textbf{96.78} & \textbf{95.66} & \textbf{93.66} & \textbf{91.30} & 87.65 & 83.22\\ \hlineB{1.5}
 & Constraint & 94.02 & 94.45 & 90.87 & 92.02 & {87.24} & {85.64} \\
10,000& Closeness centrality & 94.81 & 94.09 & 89.21 & {92.34} & 85.75 & 80.75\\
 & {Vote Rank} & {95.29} & {95.18} & {92.04} & {\textbf{93.82}} & {\textbf{94.41}} & {\textbf{{87.98}}} \\
 & GraphSHS (Proposed) & \textbf{96.44} & \textbf{95.29} & \textbf{93.23} & 90.89 & 86.92& 82.45 \\ \hlineB{1.5}
 & Constraint & 95.02 & 93.97& 88.23 & 90.61 & {88.32} & \textbf{87.41} \\ 
20,000 & Closeness centrality & 94.29 & 94.35 & 87.71 & \textbf{{91.39}} & 82.78 & 80.34\\
& {Vote Rank} & {95.01} & {94.88} & {92.59} & {91.25} & {\textbf{89.64}} & {{87.28}} \\
 & GraphSHS (Proposed) & \textbf{96.31} & \textbf{95.23} & \textbf{92.97} & 90.56 & 85.80 & 81.22 \\\hlineB{1.5}
 & Constraint & 94.93 & 93.85 & 87.12 & 88.65 & 84.77 & \textbf{82.36}\\
50,000 & Closeness centrality & 93.95 & 91.89 & 85.27 & 84.91 & 81.60 & 72.37 \\
& {Vote Rank} & {94.64} & {93.27} & {91.54} & {87.83} & {85.49} &  {81.22}\\
 & GraphSHS (Proposed) & \textbf{95.03} & \textbf{94.81} & \textbf{92.01} & \textbf{89.49} & \textbf{85.55} & 80.24 \\ \hlineB{1.5}
& Constraint & NA & 90.49 & NA & 82.08 & NA  & 68.15\\
100,000 & Closeness centrality & 93.51 & 87.18 & 88.06 & 85.48 & 84.20 & \textbf{85.33}  \\
& Vote Rank & 94.18 & 92.72 & 91.43 & 86.76 & {87.11} & 80.73\\
 & GraphSHS (Proposed) &  \textbf{94.93} & \textbf{93.75} & \textbf{91.84} & \textbf{87.9} & \textbf{88.37} & 80.6 \\ \hlineB{1.5}
 
& Constraint & NA & 89.40 & NA & 82.56 & NA & 68.03\\
150,000 & Closeness centrality & 93.04 & 91.61&  91.14  &  88.50 &  {89.03} & \textbf{86.74} \\
& Vote Rank & 93.92 & 91.93 & 90.73 & 86.92 & 85.87 & 78.17\\
 & GraphSHS (Proposed) & \textbf{94.25} & \textbf{93.56} & \textbf{91.69} & \textbf{89.35} & \textbf{88.82}  & 83.42 \\ 
 \hlineB{1.5}
\end{tabular}}
\end{table*}

{For a large-scale SF graph of 100000 nodes, GraphSHS achieves higher accuracy than the baselines by achieving a Top-$5$\% accuracy of 94.93\%, whereas the best accuracy achieved by the baseline is 94.18\% (vote rank). Notably, the constraint algorithm cannot complete the computation for the SF graph of 100000 and 150000 nodes within three days, so we put NA corresponding to its accuracy and time in the results. Moreover, our proposed model achieves the best accuracy for SF and ER graphs of 150000 nodes in the case of Top-$5$\% and Top-$10$\% accuracy; however, for Top-$20$\%, closeness centrality achieves higher accuracy.} To avoid unfair comparison, we have not considered the training time of the model as none of the baseline algorithms needs to be trained. Hence, it is logical not to count the training time. Moreover, GraphSHS converges rapidly, and the convergence time is around 15 minutes. In addition, our model works in multi-stages. We can train the model whenever we have time and later use it for discovering SHSs. However, all the baselines identify SHSs in one stage only.

\begin{table*}[t!] 
\caption[Run time comparison of different algorithms on synthetic datasets.]{Run time (sec) comparison of different algorithms on synthetic datasets of different scales. Bold results indicate the best results, and second best results are underlined.}
\label{time_syn}
\renewcommand{\arraystretch}{1.1}
\centering
\footnotesize
\resizebox{\textwidth}{!}{\begin{tabular}{lllp{2cm}lp{2.6cm}l} \hlineB{1.4} 
\textbf{Scale} & \textbf{Dataset} & \textbf{Constraint} & \textbf{Closeness centrality } & {\textbf{Vote Rank}} & \textbf{GraphSHS (Proposed)}& \textbf{Speedup} \\ \hlineB{1.5}
\multirow{2}{*}{5,000} & SF & 16013.2 & {40.1} & {\underline{17.18}} & \textbf{0.09} & 190.9x \\ 
 & ER & \underline{5.8} & 45.9 & {32.26} &\textbf{0.1} & 58x\\\hline
\multirow{2}{*}{10,000} & SF & 21475.3 & {199.4} & {\underline{29.72}} & \textbf{0.3} & 99.1x\\
 & ER & \underline{67.2} & 286.1 & {316.88} & \textbf{0.5} & 134.4x\\\hline
\multirow{2}{*}{20,000} & SF & 24965.3 & {836.7} & {\underline{155.89}} & \textbf{0.7} & 222.7x \\
 & ER & \underline{884.7} & 1820.7 &  {2970.47} & \textbf{1.75} & 505.5x\\\hline
\multirow{2}{*}{50,000} & SF & 28336.1 & {5675.8} &  {\underline{1088.39}} & \textbf{2.5} & 435.3x\\
 & ER & 27754.2 & \underline{2055.4} & {3987.62} & \textbf{12.6} & 163.1x\\\hlineB{1.5}
\multirow{2}{*}{100,000} & SF & NA & 4442.9 &  {\underline{2164.68}} & \textbf{15.4} & 140.5x\\
 & ER & 29345.1 & 13746.9 &  {\underline{3512.17}} & \textbf{27.4} & 128.2x \\\hlineB{1.5}
 \multirow{2}{*}{150,000} & SF & NA & {3143.1} & {\underline{1592.46}} & \textbf{21.6} & 73.7x\\
 & ER & 31601.73 &  22338.3 &  {\underline{15631.67}} & \textbf{33.7} & 463.8x \\\hlineB{1.5}
\end{tabular}}
\end{table*}

\textcolor{blue}{Table \ref{time_syn}} reports the running time of baselines and GraphSHS. For a small scale ER graph of 5000 nodes, GraphSHS takes 1 sec to discover SHSs, whereas closeness centrality takes 45.9 sec and vote rank takes 32.26 sec. For a large-scale ER graph of 50000 nodes, GraphSHS takes less than 13 sec to discover SHSs. However, constraint, vote rank and closeness centrality require a large amount of time to discover SHSs in large-scale networks. For ER graph of 50000 nodes, constraint took around 7.5 hours, whereas both closeness centrality and vote rank took around 1 hour to discover SHSs. {For ER graph of 150000 nodes, GraphSHS takes less than 34 sec to discover SHSs. However, all the baselines require a large amount of time to discover SHSs and GraphSHS is 463.8 times faster than the most efficient baseline.} The results prove that our model has a considerable efficiency advantage over other models in run time. 

The proposed model GraphSHS consistently achieves the best Top-$5$\% accuracy for ER as well as SF graphs of all scales. GraphSHS achieves the highest Top-$10$\% accuracy for  most of the cases; however, vote rank achieves better Top-$10$\% accuracy for ER graphs of 10000 nodes and closeness centrality for ER graphs of 20000 nodes. The vote rank algorithm outperforms most of the comparative methods for ER and SF graphs in terms of Top-$20$\% accuracy. Although other algorithms achieve better accuracy than GraphSHS in a few cases, but our model runs faster. GraphSHS is at least 58 times faster than the baselines on synthetic graphs. Results from \textcolor{blue}{Table \ref{result_syn}} show that the classification accuracy is inversely proportional to the size of the network. In addition, there is a decrease in Top-$k$\% accuracy as we increase the value of $k$. \textcolor{blue}{Table \ref{type_syn}} presents the generalization accuracy of the proposed model GraphSHS across different types of graphs. We train the GraphSHS on ER and SF graphs separately, and test on both types of graphs. For this analysis, we only consider graphs of 5,000 nodes for training and testing. The results demonstrate that GraphSHS attains the best accuracy when the training graph is similar to testing graphs.

\begin{table}[t!] 
\caption{Generalization accuracy of GraphSHS on different types of synthetic datasets.}
\label{type_syn}
\renewcommand{\arraystretch}{1.1}
\centering 
\footnotesize
\begin{tabular}{p{3.9cm}|p{1.9cm}p{1.9cm}} \hlineB{1.5} 
\backslashbox{\textbf{Train} $\downarrow$}{\textbf{Accuracy}}{\textbf{Test$\rightarrow$}}
&\makebox[5em]{\textbf{ER\_5,000}}&\makebox[5em]{\textbf{SF\_5,000}}\\\hlineB{1.5}

ER\_5,000 & \textbf{95.66\%} & 94.22\% \\
SF\_5,000 & 93.16\% & \textbf{96.78\%} \\\hlineB{1.5}
\end{tabular} 
\end{table}

\subsection{Performance of GraphSHS on Real-World Datasets}

\begin{table*}[t!] 
\caption[Classification accuracy on real-world datasets.]{Classification accuracy (\%) on real-world datasets. Top-$5$\%, Top-$10$\% and Top-$20$\% indicate the percentage of nodes labelled as SHSs. Bold results indicate the best results among the proposed and all baselines.}
\label{result_real}
\renewcommand{\arraystretch}{1.1}
\centering
\footnotesize
\begin{tabular}{p{3cm}p{4cm}p{2cm}p{2cm}p{1.7cm}} \hlineB{1.5} 
\textbf{Dataset}& \textbf{Method} & {\textbf{{Top-$5$\%}}} & {\textbf{{Top-$10$\%}}}& {\textbf{{Top-$20$\%}}} \\ \hlineB{1.5}

 & Constraint & 94.41 & 90.18 & {86.23}\\
ca-CondMat & Closeness centrality & 95.05 & 89.78 & 82.56\\
& {Vote Rank} & {95.59} & {90.15} & {\textbf{89.44}} \\
 & GraphSHS (Proposed) & \textbf{95.73} & \textbf{90.43} & 83.23\\ \hline
 
 & Constraint & 95.77 & 91.87 & \textbf{87.38}\\ 
email-Enron & Closeness centrality & 95.41 & 90.98 & 83.71 \\ 
& {Vote Rank} & {95.93} & {93.01} & {87.28}\\
 & GraphSHS (Proposed) & \textbf{96.2} & \textbf{93.13} & 86.49 \\ \hline
 & Constraint & 93.60 & 90.77 & \textbf{86.61}\\
coauthor & Closeness centrality & 94.4 & 88.95 & 81.07\\
& {Vote Rank} &  {94.59} &  {\textbf{93.6}} & {86.53}\\
 & GraphSHS (Proposed) & \textbf{95.03} & {91.28} & 80.91 \\ \hline
 
 & Constraint & 92.4 & \textbf{91.42} & \textbf{84.21}\\ 
com-DBLP & Closeness centrality & \textbf{95.1} & 89.9 & 80.2\\ 
& {Vote Rank} &{NA}  &{NA}  &{NA}\\
 & GraphSHS (Proposed) & 93.11 & 89.2 & 81.24 \\ \hline
 
 & Constraint & 94.61 & \textbf{88.12} & \textbf{83.15}\\
com-Amazon & Closeness centrality & 93.13 & 87.30 & 77.83 \\ 
& {Vote Rank} &{NA}  &{NA}  &{NA} \\
 & GraphSHS (Proposed) & \textbf{94.71} & 85.21 & 78.23\\ \hlineB{1.5}
\end{tabular}
\end{table*}

\begin{table*}[t!] 
\caption[Run time comparison of different algorithms on real-world datasets.]{Run time (sec) comparison of different algorithms on real-world datasets. Bold results indicate the best results and the second best results are underlined.}
\label{time_real}
\renewcommand{\arraystretch}{1.1}
\centering
\footnotesize
\begin{tabular}{p{2.4cm}p{2.1cm}p{2.1cm}p{2.1cm}p{2.1cm}p{1.5cm}} \hlineB{1.5} 
\textbf{Dataset} & \textbf{Constraint} & \textbf{Closeness centrality } & {\textbf{Vote Rank}} &  \textbf{GraphSHS (Proposed)}& \textbf{Speedup} \\ \hlineB{1.5}
ca-CondMat & {1403.2} & 2853.4 & {\underline{983.88}} & \textbf{1.07} & 919.5x\\
email-Enron & \underline{1968.5} & 2903.4 & {2541.6} & \textbf{2.2} & 894.7x\\
coauthor & \underline{417.8} & 5149.6 &  {3948.53} & \textbf{2.5} & 167.1x\\
com-DBLP & \underline{8574.9} & 38522.1 &{NA} & \textbf{19.2} & 446.6x\\
com-Amazon & \underline{4533.4} & 42116.9 &{NA} & \textbf{18.9} & 239.8x\\
\hlineB{1.5}
\end{tabular} 
\end{table*}

This section evaluates GraphSHS performance on five real-world datasets. Since real-world networks exhibit some characteristics similar to that of SF graphs; therefore, we train our model on an SF graph (SF graph of 5000 nodes having the same properties as discussed in \textcolor{blue}{Table \ref{summary_syn}}) and test the model on real-world datasets. We present the Top-$k$\% accuracy and running time of the baselines in \textcolor{blue}{Tables} \ref{result_real} and \ref{time_real}, respectively. The results illustrate that GraphSHS attains competitive Top-$k$\% accuracy compared to other baselines. Nevertheless, considering the trade-off between accuracy and run time, GraphSHS runs much faster than the baselines. Take the example of the ca-CondMat network; GraphSHS performs better than the baselines by achieving Top-$5$\% accuracy of 95.73\% and Top-$10$\% accuracy of 90.43\%. 
Although vote rank performs better in the Top-$20$\% accuracy for the same network; however, GraphSHS is 919.5 times faster. {In the email-Enron graph, GraphSHS performs better than the baselines by achieving the highest Top-$5$\% and Top-$10$\% accuracy.} On the other hand, if we take an example of a large-scale network, such as com-Amazon, constraint outperforms GraphSHS in Top-$10$\% and Top-$20$\% accuracy; however, GraphSHS is 239.8 times faster than the best baseline. {The vote rank algorithm cannot complete the computation for com-DBLP and com-Amazon networks within 3 days, probably due to the large network size (approximately 1,000,000 edges). Therefore, we have included NA corresponding to its accuracy and time in the results.} GraphSHS achieves the best Top-$5$\% accuracy in four real-world networks and the best Top-$10$\% accuracy in three out of five networks. However, for the Top-$20$\% accuracy, constraint algorithm is more accurate. \textcolor{blue}{Table \ref{time_real}} shows that GraphSHS achieves a minimum speedup of 167.1 and is up to 919.5 times faster than the baseline algorithms. The run time comparison indicates the efficiency advantage of our model over other baselines. Furthermore, our results proved that the proposed simple graph neural network architecture GraphSHS is sufficient to solve the SHSs discovering problem on real-world networks.

\subsection{Performance of Meta-GraphSHS}
In order to obtain a classifier Meta-GraphSHS that can discover SHSs in diverse networks, we train our model on different types of networks. {We evaluate the performance of our model on the following synthetic and real-world datasets. The summary of dataset is presented in \textcolor{blue}{Table \ref{meta_data}}.}

\begin{itemize}[leftmargin=*]
\item {\textbf{Synthetic graph.} We generate one synthetic graph consisting of 36 sub-graphs of 3 different types, i.e., Erdos-Renyi, Scale-Free, and Gaussian Random Partition graphs. The graph contains 12 sub-graphs of each type, and each sub-graph consists of a minimum of 1000 nodes and a maximum of 5000 nodes. }

\item {\textbf{Real-world graph.} We obtain one real-world graph by combining 2 diverse real-world graphs, i.e., ca-CondMat and email-Enron (refer \textcolor{blue}{Table \ref{summary_real}} for properties of these graphs). For each of these graphs, we disconnect the original graph to obtain 12 much smaller subgraphs. In this way, the overall graph contains 24 sub-graphs (12 of each type) and each sub-graph consists of a minimum of 1000 nodes and a maximum of 3000 nodes. }
\end{itemize}

We follow the procedure discussed in \textcolor{blue}{Section \ref{ground}} for obtaining the ground truths for the graphs and label the top 5\% nodes in each of the sub-graph as SHS nodes. We use 80\% of the sub-graphs for training (meta-training), and 20\% for testing (meta-testing). We train Meta-GraphSHS for 200 epochs, and set $\alpha$ to 0.1 and $\gamma$ to 0.001. The training sub-graphs are used to optimize the model parameters (to learn generalizable parameters by observing multiple graphs from different domains). Only 50\% of the nodes in the testing sub-graphs are labelled. The labelled nodes in testing sub-graphs are used to fine-tune the trained model to accurately determine labels for the rest of the nodes in the test graphs.

\begin{table*}[t!] 
\caption{Summary of dataset for evaluating Meta-GraphSHS.}
\label{meta_data}
\renewcommand{\arraystretch}{1.4}
\centering 
\footnotesize
\begin{tabular}{p{2.6cm}|p{4cm}|c|p{4.5cm}} \hlineB{1.5} 
\textbf{Dataset} & \textbf{Types of subgraph} & \textbf{\#Subgraphs} & \textbf{\#Nodes in each subgraph} \\ \hlineB{1.5}
\multirow{3}{*}{{Synthetic graph}} & Erdos-Renyi graphs & \multirow{3}{*}{{36}} & \\ \cline{2-2}
 & Scale-Free graphs & & 1000 to 5000 \\ \cline{2-2}
 & Gaussian Random Partition graphs & & \\
\hlineB{1.5}
\multirow{2}{*}{{Real-world graph}} & ca-CondMat graphs & \multirow{2}{*}{{24}} & 1000 to 3000\\ \cline{2-2}
 & email-Enron graphs & &  \\ \cline{2-2}
\hlineB{1.5}
\end{tabular} 
\end{table*}

\begin{table}[t!] 
\caption{Classification accuracy of Meta-GraphSHS.}
\label{result_meta}
\renewcommand{\arraystretch}{1.4}
\centering 
\footnotesize
\begin{tabular}{c|c|c} \hlineB{1.5} 
\textbf{Dataset}& \textbf{Method} & \textbf{Accuracy} \\ \hlineB{1.5}
\multirow{2}{*}{Synthetic graph} & \multirow{1}{*}{GraphSHS} & 93.5\% \\ \cline{2-3}
 & \multirow{1}{*}{Meta-GraphSHS} & 96.2\% \\ \hlineB{1.5} 

 \multirow{2}{*}{Real-world graph} & \multirow{1}{*}{GraphSHS} & 92.1\%\\ \cline{2-3}
 & \multirow{1}{*}{Meta-GraphSHS} & 94.8 \%\\ \hlineB{1.5} 
\end{tabular} 
\end{table}

{\textcolor{blue}{Table \ref{result_meta}} shows the accuracy achieved by Meta-GraphSHS for discovering SHSs in diverse synthetic and real-world graphs compared to that of GraphSHS. For diverse synthetic graphs, Meta-GraphSHS discovers SHS nodes with high accuracy of 93.5\% and outperforms GraphSHS by an accuracy of 2.7\%. For diverse real-world graphs, GraphSHS is able to achieve an accuracy of 92.1\%, whereas Meta-GraphSHS model achieves an accuracy of 94.8\%, which is much higher than GraphSHS accuracy.} Our previous results from \textcolor{blue}{Table \ref{result_syn}} and \textcolor{blue}{Table \ref{result_real}} illustrate that even though GraphSHS discovers SHSs with high accuracy when trained and tested on graphs from the same domain; however, the accuracy decreases when GraphSHS is tested on graphs from different domains than what the model is trained on. The reason for the low accuracy of GraphSHS in the case of diverse graphs is that the model is not able to capture the inter-graph differences. {The performance of Meta-GraphSHS on both synthetic and real-world graphs shows that machine learning models explicitly designed for a particular task outperform the models designed for generalized tasks. The advantage of meta learning models is that they learn from experience and quickly adapt to new tasks with minimal training data. This allows them to achieve high accuracy, even with diverse data. This is why once trained, Meta-GraphSHS generalizes well and discovers SHSs from diverse networks with high accuracy.}

\subsection{Parameter Sensitivity}
We perform experiments on the real and synthetic networks to determine the impact of parameters on the accuracy of both the proposed models, GraphSHS and Meta-GraphSHS. Particularly, we study the sensitivity of the number of layers (depth) and embedding dimensions for the models. We vary the number of layers and embedding dimension among \{1, 2, 3, 4, 5, 6\} and \{16, 32, 64, 128, 256\}, respectively. \textcolor{blue}{Figures \ref{fig:real}} and \textcolor{blue}{\ref{fig:syn}} show the parameter sensitivity of GraphSHS on real-world and synthetic datasets, respectively. The results illustrate that the accuracy is relatively low for fewer aggregation layers (depth), as shown in \textcolor{blue}{Figures \ref{fig:real}(a)} and \textcolor{blue}{\ref{fig:syn}(a)}. The reason for low accuracy is insufficient aggregated information due to the limited reachability of the nodes. Our results show that initially, the SHSs identification accuracy increases with the increase in the number of layers (model depth); however, if we increase the depth of the model over four layers, the accuracy starts decreasing. The reason for this is the over-smoothing problem  \cite{li2018deeper, yang2020toward, pasa2021multiresolution}. Besides, results from \textcolor{blue}{Figures \ref{fig:real}(b)} and \textcolor{blue}{\ref{fig:syn}(b)} show that for higher embedding dimensions, GraphSHS performs better as higher embedding dimensions provide the GraphSHS with more ability to represent the network. 
\begin{figure*}[t!]
 \centering
 \includegraphics[width=0.7\paperwidth]{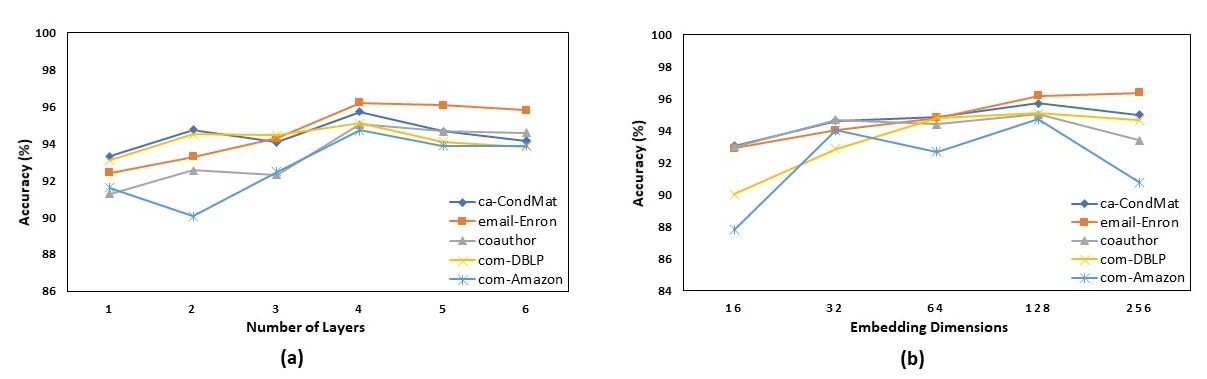}
 \caption{Parameter sensitivity of GraphSHS on real-world datasets.}
 \label{fig:real}
\end{figure*}
\begin{figure*}[t!]
 \centering
 \includegraphics[width=0.7\paperwidth]{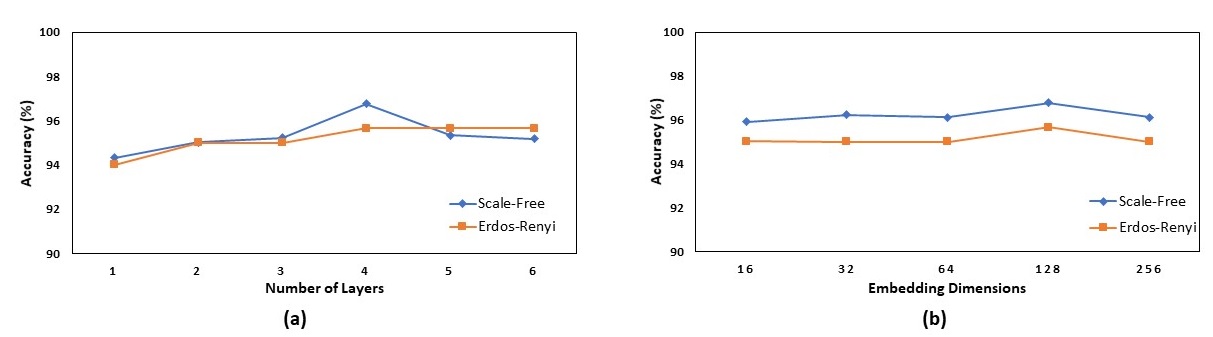}
 \caption{Parameter sensitivity of GraphSHS on synthetic datasets.}
 \label{fig:syn}
\end{figure*}
\begin{figure*}[t!]
 \centering
 \includegraphics[width=0.7\paperwidth]{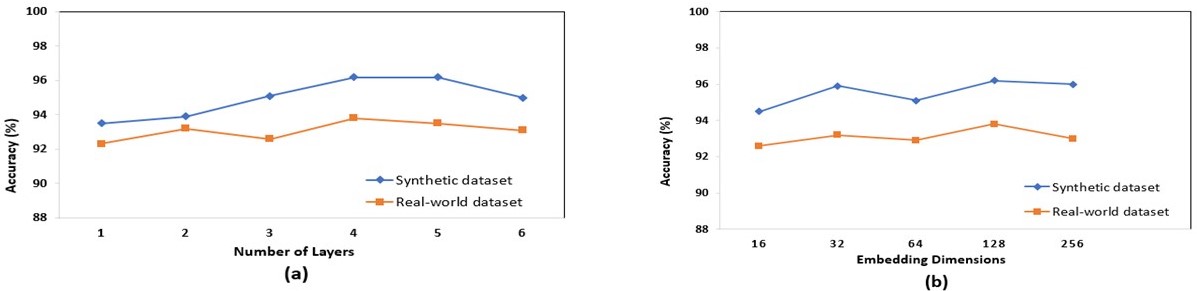}
 \caption[Parameter sensitivity of Meta-GraphSHS on datasets.]{Parameter sensitivity of Meta-GraphSHS on synthetic and real-world datasets.}
 \label{fig:metaablagtion}
\end{figure*}
Our reasoning behind the improved accuracy of GraphSHS, with the increase in the number of layers and embedding dimensions, is further supported by the parameter sensitivity analysis results of the Meta-GraphSHS model. The parameter sensitivity analysis of Meta-GraphSHS on synthetic and real-world datasets is demonstrated in \textcolor{blue}{Figure \ref{fig:metaablagtion}}. The results indicate that the accuracy is comparatively low for a smaller number of aggregation layers, as depicted in \textcolor{blue}{Figure \ref{fig:metaablagtion}(a)}, and the accuracy increases as the number of layers increases. However, the accuracy begins to decline after four layers. Similarly, \textcolor{blue}{Figure \ref{fig:metaablagtion}(b)} presents the accuracy of Meta-GraphSHS on varying the number of embedding dimensions. Firstly, the accuracy increases with the increase in embedding dimensions; however, the accuracy starts deteriorating on increasing the embedding dimension to 256 or higher. The reason for the deteriorating performance of Meta-GraphSHS on increasing the embedding dimensions beyond 128 is the similar node embeddings that make it difficult for the model to distinguish between the nodes and, consequently, the model mislabels the nodes.

\subsection{Application Improvement}
GraphSHS can be used to discover SHSs in a dynamic network, where nodes and edges change over time. For example, on Facebook and Twitter, links appear/disappear whenever a user friend/unfriend others on Facebook or follow/unfollow others on Twitter. As a result, discovered SHSs change, and hence, it is essential to track the new SHSs in the updated network. Traditional algorithms are highly time-consuming and might not work efficiently for dynamic networks. Additionally, it is highly possible that the network has already been changed by the time these algorithms re-compute SHSs. Therefore, we need a fast heuristic that can quickly update SHSs in dynamic networks. We can use our proposed model GraphSHS for discovering SHSs in dynamic networks. Even if training the model takes a few hours to learn the evolving pattern of the network, we only need to train the model once, and after that, whenever there is a change in the network, our trained GraphSHS can identify the new SHSs within a few seconds. We compare our proposed GraphSHS with the solution designed by Goel et al. \cite{goel2021maintenance} that discovers SHSs in dynamic networks. We start with an entire network and arbitrarily delete 100 edges, one edge at a time, and calculate the average speedup of GraphSHS over dynamic solution \cite{goel2021maintenance}. As shown in \textcolor{blue}{Table \ref{appl}}, GraphSHS is at least $89.8$ times faster than \cite{goel2021maintenance}. This confirms the efficiency of the proposed model in dynamic networks. 
\begin{table}[ht!] 
\caption{Performance of GraphSHS on dynamic networks.}
\label{appl}
\renewcommand{\arraystretch}{1.1}
\centering
\footnotesize
\begin{tabular}{c|c|c|c} \hlineB{1.5} 
\textbf{Dataset} & \textbf{\# Nodes} & \textbf{\# SHSs discovered} & \textbf{Speedup} \\ \hlineB{1.5}
\multirow{3}{*}{{Scale-Free}} & \multirow{3}{*}{{5,000}} & 1 & 89.8x \\ \cline{3-4}
 & & 5 & 139.5x \\ \cline{3-4}
 & & 10 & 163.7x \\ 
\hlineB{1.5}
\end{tabular} 
\end{table}

\subsection{Discussion}
Our experiments on various datasets show that our simple message-passing graph neural network models are sufficient to solve the SHSs discovering problem. Our proposed model GraphSHS, provides a significant run time advantage over other algorithms. GraphSHS is at least 167.1 times faster than the baselines on real-world networks and at least 58 times faster than the baselines on synthetic networks. Even though we trained GraphSHS on the synthetic SF graph, the model achieved high accuracy when tested on real-world graphs. This shows the inductive nature of our proposed model, where the model can be trained on one graph and used to predict SHSs in another graph. Besides, once trained, Meta-GraphSHS generalizes well and identifies SHSs from diverse networks with a high accuracy of 96.2\% for synthetic graphs and 94.8\% for real-world graphs. The following observations are the potential reasons behind the success of the proposed models:
\vspace{0.15in}

\begin{enumerate}
    \item Our proposed graph neural network based models follow a similar architecture to that of message-passing graph neural networks, which are proved to be universal under sufficient conditions \cite{loukas2019graph}. Our model meets those conditions, and we believe that the \textit{universal characteristics} of the message passing graph neural network enable our model to capture the relevant features that are important for discovering SHSs, which is confirmed from our experimental results.
    \item We train the proposed model in an end-to-end manner with the exact betweenness centrality values as ground truth. Similar to successful applications of deep learning on text or speech,  the model generally learns and performs well if provided with sufficient training data.
\end{enumerate}

In addition, theoretically, we showed that the depth of the GraphSHS should be at least $\Omega(\sqrt{n}/\log n)$. However, practically, deep GNNs suffer from the over-smoothing issue that leads to the embeddings of nodes indistinguishable from each other. We conduct parameter sensitivity analysis to investigate the effect of model depth (number of layers) on the accuracy of discovering SHSs in the network. Our experimental results showed that after a few layers, the performance of GraphSHS starts deteriorating. The reason for the dropping performance of GraphSHS is similar node embeddings, which results in the model being unable to differentiate between the nodes and, hence, mislabel them. Therefore, in our experiments, we made the necessary adjustments and used four layers instead in order to avoid the over-smoothing problem.

\section{Chapter Summary}\label{c6_summary}

Structural hole spanner identification problem has various real-world applications such as information diffusion, community detection etc. However, there are two challenges that need to be addressed; 1) to discover SHSs efficiently in large-scale networks; 2) to discover SHSs effectively across diverse networks. This chapter investigated the power of message-passing GNNs for identifying SHSs in large-scale networks and diverse networks. We first transformed the SHS identification problem into a learning problem and designed an efficient message-passing GNN-based model, GraphSHS, that identifies SHS nodes in large-scale networks with high accuracy. We then proposed another effective meta-learning model, Meta-GraphSHS, that discovers SHSs across different types of networks. Meta-GraphSHS learns general transferable knowledge during the training process and then quickly adapts by fine-tuning the model parameters for each new unseen graph. We used an inductive setting that enables the proposed models to be generalizable to new unseen graphs. Theoretically, we showed that the proposed graph neural network model needs to be at least $\Omega(\sqrt{n}/\log n)$ deep to calculate the SHSs discovery problem. To evaluate the model’s performance, we performed empirical analysis on various datasets. Our experimental results demonstrated that the proposed models achieved high accuracy. GraphSHS is at least 167.1 times faster than the baselines on large-scale real-world networks, showing a considerable advantage in run time over the baseline algorithms.

\chapter{Discovering Top-\textit{{k}} Structural Hole Spanners in Dynamic Networks} 

\label{SHS_LCN} 

\textbf{\underline{Related publication:}} 
\vspace{0.06in}

\noindent This chapter is based on our following three papers titled:
\begin{enumerate}
    \item “\textit{Maintenance of Structural Hole Spanners in Dynamic Networks}” published in the 46\textsuperscript{th} IEEE Conference on Local Computer Networks (LCN), 2021 \cite{goel2021maintenance}.
    \item “\textit{Discovering Structural Hole Spanners in Dynamic Networks via Graph Neural Networks}” published in The 21\textsuperscript{st} IEEE/WIC/ACM International Conference on Web Intelligence and Intelligent Agent Technology (WI-IAT), 2022 \cite{goel2022discovering}.
    \item ``\textit{Discovering Top-\textit{k} Structural Hole Spanners in Dynamic Networks}'' \cite{goel2023discovering}.
\end{enumerate}

\vspace{0.1in}

\noindent \textcolor{blue}{Chapter 5} studied the graph neural network models for determining bottleneck structural hole spanner nodes in large-scale and diverse networks. However, real-world networks are dynamic in nature. Therefore, this chapter aims to discover the SHSs in dynamic networks. This chapter first develops a Tracking-SHS algorithm for updating SHSs in dynamic networks. Our algorithm reuses the information obtained during the initial runs of the static algorithm and avoids the recomputations for the nodes unaffected by the updates. We also design a Graph Neural Network model, GNN-SHS, to discover SHSs in dynamic networks. We provide a theoretical analysis of the Tracking-SHS algorithm, and our theoretical results prove that the Tracking-SHS algorithm attains high speedup compared with the static algorithm. We perform extensive experiments on various datasets to demonstrate the efficiency of our proposed algorithm and model. Our results show that the proposed techniques are more efficient than the baselines.

\section{Introduction}
With the emergence of large-scale networks, researchers are designing new techniques to analyze and study the properties of large-scale networks \cite{guo2021itcn, binesh2021distance}. These networks inherently possess a community structure where the nodes within the community share close interests, characteristics, behaviour, and opinions \cite{zannettou2018origins}. The absence of connection between different communities in the network is known as \textbf{\textit{Structural Holes (SH)}} \cite{burt2009structural}. A community needs to have connectivity with other communities to access novel information \cite{rinia2001citation}. The structural hole theory states that the users who fill the “\textit{holes}” between various users or groups of users that are otherwise disconnected get positional advantages in the network, and these users are known as \textbf{\textit{Structural Hole Spanners}} \cite{lou2013mining}. SHSs have many applications, including information diffusion, opinion control, identifying central hubs, preventing the spread of rumours, community detection and identifying critical nodes in tactical environments \cite{lin2021efficient, amelkin2019fighting, zareie2019influential, yu2021modeling, zhao2021community, ahmad2023review}. A large number of centrality measures, such as pairwise connectivity \cite{borgatti2006identifying}, closeness centrality \cite{bavelas1950communication}, degree centrality, constraint \cite{burt1992structural} etc., exist in literature to discover critical nodes in the network. However, SHS acts as a bridge between the nodes of different communities \cite{rezvani2015identifying} and controls information diffusion in the network \cite{lou2013mining}. Therefore, we have an important observation that removing SHS will disconnect a maximum number of node pairs in the network and block information propagation among them. \textcolor{blue}{Figure} \ref{fig:ch6_comp} illustrates that node $i$ plays a vital role, and removing node $i$ will disconnect maximum node pairs in the network, in turn blocking information propagation among the maximum number of nodes in the network, whereas the impact of the removal of other nodes is less significant. Therefore, we define SHS as a node whose removal minimizes the \textit{Pairwise Connectivity (PC)} of the residual network, i.e., the node with the maximum pairwise connectivity score. This SHS definition aims to capture the nodes located between otherwise disconnected groups of nodes.

\begin{figure*}[!t]
  \centering
    \includegraphics[width=0.51\paperwidth]{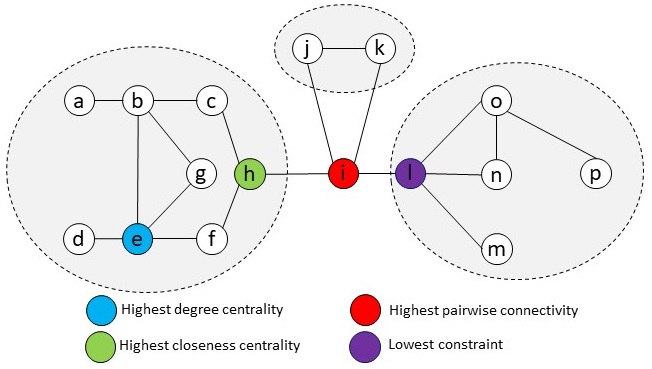}
    \caption{Comparison of centrality measures.}
    \label{fig:ch6_comp}
  \end{figure*}

Various solutions \cite{lou2013mining, rezvani2015identifying, gong2019identifying, goyal2007structural, he2016joint, xu2017efficient, xu2019identifying, tang2012inferring, ding2016method} are developed to discover SHSs in static networks. Nevertheless, the real-world network changes over time. For example, it is essential to handle outdated web links in web graphs or obsolete user profiles in a social network; these are examples of decremental updates in the networks. In contrast, on Facebook and Twitter, links appear and disappear whenever a user friend/unfriend others on Facebook or follow/unfollow others on Twitter; which is an example where decremental as well as incremental updates happen in the networks. As a result, the discovered SHSs change. \textit{The limitation here is that in the literature, there is no solution available to discover SHS nodes in dynamic networks.} The classical SHS identification algorithms are considerably time-consuming and may not work efficiently for dynamic networks. In addition, the network may have already been changed by the time classical algorithms recompute SHSs. Hence, designing a fast mechanism that can efficiently discover SHSs as the network evolves is crucial.  

\textit{\textbf{We aim to propose efficient algorithms for discovering structural hole spanner nodes in dynamic networks.}} We formulate the problem of discovering SHS nodes in dynamic networks as \textbf{\textit{Structural hole Spanner Tracking (SST) problem}}. While the traditional SHS problem focuses on discovering a set of SHSs that minimizes the PC of the network, the SST problem intends to update the already discovered set of SHSs, as the network evolves. In order to track SHS nodes in dynamic networks, we first propose an efficient \textbf{\textit{Tracking-SHS algorithm}} that maintains Top-$k$ SHSs for decremental edge updates in the network by discovering a set of affected nodes. {Tracking-SHS aims to maintain and update the SHS nodes faster than recomputing them from the ground.} We obtain some properties to identify the set of affected nodes due to changes in the network. In addition, we reuse the information from the initial runs of the static algorithm in order to avoid the recomputations for the unaffected nodes. Tracking-SHS executes greedy interchange by replacing an old SHS node with a high PC score node from the network. Besides, the use of priority queues saves the repetitive computation of PC scores. Our theoretical and empirical results demonstrate that the proposed Tracking-SHS algorithm achieves higher speedup than the static algorithm. In addition, inspired by the recent advancements of Graph Neural Network on various graph problems, we propose another model \textbf{\textit{GNN-SHS}} (\underline{G}raph \underline{N}eural \underline{N}etworks for discovering \underline{S}tructural \underline{H}ole \underline{S}panners in dynamic networks), a graph neural network-based framework that discovers SHSs in the dynamic network. Our proposed model works for both incremental and decremental edge updates of the network. We regard the dynamic network as a sequence of snapshots and aim to discover SHSs in these snapshots. Our proposed  GNN-SHS model uses the network structure and features of nodes to learn the embedding vectors of nodes. Our model aggregates embedding vectors from the neighbors of the nodes, and the final embeddings are used to discover SHS nodes in the network. GNN-SHS aims to discover SHSs in dynamic networks by reducing computational time while achieving high accuracy. We perform a detailed theoretical analysis of the proposed Tracking-SHS algorithm, and our results show that for a specific type of graph, the Tracking-SHS algorithm achieves 1.6 times of speedup compared with the static algorithm. Besides, we validated the performance of Tracking-SHS algorithm and GNN-SHS model by performing extensive experiments on real-world and synthetic datasets. Our experimental results show that the Tracking-SHS algorithm is at least 3.24 times faster than the static algorithm. In addition, the results demonstrate that the proposed GNN-SHS model is at least 31.8 times faster and up to 2996.9 times faster than the baseline, providing a considerable efficiency advantage in run time.

\noindent \textbf{\textit{Contributions:}} In this chapter, we make the following contributions.
\begin{itemize}
\item  \textbf{Tracking-SHS algorithm.} We propose an efficient algorithm Tracking-SHS, that maintains SHS nodes for decremental edge updates in the network. We derive some properties to discover the set of nodes affected due to updates in the network and avoid recomputation for the unaffected nodes, so as to enhance the efficiency of the proposed algorithm. In addition, we extend our proposed algorithm from a single edge update to a batch of updates.
\item \textbf{GNN-SHS model.} We propose an efficient graph neural network-based model GNN-SHS, for discovering SHSs in dynamic networks. Our model considers both the incremental and decremental edge updates in the network. GNN-SHS model preserves the network structure and node feature, and uses the final node embedding to discover SHSs.
\item \textbf{Theoretical analysis.} We theoretically show that the depth of the proposed graph neural network-based model GNN-SHS should be at least $\Omega({n}^2/\log^2 n)$ to solve the SHSs problem. In addition, we analyze the performance of the proposed Tracking-SHS algorithm theoretically, and our theoretical results show that for specific types of graphs, such as Preferential Attachment graphs, the proposed algorithm achieves 1.6 times of speedup compared with the static algorithm.
\item \textbf{Experimental analysis.} We validate the performance of the proposed Tracking-SHS algorithm and GNN-SHS model by conducting extensive experiments on various real-world and synthetic datasets. The results demonstrate that the Tracking-SHS algorithm achieves a minimum of 3.24 times speedup on real-world datasets compared with the static algorithm. Besides,  GNN-SHS model is at least 31.8 times faster and, on average, 671.6 times faster than the Tracking-SHS algorithm. However, this significant speedup with the GNN-SHS model comes at the cost of a potential decrease in accuracy.

\end{itemize}

\noindent \textbf{\textit{Chapter organization.}}  \textcolor{blue}{Section \ref{sec3}} presents the preliminaries and background of the problem. \textcolor{blue}{Section \ref{last_Problem}} presents the problem description. \textcolor{blue}{Section \ref{sec4}} discusses the proposed Tracking-SHS algorithm for tracking SHSs in dynamic networks. \textcolor{blue}{Section \ref{sec5}} discusses the proposed model GNN-SHS for discovering SHSs in dynamic networks. \textcolor{blue}{Section \ref{sec6}} presents the theoretical performance analysis of Tracking-SHS algorithm, and \textcolor{blue}{Section \ref{sec7}} discusses the extensive experimental results. Finally, \textcolor{blue}{Section \ref{sec8}} concludes the chapter.

\section{Preliminaries and Background}\label{sec3}
\noindent This section discusses the preliminaries and background of the problem. 

\noindent \textbf{Network model.}
\noindent A graph is defined as $G = (V, E)$, where $V$ is the set of nodes (vertices), and $E \subseteq V \times V$ is the set of edges. Let $n = |V |$ and $m = |E|$. We have considered unweighted and undirected graphs\footnote{We consider only graphs without self-loops or multiple edges.}. A \textit{path} $p_{ij}$ from node $i$ to $j$ in an undirected graph $G$ is a sequence of nodes ${\{v_{i},v_{i+1},.....,v_{j}\}}$  such that each pair  $(v_{i},v_{i+1})$ is an edge in $E$. A pair of nodes $i$, $j \in V$ is \textit{connected} if there is a path between $i$ and $j$. Graph $G$ is connected if all pair of nodes in $G$ are connected. Otherwise, $G$ is \textit{disconnected}. A \textit{connected component} or \textit{component} $C$ in an undirected graph $G$ is a maximal set of nodes in which a path connects each pair of node. A dynamic graph $G$ can be modelled as a finite sequence of graphs $(G_t, G_{t+1}, ..., G_T)$. Each $G_t$ graph represents the network’s state at a discrete-time interval $t$. We refer to each of the graph in the sequence as a snapshot graph. Each snapshot consists of the same set of nodes, whereas edges may appear or disappear over time. Hence, each graph snapshot can be described as an undirected graph $G_t = (V, E_t)$, containing all nodes and only alive edges at the time interval under consideration. Due to the dynamic nature of the graph, the edges in the graph may appear or disappear, due to which the label of the nodes (SHS or normal node) may change. Therefore, we need to design techniques that can discover SHSs in each new snapshot graph quickly. Let $\vec{x}(i)$ denotes the feature vector of node $i$, $N(i)$ denotes the neighbors of node $i$, $h^{(l)}(i)$ represents the embedding of node $i$ at the $l^{th}$ layer of the model. Let $l$ represents the index of the aggregation layer, where $l = (1,2,...,L)$. 

\begin{definition}
\noindent \textbf{Pairwise connectivity.} \normalfont The pairwise connectivity $u(i,j)$ for any node pair $(i,j)\in V \times V $ is quantified as: 
{\begin{equation}
u(i,j)=\begin{cases} 1 & \text{if\,$i$\,and\,$j$\,are\,connected} \\ 0 & \text{otherwise}
\end{cases} 
\end{equation}}
\end{definition}

\begin{definition}
\noindent  \textbf{Total pairwise connectivity.}  \normalfont The total pairwise connectivity $P(G)$, i.e., pairwise connectivity across all pair of nodes in the graph, is given by:
{\begin{equation}
\text{\ensuremath{{\displaystyle P(G)=\sum_{i,j\in V\times V,i\neq j}{\textstyle u(i,j)}}}}.
\end{equation}}
\end{definition}
\begin{definition}

\noindent  \textbf{Pairwise connectivity score.}  \normalfont The pairwise connectivity score $c(i)$ of node $i$ is the contribution of node $i$ to the total pairwise connectivity score of the graph. Pairwise connectivity score $c(i)$ of node $i$ is computed as follows:
{\begin{equation}
\label{pc_score}
c(i)=P(G)-P(G\backslash\{i\}).
\end{equation}}
\end{definition}

\noindent \textbf{Graph Neural Networks (GNNs).} GNNs \cite{wu2020comprehensive} are designed by extending the deep learning methods for the graph data and are broadly utilized in various fields, e.g., computer vision, graph mining problems, etc. GNNs usually consist of graph convolution layers that extract local structural features of the nodes  \cite{velivckovic2017graph}. GNNs learn the representation of nodes by aggregating features from the ego network of node. Every node in network is described by its own features, and features of its neighbors \cite{kipf2016semi}.

\begin{figure}[ht!]
 \centering
\includegraphics[width=0.4\paperwidth]{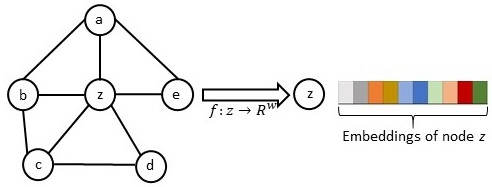}
 \caption{Embedding of node $i$.}
 \label{fig:embed_ch6}
\end{figure}

\noindent \textbf{Network Embedding.} Network embedding \cite{cui2018survey} is a procedure with which network nodes can be described in a low-dimensional vector. Embedding intends to encode the nodes so that the resemblance in the embedding, approximates the resemblance in the network \cite{aguilar2021novel}. Embedding can be used for various graph mining problems, including node classification, regression problems, graph classification, etc. \textcolor{blue}{Figure \ref{fig:embed_ch6}} depicts an example of node embedding in low-dimensional space.

\section{Problem Description}\label{last_Problem}
In this section, we formally state the structural hole spanner problem for static and dynamic networks.

\subsubsection{Structural Hole Spanner Problem for Static Networks}

\noindent This section discusses the greedy algorithm for discovering SHSs in the static network. Selecting a set of SHS nodes in one observation may not be a correct approach due to the influence of cut nodes. Therefore, the algorithm works iteratively by identifying one SHS node in each iteration.

\begin{definition}
\noindent \textbf{Structural Hole Spanner Problem.}  \normalfont Given a graph $G = (V,E)$, and a positive integer $k$, the \textit{SHS problem} is to identify a set of SHSs Top-$k$ in $G(V, E)$, where Top-$k$ $\subset V$ and $ \mid $Top-$k \mid = k$, such that the removal of nodes in Top-$k$ from $G$ minimizes the total pairwise connectivity in the residual subgraph $ G(V\backslash $Top-$k)$.
\end{definition}
\begin{equation}
\text{Top-}k = {min}\,\,\{{P(G\backslash \text{Top-}k)}\},
\end{equation}
where Top-$k\,\subset\,V$ and $ \mid $Top-$k \mid =k$.

\textcolor{blue}{Algorithm \ref{ch_6_alg_1}} describes a greedy heuristic for identifying the SHSs in the static networks. The algorithm repeatedly selects a node $v'$ with a maximum PC score, i.e., the node which when removed from the network minimizes the total PC of the residual network. In step 3, for computing the PC score of each node $v \in V$, the algorithm initiates a depth-first search (DFS) and selects a node with the maximum PC score. The selected node is then eliminated from the network and added to the SHS set Top-$k$. The process repeats until $k$ nodes are discovered. The run time of \textcolor{blue}{Algorithm \ref{ch_6_alg_1}} is $O(kn(m+n))$. DFS takes $O(m+n)$ time, and the PC score is calculated for $n$ nodes. The process repeats for $k$ iterations, and hence, the complexity follows.

\begin{algorithm}[ht!]
 \caption{Structural hole spanner identification.}
 \label{ch_6_alg_1}
 \begin{algorithmic}[1]
 \renewcommand{\algorithmicrequire}{\textbf{Input:}}
 \renewcommand{\algorithmicensure}{\textbf{Output:}}
 \REQUIRE Graph $G(V, E)$, $k$
 \ENSURE  SHS set Top-$k$
  \STATE Initialize Top-$k=\phi$
  \WHILE{$|$Top-$k|<k$}
  \STATE $v'=argmax_{v\,\text{\ensuremath{\in}}\,V}\,c(v)$
  \STATE $G=G\backslash\{v'\}$
  \STATE Top-$k=$Top-$k\bigcup\{v'\}$
  \ENDWHILE
 \RETURN Top-$k$
 \end{algorithmic}
 \end{algorithm}

\begin{theorem} [\textbf{Dinh et al. \cite{dinh2011new}}] \label{thm1_ch6}
\textit{\textbf{Discovering SHS problem is NP-hard.}}
\end{theorem}
\noindent \begin{proof}
We present an alternative proof, where we reduce the SHS model to vertex cover instead of $\beta$-Vertex Disruptor used in Dinh et al. \cite{dinh2011new}. The reason for this alternative proof is that it will be used as a foundation for \textcolor{blue}{Theorem \ref{thm2_ch6}}. We show that the Vertex Cover (VC) problem is reducible to the SHS discovery problem. The definition of \textit{Structural Hole Spanners} states that SHSs are the set of $k$ nodes,  which, when deleted from the graph, minimizes the total pairwise connectivity of the subgraph. Let $G=(V, E)$ be an instance of a VC problem in an undirected graph $G$ with $V$ vertices and $E$ edges. \textit{VC problem} aims to discover a set of vertices of size $k$ such that the set includes at least one endpoint of every edge of the graph. If we delete the nodes in vertex cover from the graph, there will be no edge in the graph, and the pairwise connectivity of the residual graph will become $0$, i.e., $P(G)$ is minimized. In this way, we can say that graph $G$ has a VC of size $k$ if and only if graph $G$ have structural hole spanners of size $k$ (that makes the residual graph's connectivity to be 0). Therefore, discovering the exact $k$ SHSs problem is NP-hard as a similar instance of a vertex cover problem is a known NP-hard problem.
\end{proof}

\begin{figure}[ht!]
 \centering
 \includegraphics[width=0.4\paperwidth]{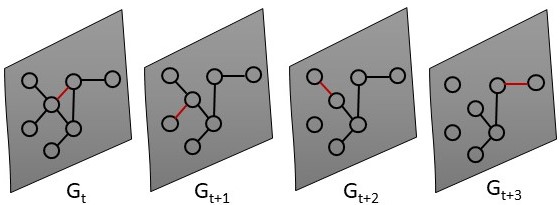}
 \caption{Illustration of graph snapshots.}
 \label{fig:evolving}
\end{figure}

\subsubsection{Structural Hole Spanner Problem for Dynamic Networks}

\noindent We represent the dynamic network as a sequence of snapshots of graph $(G_t, G_{t+1}, ..., G_T)$ and each snapshot graph describes all the edges that occur between a specified discrete-time interval (e.g., sec, minute, hour). \textcolor{blue}{Figure \ref{fig:evolving}} illustrates four snapshots of graph taken at time $t, t+1, t+2$, $t+3$. Due to the dynamic nature of the graph, SHSs in the graph also change, and it is crucial to discover SHSs in each new snapshot graph quickly. Traditional SHSs techniques are time-consuming and may not be suitable for dynamic graphs. Therefore, we need fast algorithms to discover SHSs in dynamic networks. We formally define the structural hole spanner problems for dynamic networks as follows:

\noindent \textbf{{Problem 1:}} \textit{\textbf{Tracking Top-$k$ SHS nodes in Dynamic Networks (Structural Hole Spanner Tracking Problem).}}

\noindent \textbf{{Given:}}   \normalfont Given a graph $G = (V, E)$, SHS set Top-$k$, and edge update $\Delta E$.

\noindent \textbf{\textit{Goal:}} Design a Tracking-SHS algorithm to identify SHS set Top$'$-$k$ with cardinality $k$ in $G'=(V,E + \Delta E)$ by updating Top-$k$ such that the removal of nodes in Top$'$-$k$ from $G'$ minimizes the total pairwise connectivity in the residual subgraph $G'(V\backslash$Top$'$-$k$).

\noindent \textbf{{Problem 2:}} \textit{\textbf{Discovering Top-$k$ SHS nodes in Dynamic Networks.}}

\noindent \textbf{{Given:}} Given snapshots of graph $G_t = (V,E_t), G_{t+1} = (V,E_{t+1})$,...\\
$G_{T} = (V,E_{T})$ and integer $k > 0$.

\noindent \textbf{\textit{Goal:}} Train a model GNN-SHS to discover SHS set Top-$k$  with cardinality $k$ in the dynamic network (snapshots of graph) such that the removal of nodes in Top-$k$ from the snapshot of graph minimizes the total pairwise connectivity in the residual subgraph. We aim to utilize the pre-trained model to discover SHSs in each new snapshot of the graph quickly.

\textcolor{blue}{Theorem \ref{thm1_ch6}} showed that discovering the exact $k$ SHSs in the network is an NP-hard problem. Therefore, we adopted a greedy heuristic, as discussed in \textcolor{blue}{Algorithm \ref{ch_6_alg_1}} for finding the Top-$k$ SHS nodes. It should be noted that the greedy algorithm, despite being a heuristic, is still computationally expensive with a complexity of  $O(kn(m+n))$. Nevertheless, real-world networks are dynamic, and they change rapidly. Since these networks change quickly,  the Top-$k$ SHSs in the network also change continuously. A run time of $O(kn(m+n))$ of the greedy algorithm is too high and not suitable for practical purposes where speed is the key. For instance, it is highly possible that the network might already change by the time greedy algorithm computes the Top-$k$ SHSs. Therefore, we need an efficient solution that can quickly discover Top-$k$ SHSs in the changing networks. A natural approach for the SST problem is to run \textcolor{blue}{Algorithm \ref{ch_6_alg_1}} after each update, providing us with a new set of SHSs. Nevertheless, computing SHS set from scratch after every update is a time-consuming process which motivates us to design \textit{Tracking-SHS algorithm} that can handle the dynamic nature of the networks. Tracking-SHS algorithm focuses on updating the Top-$k$ SHSs without explicitly running \textcolor{blue}{Algorithm \ref{ch_6_alg_1}} after every update. We also propose a GNN-based model to discover SHSs in dynamic networks by transforming the SHS discovery problems into a learning problem. While our proposed Tracking-SHS algorithm only considers decremental edge updates in the network, our second proposed model, GNN-SHS considers both incremental as well as decremental edge updates of the network. The main idea of our second model \textit{GNN-SHS}, is to rely on the greedy heuristic (\textcolor{blue}{Algorithm \ref{ch_6_alg_1}}) and treat its results as true labels to train a graph neural network for identifying the Top-$k$ SHSs. The end result is a significantly faster heuristic for identifying the Top-$k$ SHSs. Our heuristic is faster because we only have to train our graph neural network model once, and thereafter, whenever the graph changes, the trained model can be utilized to discover SHSs. \textcolor{blue}{Table \ref{inc_dec}} presents the comparison of both the proposed models for discovering SHS nodes in dynamic networks.

\begin{table}[!h] 
\caption{Comparison of proposed models for discovering SHSs in dynamic networks.}
\label{inc_dec}
\renewcommand{\arraystretch}{1.15}
\centering
\footnotesize
\begin{tabular}{p{3cm}p{2.5cm}p{2.5cm}p{1.5cm}} \hlineB{1.5} 
\textbf{Proposed algorithm/ model} & \textbf{Decremental update} & \textbf{Incremental update} & \textbf{Batch update}\\ \hlineB{1.5}
Tracking-SHS & \checkmark &  $\times$  & \checkmark\\
GNN-SHS & \checkmark & \checkmark & \checkmark\\ \hlineB{1.5} 
\end{tabular} 
\end{table}

\section{Proposed Algorithm: Tracking-SHS}\label{sec4}

\noindent In the real world, it is improbable that the network evolves drastically within a short time. The similarity in the structure of the network before and after updates leads to a similar SHS set. We propose an efficient {Tracking-SHS algorithm}, that maintains Top-$k$ SHSs in the dynamic network. This algorithm considers the situations where it is crucial to handle the decremental updates in the network, such as outdated web links in web graphs or obsolete user profiles in a social network. As a result, the discovered SHSs change due to the dynamic nature of the network. Therefore, it is vital to design an algorithm that quickly discovers SHSs as the network evolves. Our proposed Tracking-SHS algorithm addresses this issue. Instead of constructing the SHS set from the ground; we start from old Top-$k$ set and repeatedly update it. Our algorithm aims to maintain and update the SHS set faster than recomputing it from scratch. This section first discusses the mechanism to find the set of affected nodes and various cases due to updates in the network. We also present a method to compute the PC score of the nodes efficiently and the procedure to update Top-$k$ SHSs for decremental edge updates in the network. Later, we extend our proposed single-edge update algorithm to a batch of updates.

\subsection{Finding Affected Nodes}\label{sec4.1}
\noindent Whenever there is an edge update in the network, by identifying the set of affected nodes, we need to recompute the pairwise connectivity scores of these nodes only.

\begin{definition}
\noindent \textbf{Affected Nodes.}  \normalfont Given an original graph $G(V,E)$, and an updated graph $G’ = G(V, E\backslash (a,b))$, affected nodes are the set of nodes whose pairwise connectivities change as a result of the deletion of edge $(a, b)$. More precisely, the set of affected nodes $A$ is defined as $A = \{y \in V$ such that $c(y) \ne c’ (y)\}$.
\end{definition}

 \begin{figure}[ht!]
  \centering
    \includegraphics[width=0.35\paperwidth]{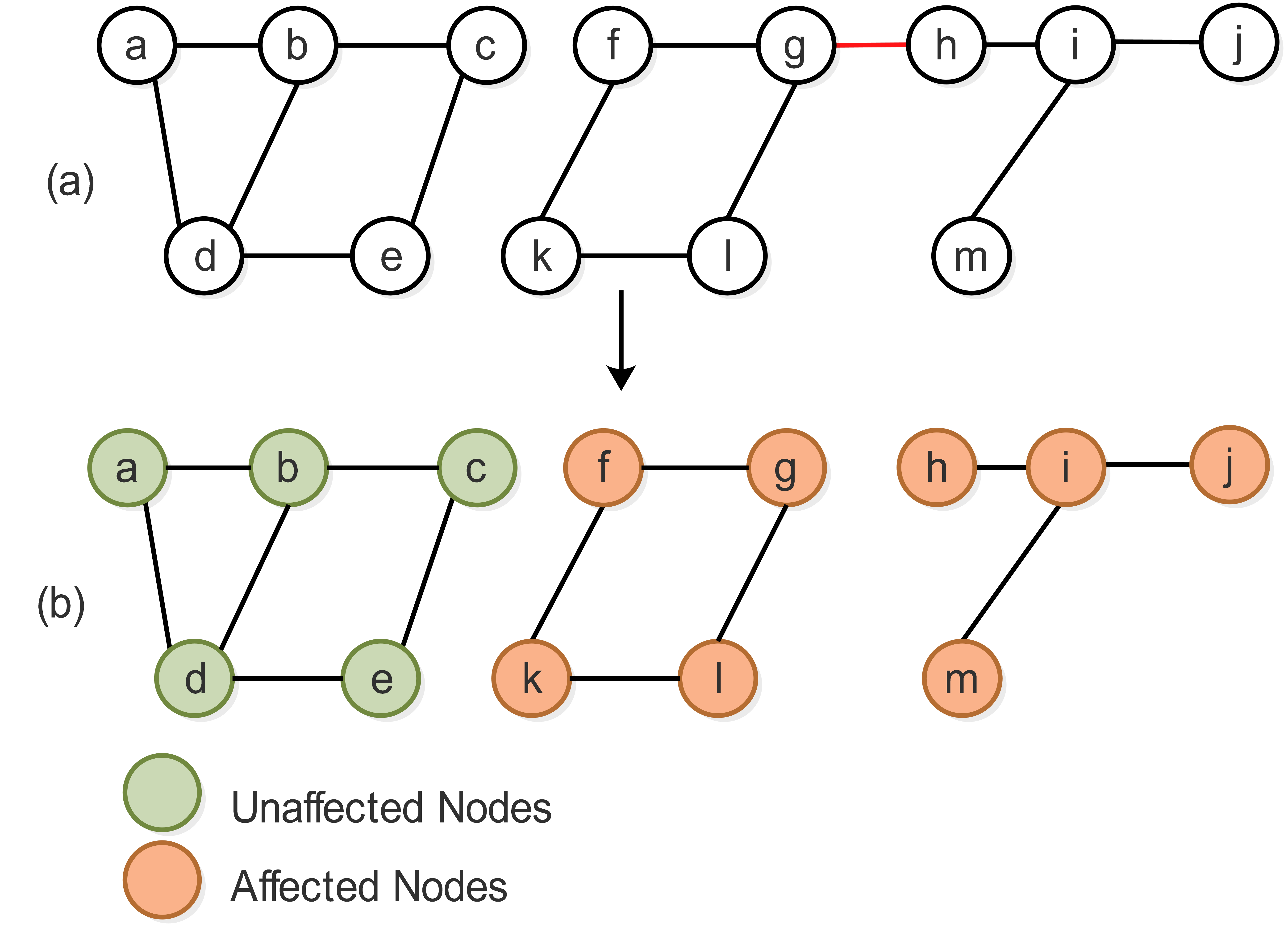}
    \caption[Illustration of affected and unaffected nodes due to network updates.]{Illustration of affected and unaffected nodes due to network updates (a) Original network (b) Updated network.}
    \label{fig:aff}
  \end{figure}
 
\noindent When an edge $(a, b)$ is deleted from the network, affected nodes $A$ are the set of nodes reachable from node $a$, in case edge $(a, b)$ is a non-bridge edge. On the other hand, if edge $(a, b)$ is a bridge edge, we have two set of affected nodes $A_a$ and $A_b$, representing the set of nodes reachable from node $a$ and $b$, respectively. Let $G(V,E)$ be the original network, as shown in \textcolor{blue}{Figure \ref{fig:aff}(a)}, and \textcolor{blue}{Figure \ref{fig:aff}(b)} shows the updated network $G'=G (V,E \backslash (g,h))$. We first show the case where edge deletion changes the PC score of some nodes $v$, i.e., $c(v) \ne c’ (v)$. When an edge $(g, h)$ is deleted from $G$, the PC score of nodes $\{f, g, k, l, h, i, j, m\}$ changes as shown in \textcolor{blue}{Figure \ref{fig:aff}(b)}, and these nodes are called affected nodes. Next is the case where edge deletion does not change the PC score of the other set of nodes $u$, i.e., $c(u) = c’ (u)$. \textcolor{blue}{Figure \ref{fig:aff}(b)} shows that on the deletion of edge $(g,h)$, the PC score of the nodes $\{a, b, c, d, e\}$ does not change, and therefore, these nodes are unaffected nodes. On deletion of edge $(g, h)$ in \textcolor{blue}{Figure \ref{fig:aff}}, set of affected nodes are $A_g = \{f, g, k, l\}$ and $A_h = \{h, i, j, m\}$.

\begin{lemma}
\label{ch6_l_1}
\normalfont \textit{\textbf{Given a graph $G = (V,E)$ and an edge update $(a, b)$, any node $v \in V$ is an affected node if $u(v, a) = 1$ or $u(v, b) = 1$, resulting in $c(v)$ before deletion not equal to $c'(v)$ after deletion. Otherwise, the node is unaffected.}}
\end{lemma}

\noindent \begin{proof}
Node $v$ is affected if it is either reachable from node $a$ or $b$ via any path in the graph. Nodes that are not reachable either from node $a$ or $b$ are not affected as these nodes do not contribute to the PC score of node $a$ or $b$. 
\end{proof}

\subsection{Various Cases for Edge Deletion}
\noindent This section enumerates various cases that may arise due to the deletion of an edge from the network.

\noindent \textbf{Case 1: No change in connected component (Non-bridge edge)}

\noindent When the deleted edge $(a,b)$ is a non-bridge edge, there is no change in the connected components of the updated network, and node $a$ and $b$ still belong to the same connected component. The PC score of only few nodes in the component changes, i.e., $c(r) \neq c’ (r), \, r\in C(a)$. Here, the maximum number of affected nodes are the nodes in the connected component containing node $a$, i.e., $ \mid A \mid $ = $ \mid C(a) \mid $. Let \textcolor{blue}{Figure \ref{fig:cases}(a)} shows the original network $G_{t}$, when an edge $(b,c)$ is deleted at $G_{t+1}$, node $b$ and $c$ belong to the same connected component. However, the PC score of some nodes changes as highlighted in \textcolor{blue}{Figure \ref{fig:cases}(b)}.

 \begin{figure*}[!t]
  \centering
    \includegraphics[width=0.65\paperwidth]{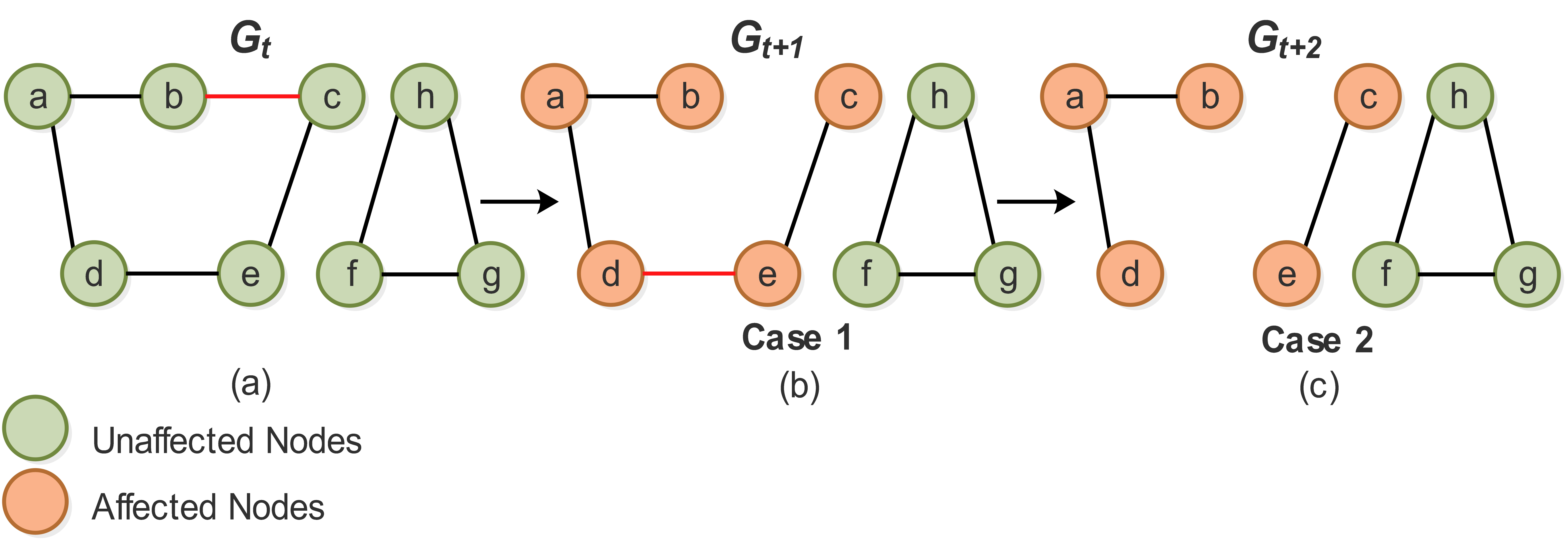}
    \caption{Illustration of cases after edge deletion.}
    \label{fig:cases}
  \end{figure*}

\noindent \textbf{Case 2: Split connected component (Bridge edge)}

\noindent When edge $(a,b)$ is a bridge edge, its deletion splits the connected component into two new connected components. The PC score of the nodes in the component containing node $a$ and node $b$ changes, i.e., $c(r) \neq c’ (r) \; \forall \,r \in C(a)$ or $C(b)$. Here, the number of affected nodes $ \mid A_a \mid $ and $ \mid A_b \mid $ are the number of nodes reachable from node $a$ and node $b$, respectively, and total affected nodes $ \mid A \mid  =  \mid A_a \mid  +  \mid A_b \mid $. Let \textcolor{blue}{Figure \ref{fig:cases}(b)} be the original network at $G_{t+1}$. When edge $(d, e)$ is deleted at $G_{t+2}$, the connected component splits into two and the PC score of all the nodes in both the split components changes, as shown in \textcolor{blue}{Figure \ref{fig:cases}(c)}.

\subsection{Fast Computation of Pairwise Connectivity Score}
\noindent We have previously shown the method to calculate the PC score of the nodes. However, as the structure of the network changes, it is important to update the PC score of the affected nodes. PC score of node $i$ is computed as:
\begin{equation*}
   c(i)=P(G)-P(G\backslash\{i\}).
\end{equation*}
Let $v \in R$ be a set of nodes not reachable from node $i$, i.e., $u(i, v) = 0\; \forall \,v \in R$. Therefore, nodes in $R$ do not contribute to the PC score of node $i$. Hence, it is not required to traverse the whole network to compute the PC score of node $i$, instead traversing the component to which node $i$ belongs is sufficient. This makes the PC score computing mechanism faster since we consider only the component to which the node belongs while ignoring the rest of the network. Updated PC score for node $i$ can be calculated as:
\begin{equation*}
c(i)=P(C(i))-P(C(i)\backslash\{i\}),\\
\end{equation*}
\begin{equation}
\label{eq_5}
    c(i)=\tbinom{ \mid C(i) \mid }{2}-{\displaystyle \sum_{1\text{\ensuremath{\le}}j\text{\ensuremath{\le}}r}}\tbinom{ \mid C_j \mid }{2}.\\
\end{equation}

Let node $i$ connects the component $C_1,..,C_r$, where $r$ denotes the number of distinct components containing neighbors of node $i$. The first term in \textcolor{blue}{Equation} \ref{eq_5} gives the PC score of the component containing node $i$, and the second term gives the PC score of this component without node $i$. The difference between both the terms gives the PC score of node $i$. Here, $P(C(i))$ denotes the PC score of the component containing node $i$.

\subsection{Updating Top-\textit{k} SHSs}\label{sec4.4}

\noindent This section discusses the procedure for updating Top-$k$ SHSs. We use \textcolor{blue}{Algorithm \ref{ch_6_alg_1}} to obtain the initial SHS set and PC score of the nodes in the original network. In addition, we use max-heap priority queue $Q$, where the nodes are sorted by their PC score. The top node $w$ in the priority queue has the highest PC score $c(w)$, among all the nodes in the network. After every update, the PC scores in the priority queue $Q$ change according to \textcolor{blue}{Lemma \ref{l_2}} and \textcolor{blue}{Lemma \ref{l_3}}. Besides, we have maintained a min-heap priority queue for the SHS nodes in Top-$k$.

\begin{lemma}
\label{l_2}
\normalfont \textit{\textbf{Given a graph $G = (V, E)$ and an edge update $(a,b)$, if the deleted edge $(a,b)$ is a non-bridge edge, then $c'(v) \geq c(v)$, $\forall \,v \in A$.}}
\end{lemma}
\noindent \begin{proof}
Deletion of a non-bridge edge $(a,b)$ may bring some nodes of the graph into the bridging position, which leads to a higher PC score of these nodes compared to their previous PC score. For the rest of the nodes, their new PC score remains the same.
\end{proof}

\begin{lemma}
\label{l_3}
\normalfont \textit{\textbf{Given a graph $G = (V,E)$ and an edge update $(a, b)$, if the deleted edge $(a, b)$ is a bridge edge, then $c'(v) < c(v)$, $\forall \,v \in A$.}}
\end{lemma}
\noindent \begin{proof}
Deletion of a bridge edge $(a,b)$ splits the connected component into two components. In the updated graph, the nodes in the split connected components are now pairwise connected to a smaller number of nodes, due to which their updated PC score is less as compared to their previous PC score, i.e., $c’ (v) < c(v)$, $\forall \, v \in A$.\end{proof}

\begin{algorithm}[t!]
\caption{Tracking-SHS algorithm: Updating Top-$k$ SHSs on deletion of an edge}
 \label{ch6_alg2}
 \begin{algorithmic}[1]
 \renewcommand{\algorithmicrequire}{\textbf{Input:}}
 \renewcommand{\algorithmicensure}{\textbf{Output:}}
 \REQUIRE Graph $G(V, E)$,  Old SHS set Top-$k$,  Deleted edge $(a, b)$, Priority queue $Q$ with nodes sorted by PC score $c$ (The PC scores in $Q$ change according to \textcolor{blue}{Lemma \ref{l_2}} and \textcolor{blue}{\ref{l_3}})
 \ENSURE Updated SHS set Top-$k$
 \STATE Determine if edge $(a,b)$ is a bridge or non-bridge edge
 \STATE Identify set of affected nodes $A$
 \FORALL{$v\in A$}
   \STATE Compute $c'(v)$
   \STATE $Q(v)\leftarrow c'(v)$
 \ENDFOR
  
 \FORALL{$v\in$ Top-$k$}
   \STATE Compute $c'(v)$
 \ENDFOR
  
 \WHILE{$Q$ is not empty}
   \STATE $w\leftarrow Q.getMax()$
   \IF {$c(w)\leq $ Top-$k.getMin()$}
      \RETURN Top-$k$
   \ELSIF{$w\in$ Top-$k$}
     \STATE $G=G\backslash\{w\}$
     \STATE update $Q$
   \ELSE
     \STATE Top-$k.removeMin()$
     \STATE Top-$k.insert(c(w),w)$
     \STATE $G=G\backslash\{w\}$
     \STATE update $Q$
   \ENDIF
 \ENDWHILE
\end{algorithmic}
\end{algorithm}

\textcolor{blue}{Algorithm \ref{ch6_alg2}} presents the Tracking-SHS procedure for updating Top-$k$ SHSs. 
The algorithm works as follow. When an edge $(a,b)$ is deleted from the network, it is first determined if it is a bridge or non-bridge edge. We run DFS from $a$ and stop as we hit $b$. If we hit $b$, it indicates that edge $(a,b)$ is a non-bridge edge, otherwise bridge edge. We then identify the set of affected nodes using the procedure discussed in \textcolor{blue}{Section \ref{sec4.1}}. We now compute the new PC score of the affected nodes using \textcolor{blue}{Equation} \ref{eq_5} and update the PC score of these nodes into the priority queue $Q$. The PC score of the nodes in Top-$k$ may also change due to updates in the network, and therefore, we need to update the PC score of affected nodes in Top-$k$. Once we have updated the PC scores of all the nodes in the network, we update the Top-$k$ SHS set. Following \textcolor{blue}{Lemma \ref{l_2}}, we know that the updated PC score of the affected nodes either increases or remains the same (in the case of non-bridge edge). Since the PC score of some nodes in the network may increase, we need to determine if any of these affected nodes have a higher PC score than existing SHS nodes in Top-$k$ and add such nodes in the SHS set. Besides, following \textcolor{blue}{Lemma \ref{l_3}}, we know that the updated PC score of the affected nodes is always less compared to their previous PC score (in case the deleted edge is a bridge edge). Since the updated PC score of the affected nodes decreases, we need to replace the affected nodes present in Top-$k$ with the high PC score non-SHS nodes in the network.

Let $w$ denotes the node with the maximum PC score in $Q$. Now, we compare the PC score of $w$, i.e., $c(w)$ with the minimum PC score node in Top-$k$, and if $c(w) \leq $ Top-$k.getMin()$, we terminate the algorithm and return Top-$k$. In contrast, if $c(w) >$ Top-$k.getMin()$, and node $w$ is already present in Top-$k$, we remove $w$ from the network and update the PC score of the nodes in the component containing node $w$ in $Q$. On the other hand, if node $w$ is not present in Top-$k$, we remove the minimum PC score node from Top-$k$ and add node $w$ to Top-$k$. Finally, $w$ is removed from $G$, and the PC score of the nodes in the priority queue is updated.

\begin{lemma}
\label{l_4}
\normalfont \textit{\textbf{Given a graph $G =(V, E)$ and an edge update $(a, b)$, \textcolor{blue}{Algorithm \ref{ch6_alg2}} replaces a maximum of $(k-k')$ nodes from the Top-$k$ SHS set on the deletion of a bridge edge.}}
\end{lemma}

\noindent \begin{proof} Following \textcolor{blue}{Lemma \ref{l_3}}, when the deleted edge $(a,b)$ is a bridge edge, the PC score $c(i) \, \forall \, i\in A$ decreases, whereas the PC score $c(i)$ $\forall$ $i \notin A$ remains the same. There may be some nodes that are affected and are present in Top-$k$ SHS set. Such nodes need to be replaced with high PC score non-SHS nodes from the network. Consider Top-$k_{rem}$, of size $ \mid $Top-$k_{rem} \mid  = k'$ as the set of nodes in Top-$k$ which are not affected, i.e., Top-$k_{rem}$ = Top-$k\backslash A$. Therefore, we have $(k-k’)$ affected nodes in Top-$k$, and at most $(k-k’)$ nodes in Top-$k$ can be replaced by high PC score non-SHS nodes in the network.\end{proof}

\subsection{Updating Top-{\textit{k}} SHSs for Batch Updates}
\noindent In the case of batch updates, a number of edges may be deleted from the network. 
One solution to update Top-$k$ SHSs for a batch update is to apply our proposed single edge update algorithm, i.e., Tracking-SHS after every update. However, this approach may not be efficient for a large number of updates. Therefore, we propose an efficient method that works as follows.

Let us consider that a set of batch updates consists of $l$ individual updates where $l\in(1,m)$. When a batch of $l$ edges is deleted from the network, each edge can either be a bridge edge or a non-bridge edge. We first determine the set of connected components in the network after the deletion of a batch of edges. Then, identify the set of affected nodes due to the updates in the network. For $l$ single updates, each update $i$ has its own affected node set $A_i,\, i = 1, 2, 3,...l$. For batch update, the total affected node set $A$ is the union of $A_i$, for all $i \in (1, l)$.
\begin{equation*}
    A=\bigcup_{i=1}^{l}A_{i}.
\end{equation*}
Once we have an affected node set due to batch updates, the PC scores of the affected nodes are recomputed using the efficient PC score computing function (\textcolor{blue}{Equation} \ref{eq_5}). Finally, update Top-$k$ SHSs set using the procedure discussed in \textcolor{blue}{Section \ref{sec4.4}}. In the case of a large number of updates, processing batch updates is more efficient than sequentially processing each update. For instance, if a node is affected several times during serial edge update, we need to recompute its PC score every time it is affected. However, in the case of a batch update, we need to recompute its PC score only once, making the batch update procedure more efficient.

\section{Proposed Model: GNN-SHS}\label{sec5}
\noindent Inspired by the recent advancement of graph neural network techniques on various graph mining problems, we propose \textbf{\textit{GNN-SHS}}, a graph neural network-based framework to discover Top-$k$ SHS nodes in the dynamic network. This model considers the situations where it is important to handle incremental as well as decremental updates in the network, such as Facebook, where links appear and disappear whenever a user friend/unfriend others. Due to dynamic nature of network, discovered SHSs change; therefore, it is crucial to design a model that can efficiently discover SHSs as the network evolves. \textcolor{blue}{Figure \ref{fig:gnn_arch_ch6}} represents the architecture of the proposed GNN-SHS model. We divided the SHSs identification process into two parts, i.e., model training and model application. The details of the GNN-SHS model are discussed below.

\subsection{Model Training} \label{model_train}
\noindent This section discusses the architecture of the proposed model and the training procedure.

\subsubsection{Architecture of GNN-SHS} 
\noindent In order to discover SHS nodes in dynamic network, we first transform the SHSs identification problem into a learning problem. We then propose a GNN-SHS model that uses the network structure as well as node features to identify SHS nodes. Our model utilizes three-node features, i.e., effective size \cite{burt1992structural}, efficiency \cite{burt1992structural} and degree, to characterize each node. These features are extracted from the one-hop ego network of the node. Given a graph and node features as input, our proposed model GNN-SHS first computes the low-dimensional node embedding vector and then uses the embedding of the nodes to determine the label of nodes (as shown in \textcolor{blue}{Figure \ref{fig:gnn_arch_ch6}}). The label of a node can either be SHS or normal. The procedure for generating embeddings of the nodes is presented in \textcolor{blue}{Algorithm \ref{gnn-algo_ch6}}. The model training is further divided into two phases: 1) Neighborhood Aggregation, 2) High Order Propagation. The two phases of the GNN-SHS model are discussed below:

\begin{algorithm}[t!]
\caption{Generating node embeddings for GNN-SHS}
 \label{gnn-algo_ch6}
 \begin{algorithmic}[1]
 \renewcommand{\algorithmicrequire}{\textbf{Input:}}
 \renewcommand{\algorithmicensure}{\textbf{Output:}}
 \REQUIRE Graph $G(V,E)$, Input features $\vec{x}(i),\,\, \forall i \in V$, Depth $L$, Weight matrices $W^{l},\,\, \forall l \in \{1,..,L\}$, Non-linearity $\sigma$
\ENSURE Node embeddings $z{(i)}, \,\, \forall i \in V$ 
 \STATE $h^{0}(i) \leftarrow \vec{x}(i),\,\,\forall i \in V$ 
 \FOR{$l = 1$ to $L$}
 \FOR{$i \in V$}
  \STATE Compute $h^{(l)}{(N(i))}$ using \textcolor{blue}{Equation \ref{eq:agg}}
  \STATE Compute $h^{(l)}{(i)}$ using \textcolor{blue}{Equation \ref{eq:comb}}
 \ENDFOR
 \ENDFOR
 \STATE $z{(i)} = h^{(L)}{(i)}$
\end{algorithmic}
\end{algorithm}

\noindent \textbf{Neighborhood Aggregation.} 
The neighborhood aggregation phase aggregates the features from the neighbors of a node to generate node embeddings. The node embeddings are the low dimensional representation of a node. Due to distinguishing characteristics exhibited by the SHSs, we considered all the one-hop neighbors of the node to create embedding. We generate the embeddings of node $i$ by aggregating the embeddings from its neighboring nodes, and we use the number of neighbors of node $i$ as weight factor:
\begin{equation}
\label{eq:agg}
h^{(l)}{(N(i))} = \sum_{j\in N(i)}{\frac{h^{(l-1)}{(j)}}{ \mid N(i) \mid }},
\end{equation}
where $h^{(l)}{(N(i))}$ represents the embedding vectors captured from the neighbors of node $i$. Embedding vector of each node is updated after aggregating embeddings from its neighbors. Node embeddings at layer $0$, i.e., $h^{0}(i)$ are initialized with the feature vectors $\vec{x}(i)$ of the nodes, i.e., Effective size, Efficiency and Degree. Each node retains its own feature information by concatenating its embedding vector from the previous layer with the aggregated embedding of its neighbors from the current layer as:
\begin{equation}
\label{eq:comb}
h^{(l)}{(i)} = \sigma{\Big(W^{l}\big(h^{(l-1)}{(i)} \mathbin\Vert h^{(l)}{(N(i))}\big) \Big)},
\end{equation}
where $W^{l}$ are the training parameter, $\mathbin\Vert$ is the concatenation operator, and $\sigma$ is the non-linearity, e.g., ReLU.

\noindent \textbf{High Order Propagation.} Our model employs multiple neighborhood aggregation layers in order to capture features from $l$-hop neighbors of a node. The output from the previous layer acts as input for the current layer. Stacking multiple layers will recursively form the representation $h^{(l)}(i)$ for node $i$ at the end of layer $l^{th}$ as: 
\begin{equation*}
z{(i)} = h^{(l)}{(i)}, \,\,\,\, \forall i \in V
\end{equation*}
where $z(i)$ denotes the final node embedding at the end of $l^{th}$ layer. For the purpose of classifying the nodes as SHS or normal node, we pass the final embeddings of each node $z{(i)}$ through the Softmax layer. This layer takes node embeddings as input and generates the probability of two classes: SHS and normal. We then train the model to distinguish between SHS and normal nodes.

\begin{figure*}[t!]
 \centering
 \includegraphics[width=0.6\paperwidth]{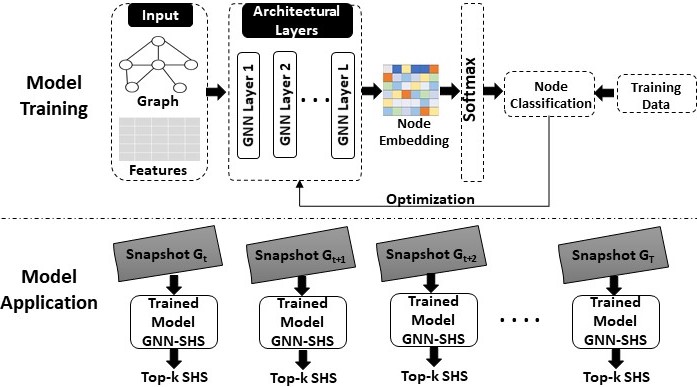}
 \caption{Architecture of proposed GNN-SHS model.}
 \label{fig:gnn_arch_ch6}
\end{figure*}

\subsubsection{Training Procedure} 
\noindent To discover SHS nodes in the network, we employ \textit{binary cross-entropy loss} to train the model with the actual label known for a set of nodes. The loss function $\mathcal{L}$ is computed as:

\begin{equation*}
\label{eq_loss_ch6}
\mathcal{L} = {-}\frac{1}{ r}{\sum_{i=1}^{r}\bigg(y(i)\log{\hat{y}(i)} + (1-y(i)) \log{(1-\hat{y}(i))\bigg)}},
\end{equation*}
where $y$ is the true label of a node, $\hat{y}$ is the label predicted by the GNN-SHS model, and $r$ is the number of nodes in the training data for which the labels are known.

\subsection{Model Application}
\noindent We first train the GNN-SHS model using the labelled data as discussed in \textcolor{blue}{Section \ref{model_train}}. Once the model is trained, we then utilize the trained GNN-SHS model to quickly discover SHS nodes in each snapshot of graph (snapshot obtained from the dynamic network). The discovered SHSs are nodes removal of which minimizes the total pairwise connectivity in the remaining subgraph.

\begin{theorem} \label{thm2_ch6}
\textit{\textbf{The depth of the proposed graph neural network-based GNN-SHS model with \textit{width} $=O(1)$ should be at least $\Omega({n}^2/\log^2 n)$ to solve the SHSs problem.} 
}
\end{theorem}

\noindent \begin{proof}
In \textcolor{blue}{Theorem} \ref{thm1_ch6}, we proved that if we can discover SHS nodes in the graph, then we can solve the VC problem. {Corollary 4.4} of Loukas \cite{loukas2019graph} showed that for solving the minimum VC problem, a message-passing GNN with a \textit{ width} $=O(1)$  should have a depth of at least $\Omega({n}^2/\log^2 n)$. Therefore, the lower bound on the depth of the VC problem also applies to our SHSs problem. Here, the \textit{depth} describes the number of layers in the GNN-SHS model, and \textit{width} indicates the number of hidden units.
\end{proof}

\section{Theoretical Analysis of Tracking-SHS algorithm}\label{sec6}
\noindent This section presents the theoretical performance analysis of the proposed Tracking-SHS algorithm. We first analyze the run time of the single edge update algorithm for various cases discussed in \textcolor{blue}{Section \ref{sec4}}. We then give an overall run time of the proposed algorithm. Later, we analyze the performance of the batch update algorithm.
When there are more connected components in the network, the time to update Top-$k$ SHSs will be less as there will be fewer affected nodes. In addition, the time to recompute the PC score of a node will also be less due to the small size of connected components. We consider the worst-case by assuming only one connected component in the initial network. We take the value of $k$ as 1. It takes $(n+m)$ time to determine if edge $(a,b)$ is a bridge or non-bridge. We run DFS to identify the affected node set, and the time required for the same is $(n+m)$. The PC score of a node is calculated using \textcolor{blue}{Equation} \ref{eq_5}. Notably, the time to update Top-$k$ SHSs is the time to recompute the PC scores of the affected nodes, i.e., $( \mid A \mid \times\, m)$, and a detailed analysis of $A$ and $m$ is shown below.

\begin{lemma}
\label{l_5}
\normalfont \textit{\textbf{Given a graph $G = (V,E)$, old SHS set Top-$k$ and an edge update $(a, b)$, \textcolor{blue}{Algorithm \ref{ch6_alg2}} takes $mn$ time to update the Top-$k$ SHS set (In case of non-bridge edge).}}
\end{lemma}
\noindent \begin{proof}
In this case, the number of affected nodes $ \mid A \mid  =  \mid C(a) \mid $ and since there is only 1 connected component in the graph, the number of edges $m(a)$ in the connected component containing node $a$ after an edge update $(a,b)$ is $(m-1)$ and the number of affected nodes $ \mid A \mid $ can be at most $n$. Therefore, the time $T_{ab}$ required to update the spanner set can be at most $m(a)\times  \mid A \mid $, i.e., $mn$. Here, $T_{ab}$ denotes the time required by the proposed algorithm Tracking-SHS to update Top-$k$ SHSs after an edge update $(a,b)$. \end{proof}

\begin{theorem}
\label{t_1_ch6}
\normalfont \textit{\textbf{Given a graph $G = (V,E)$, old SHS set Top-$k$ and an edge update $(a,b)$, \textcolor{blue}{Algorithm \ref{ch6_alg2}} takes $\frac{5mn}{8}$ time to update Top-$k$ SHS set in the expected case (In case of bridge edge).}}
\end{theorem}

\noindent \begin{proof}
Deletion of a bridge edge $(a,b)$ splits the connected component into two components, one containing node $a$ and the other containing node $b$. The probability of deletion of any edge from the graph is uniform, i.e., $\frac{1}{m}$. Then, the size of each connected component after deletion of edge $(a,b)$ is:
\begin{align*}
\begin{split}
    m(a) {}& = i\\
\end{split}\\
\begin{split}
    m(b) {}& = (m-1-i)
\end{split}
\end{align*}
where $\,\, 0\le\,i\,\text{\ensuremath{\le}}\,\frac{m-1}{2}$\\

Here, we assume that deletion of a bridge edge $(a,b)$ results in two components and the number of edges in each split component follow the uniform distribution. Using this assumption, the expected size $\bar{m}(a)$, $\bar{m}(b)$ of the connected components is computed as:
\begin{align*}
\begin{split}
    \bar{m}(a) {}& = \sum_{i=0}^{\frac{m-1}{2}}i \times \frac{2}{m+1}\\
                & = \frac{(m-1)}{4}
\end{split}
\end{align*}
\begin{align*}
\begin{split}
    \bar{m}(b) {}& = \sum_{i=0}^{\frac{m-1}{2}}(m-1-i) \times \frac{2}{m+1}\\
                & = \frac{3(m-1)}{4}
\end{split}
\end{align*}

After an edge update $(a,b)$, the affected nodes are the set of nodes reachable from node $a$ and $b$. We use the property that the maximum number of nodes that may span $i$ edges in one connected component are $(i+1)$ to compute the number of affected nodes. The expected number of affected nodes $ \mid  \bar{A}_a \mid $, $ \mid \bar{A}_b \mid $ in both the components is calculated as:
\begin{align*}
\begin{split}
     \mid \bar{A}_a \mid  {}& = \sum_{i=0}^{\frac{m-1}{2}}(i+1)\times \frac{2}{m+1}\\
                & = \frac{(m-1)}{4}+1
\end{split}\\
\begin{split}
     \mid \bar{A}_b \mid  {}& = \sum_{i=0}^{\frac{m-1}{2}}(n-i-1)\times \frac{2}{m+1}\\
                & = n-(\frac{m-1}{4})-1
\end{split}
\end{align*}

To compute the number of affected nodes, we take the upper bound on the number of edges. Using the same property, we know that $m$ edges can be spanned by a maximum of $(m+1)$ nodes, but we have a maximum of $n$ nodes in the graph.\\

We put $m=n-1$ to compute the expected number of affected nodes.
\begin{align*}
\begin{split}
     \mid \bar{A}_a \mid  {}& = \frac{n-2}{4}+1\\
\end{split}\\
\begin{split}
     \mid \bar{A}_b \mid  {}& = n-(\frac{n-2}{4})-1\\
\end{split}
\end{align*}

Therefore, the expected time $\bar{T}_{ab}$ required to update the SHS set for this case is:
\begin{align*}
\begin{split}
    \bar{T}_{ab} {}& =  \mid \bar{A}_a \mid \times \,\bar{m}(a) +  \mid \bar{A}_b \mid \times \,\bar{m}(b)\\
\end{split}\\
\begin{split}
    \bar{T}_{ab} {}& = \frac{5mn}{8}\\
\end{split}
\end{align*}
Hence proved. \end{proof}

\begin{theorem}
\label{t_2_ch6}
\normalfont \textbf{\textit{For the general case, the overall time ${T}_{ab}$ required by the proposed algorithm to update Top-$k$ SHSs after an edge update $(a,b)$ is $P_{ab}(\frac{5mn}{8})+(1-P_{ab})(mn)$.}}
\end{theorem}

\noindent \begin{proof} Let $P_{ab}$ be the probability that edge $(a,b)$ is a bridge edge, removal of which splits the connected component into two components (Case 2), then the probability of an edge being a non-bridge is $(1-P_{ab})$ (Case 1). Using the results of \textcolor{blue}{Lemma \ref{l_5}} and \textcolor{blue}{Theorem \ref{t_1_ch6}}, the overall time for updating Top-$k$ SHSs after an edge update $(a,b)$ is given by ${T}_{ab}$ = $P_{ab}(\frac{5mn}{8})+(1-P_{ab})(mn)$.\end{proof}

\noindent The following special case stems from \textcolor{blue}{Theorem \ref{t_2_ch6}}. 

\noindent \textbf{\textit{Special case.}} For the bridge-edge dominating graph where all the edges in the graph are bridges, the probability of an edge being a bridge is 1, i.e., $P_{ab} = 1$ and non-bridge edge is $0$. Theoretically, the speedup of the proposed algorithm is $P_{ab}(\frac{5mn}{8})+(1-P_{ab})(mn)$ and substituting the values of $P_{ab}$ and $(1-P_{ab})$ gives us the overall update time, i.e., $T_{ab} = \frac{5mn}{8}$. In contrast, the time required for static recomputation is $T’_{ab} = (kn(m+n))$, i.e., $mn$ where $k$ is constant. Therefore, the proposed algorithm achieves an overall speedup of 1.6 times (speedup = $\frac{T’_{ab}}{T_{ab}}$) over recomputation. An example of such a graph is the Preferential Attachment graph. We use experimental analysis to validate our theoretical results, and the experimental results support our arguments.

\vspace{0.2in}
\begin{theorem}
\label{t_3_ch6}
\normalfont \textit{\textbf{Given a graph $G = (V,E)$, old SHS set Top-$k$, and a batch of $l$ edge updates, the batch update algorithm, i.e., Tracking-SHS algorithm for batch updates takes $(N\_CC \times S\_CC \times A\_CC \times P_{l})$ time to update the Top-$k$ SHS set. Here, $N\_CC$ refers to the number of connected components in the resulting graph, $A\_CC$ is the number of affected nodes in the connected component, $S\_CC$ is the size of the connected component, and $P_l$ is the probability that deletion of $l$ edges results in $N\_CC$ in the updated graph.}}
\end{theorem}
\noindent \begin{proof}
The proof is straightforward and thus omitted. In \textcolor{blue}{Section \ref{sec7.3}}, we will show that our batch update algorithm produces good results over static recomputation. \end{proof}

\section{Experimental Results}\label{sec7}

\noindent This section analyzes the performance of the proposed Tracking-SHS algorithm and GNN-SHS model. We first discuss the datasets used to evaluate the performance. We then evaluate the performance of the Tracking-SHS algorithm for single update followed by the batch of updates. Later, we discuss the performance of GNN-SHS model. We implemented our algorithms in Python 3.7. The experiments are performed on a Windows 10 PC with CPU 3.20 GHz and 16 GB RAM.

\subsection{Datasets} 

\noindent We measure the update time of the proposed Tracking-SHS algorithm, static recomputation and GNN-SHS model by conducting extensive experiments on various real-world and synthetic datasets. The details of the datasets are discussed below.

\subsubsection{Real-World Datasets}

\noindent We analyze the performance of the proposed algorithms on four real-world networks having different sizes. Karate \cite{zachary1977information} is a friendship network among the members of the karate club. Dolphin \cite{lusseau2003bottlenose} is a social network representing the frequent association between 62 dolphins of a community. American College Football \cite{girvan2002community} is a football games network between Division IA colleges, and HC-BIOGRID\footnote{\url{https://www.pilucrescenzi.it/wp/networks/biological/}} is a biological network. The characteristics of the real-world datasets are summarized in \textcolor{blue}{Table \ref{dataset1}}.

\subsubsection{Synthetic Datasets}

\noindent  We analyze the performance of proposed algorithms by conducting experiments on synthetic datasets. We generate synthetic networks using graph-generating algorithms and vary the network size to determine its effect on algorithms performance. We conduct experiments on synthetic networks with diverse topologies: Preferential Attachment (PA) networks and Erdos-Renyi (ER) \cite{erdHos1959random} networks. We generate PA$(n)$ network with 500, 1000 and 1500 nodes, where $n$ denotes total nodes in the network. For ER$(n,p)$, we generate networks with 250 and 500 nodes, where $p$ is the probability of adding an edge to the network. In PA$(n)$ network, a highly connected node is more likely to get new neighbors. In ER$(n,p)$ network, parameter $p$ acts as a weighting function, and there are higher chances that the graph contains more edges as $p$ increases from 0 to 1. The properties of synthetic datasets are presented in \textcolor{blue}{Table \ref{dataset2}}. 

\begin{table}[t!] 
\caption{Summary of real-world datasets.}
\label{dataset1}
\renewcommand{\arraystretch}{1.1}
\centering 
\footnotesize
\begin{tabular}{llll} \hlineB{1.5} 
\textbf{Dataset} & \textbf{Nodes} & \textbf{Edges} & \textbf{Avg degree} \\ \hlineB{1.5}
Karate & 34 & 78 & 4.59 \\ 
Dolphins & 62 & 159 & 5.13\\ 
Football & 115 & 613 & 10\\ 
HC-BIOGRID & 4039 & 14342 & 7\\ \hlineB{1.5}
\end{tabular} 
\end{table}

\begin{table}[t!] 
\caption{Summary of synthetic datasets.}
\label{dataset2}
\renewcommand{\arraystretch}{1.1}
\centering 
\footnotesize
\begin{tabular}{llll} \hlineB{1.5} 
\textbf{Dataset} & \textbf{Nodes} & \textbf{Edges} & \textbf{Avg degree} \\ \hlineB{1.5}
PA (500) & 500 & 499 & 2 \\ 
PA (1000) & 1000 & 999 & 2\\ 
PA (1500) & 1500 & 1499 & 2\\  
ER (250, 0.01) & 250 & 304 & 2\\ 
ER (250, 0.5) & 250 & 15583 & 124\\ 
ER (500, 0.04) & 500 & 512 & 2\\ 
ER (500, 0.5) & 500 & 62346 & 249\\ \hlineB{1.5} 
\end{tabular} 
\end{table}

\subsection{Performance of Tracking-SHS algorithm on Single Update}
\subsubsection{Performance on Real-World Dataset}

\noindent To evaluate the  performance of the Tracking-SHS algorithm for single edge update, we compare it against static recomputation. \textcolor{blue}{Table \ref{result_dataset1}} shows the speedup achieved by the Tracking-SHS algorithm over recomputation. Speedup is the ratio of the speed of the static algorithm to that of Tracking-SHS algorithm, which is proportional to the ratio of computation time used by the static algorithm to that by the proposed Tracking-SHS algorithm. In order to determine how the two algorithms (static algorithm and proposed dynamic Tracking-SHS algorithm) perform for the dynamic network, we start with a  full network and randomly remove 50 edges, one at a time.

\begin{table*}[t!] 
\caption[Speedup of Tracking-SHS algorithm on static recomputation on real-world datasets.]{Speedup of Tracking-SHS algorithm on static recomputation over 50 edge deletions on real-world datasets.}
\label{result_dataset1}
\renewcommand{\arraystretch}{1.1}
\centering
\resizebox{\textwidth}{!}{\begin{tabular}{l|ccc|ccc|ccc} \hlineB{2} 
\textbf{Dataset} & \multicolumn{3}{c}{\textbf{\textit{k} = 1}} & \multicolumn{3}{|c}{\textbf{\textit{k} = 5}}& \multicolumn{3}{|c}{\textbf{\textit{k} = 10}}\\ \hlineB{1.5}
& \textbf{Gmean} & \textbf{Min} & \textbf{Max} & \textbf{Gmean} & \textbf{Min} & \textbf{Max} & \textbf{Gmean} & \textbf{Min} & \textbf{Max}\\ \hlineB{1.5}
Karate & 2.35 & 1.73 & 3.1 & 3.92 & 2.98 & 4.18 & 5.02 & 4.98 & 5.17 \\  
Dolphins & 3.34 &  2.11 & 4.18 & 4.16 & 3.06 & 5.33 & 7.52 & 5.21 & 9.22 \\  
Football & 3.72 & 3.42 & 4.21 & 10.17 & 9.6 & 11.47 & 17.26 & 15.45 & 19.84 \\  
HC-BIOGRID & 3.76 & 1.85 & 4.11 & 11.16 & 10.21 & 11.89 & 21.79 & 20.23 & 22.65 \\ \hlineB{1.5} 
\textbf{Mean (Geometric)} & \textbf{3.24} & \textbf{2.19} & \textbf{3.87} & \textbf{6.56} & \textbf{5.47} & \textbf{7.42} & \textbf{10.91} & \textbf{9.49} & \textbf{12.1}\\ \hlineB{1.5}
\end{tabular}}
\end{table*}

We then compute the geometric mean of the speedup for the proposed Tracking-SHS algorithm in terms of its execution time against the static algorithm. For instance, if the Tracking-SHS algorithm takes 10 seconds to execute for an edge update, whereas the static algorithm takes 50 seconds to execute for the same update, we say that the Tracking-SHS algorithm is 5 times faster than static recomputation. The column “Gmean” has the geometric mean of the achieved speedup (over 50 edge deletions), “Min” contains minimum achieved speedup and, “Max” contains maximum speedup. We run our Tracking-SHS algorithm for 3 different values of $k$, i.e., $k$ = 1, 5 and 10. For real-world networks, the gmean speedup is always at least 2.35 times for $k$ = 1, 3.92 for $k$ = 5 and 5.02 for $k$ = 10. The experimental results demonstrate that speedup increases with the value of $k$. The average speedup reaches 21.79 times for $k$ =10 from a speedup of 3.76 for $k$ = 1 (HC-BIOGRID dataset), which shows a significant improvement for a larger value of $k$. The average speedup over all tested datasets is 3.24 times for $k$ = 1, 6.56 for $k$ = 5, and 10.91 for $k$ = 10. The minimum speedup achieved by the proposed Tracking-SHS algorithm is 1.73 times for the Karate dataset, and the maximum speedup achieved is 22.65 times for HC-BIOGRID dataset. In addition, it has been observed that the speedup also increases with the size of the network. For a small size network (Karate dataset), the speedup is 5.02 times for $k$ = 10. In contrast, for the same value of $k$, the speedup increases significantly to 21.79 times for the large size network (HC-BIOGRID dataset).

\begin{table*}[t!]
\caption[Speedup of Tracking-SHS algorithm on static recomputation on synthetic datasets.]{Speedup of Tracking-SHS algorithm on static recomputation over 50 edge deletions on synthetic datasets.}
\label{result_dataset2}
\renewcommand{\arraystretch}{1.1}
\centering 
\resizebox{\textwidth}{!}{\begin{tabular}{l|lll|lll|lll} \hlineB{2} 
\textbf{Dataset} & \multicolumn{3}{c}{\textbf{\textit{k} = 1}} & \multicolumn{3}{|c}{\textbf{\textit{k} = 5}}& \multicolumn{3}{|c}{\textbf{\textit{k} = 10}}\\ \hlineB{1.5}
& \textbf{Gmean} & \textbf{Min} & \textbf{Max} & \textbf{Gmean} & \textbf{Min} & \textbf{Max} & \textbf{Gmean} & \textbf{Min} & \textbf{Max}\\ \hlineB{1.5}
PA (500) & 3.24 & 2.75 & 3.39 & 4.76 & 4.44 & 5.65 & 6.35 & 5.79 & 7.24 \\
PA (1000) & 3.32 & 3.08 & 3.83 & 5.3 & 4.46 & 6.12 & 8.31 & 7.75 & 9.31 \\
PA (1500) & 3.54 & 3.22 & 4.02 & 5.46 & 4.57 & 6.13 & 8.95 & 7.54 & 10.6\\
ER (250, 0.01) & 4.22 & 3.98 & 4.65 & 6.15 & 5.82 & 6.62 & 10.65 & 10.1 & 11.25 \\
ER (250, 0.5) & 4.08 & 3.9 & 4.17 & 11.66 & 10.42 & 13.15 & 20.4 & 18.34 & 22.21 \\ 
ER (500, 0.04) & 4.41 & 4.02 & 4.88 & 7.83 & 5.08 & 9.51 & 9.81 & 8.06 & 11.78 \\
ER (500, 0.5) & 3.95 & 3.78 & 4.35 & 11.39 & 10.33 & 12.52 & 21.44 & 20.35 & 29.08 \\  \hlineB{1.5}
\textbf{Mean (Geometric)} & \textbf{3.8} & \textbf{3.5} & \textbf{4.16} & \textbf{7.07} & \textbf{6.02} & \textbf{8.05} & \textbf{11.16} & \textbf{10.04} & \textbf{12.95}\\ \hlineB{1.5}
\end{tabular}}
\end{table*}

\subsubsection{Performance on Synthetic Dataset} 
\noindent We perform the similar experiments on diverse synthetic datasets of varying scales. We start with a full network and randomly remove 50 edges, one at a time. \textcolor{blue}{Table \ref{result_dataset2}} shows the speedup of the proposed Tracking-SHS algorithm over static recomputation. For synthetic datasets, the gmean of the achieved speedup is always at least 3.24 times for $k$ = 1, 4.76 for $k$ = 5 and 6.35 for $k$ = 10 (over 50 edge deletions). Similar to the real-world dataset, the speedup increases with the increase in the value of $k$ for the synthetic dataset. For instance, in ER (500, 0.5) network, the average speedup increases to 21.44 times for $k$ = 10 from a speedup of 3.95 for $k$ = 1. The average speedup over all tested datasets is 3.8 for $k$ = 1, 7.07 for $k$ = 5, and 11.16 for $k$ = 10. Besides, speedup for ER dataset is relatively higher than the PA dataset of the same scale. Take an example of the PA dataset of 500 nodes, the mean speedup is 6.35 times, whereas, for ER dataset of 500 nodes, the mean speedup is at least 9.81 for $k$=10.

\begin{figure}[ht!]
  \centering
    \includegraphics[width=0.34\paperwidth]{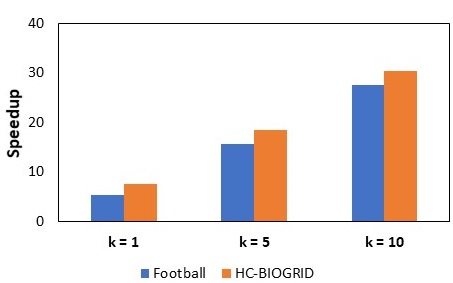}
    \caption{Speedup of Tracking-SHS algorithm on batch update.}
    \label{fig:batch}
  \end{figure}

\subsection{Performance of Tracking-SHS algorithm on Batch Update}\label{sec7.3}
\noindent We now demonstrate the performance of the batch update algorithm, i.e., Tracking-SHS algorithm for batch updates. Initially, we consider all the edges are present in the network and compute the Top-$k$ SHSs using \textcolor{blue}{Algorithm \ref{ch_6_alg_1}}. We then randomly remove a set of 50 edges from the network. These edges are considered as a batch of updates, and our goal is to update Top-$k$ SHS set corresponding to these updates. To evaluate the performance of our batch update algorithm, we performed experiments on the football and HC-BIOGRID dataset, and set the value of $k$ to 1, 5 and 10. \textcolor{blue}{Figure \ref{fig:batch}} shows the speedup achieved by the Tracking-SHS algorithm on batch updates. We attained a speedup of 5.29 times for the football dataset and 7.51 times for the HC-BIOGRID dataset (for $k$ =1). Experimental results demonstrate that for batch update, the speedup increases with the increase in the value of $k$, e.g., for HC-BIOGRID, speedup is 7.51 times for $k$ =1, whereas it is 30.35 times for $k$ =10.

\subsection{Performance of GNN-SHS model}
\noindent In dynamic network, as the graph changes with time, we get multiple snapshots of graph. The trained model can be used to identify the SHSs in each snapshot of the dynamic graph. Even if it takes some time to train the model, we need to train it only once, and after that, whenever the graph changes, the trained GNN-SHS model can be used to discover the updated SHSs in a few seconds.

\subsubsection{Baseline}
\noindent In the literature, there is no solution that addresses the problem of discovering SHSs in dynamic networks. Therefore, we compare the performance of our model GNN-SHS with the proposed Tracking-SHS algorithm (as discussed in \textcolor{blue}{Section \ref{sec4}}).

\subsubsection{Evaluation Metrics}
\noindent We measure the \textit{efficiency} of our proposed GNN-SHS model in terms of speedup achieved by GNN-SHS over Tracking-SHS algorithm. The speedup is computed as follows:
\begin{equation*}
\label{agg_ch6}
\text{Speedup} = \frac{\text{Run time of Tracking-SHS algorithm}}{\text{Run time of proposed GNN-SHS model}}.
\end{equation*}

In addition, we measure the \textit{effectiveness} of our model in terms of classification accuracy achieved by the model. 

\subsubsection{Ground Truth Computation} 
\noindent To compute the ground truth labels, we first calculate the pairwise connectivity score $c$ using \textcolor{blue}{Equation} \ref{pc_score} for each node and then label Top-$k$ nodes with the highest score as SHS nodes and rest as normal nodes. For experimental analysis, we set the value of $k$ to 50. 

\subsubsection{Training Settings} 
\noindent We implemented the code of GNN-SHS in PyTorch and fixed the number of layers to 2 and the embedding dimension to 32. The model parameters are trained using Adam optimizer, learning rate of 0.01 and weight decay of $5e-4$. We train GNN-SHS model for 200 epochs. We used 60\% of the nodes for training, 20\% for validation, and 20\% for testing. In addition, we used an inductive setting where test nodes are unseen to the model during the training phase.

\subsubsection{Performance of GNN-SHS on Synthetic Dataset}
\noindent Tracking-SHS algorithm works for a single edge deletion update. Therefore, we analyze GNN-SHS performance on a single edge deletion update only so that we can compare the speedup of GNN-SHS over Tracking-SHS algorithm. To determine the speedup of the proposed GNN-SHS model over the proposed Tracking-SHS algorithm, we start from the whole network and arbitrarily delete 50 edges; we only delete one edge at a time. In this way, we obtain multiple snapshots of graph. We set the value of $k$ (number of SHS) to 50 and make use of the trained GNN-SHS model to discover SHS nodes in each new snapshot graph. We calculate the geometric mean of the speedup achieved by GNN-SHS over the Tracking-SHS algorithm. \textcolor{blue}{Table \ref{ACC_GNN-SHS}} reports the classification accuracy (SHS detection accuracy) achieved by GNN-SHS on various synthetic graphs. GNN-SHS achieves a minimum accuracy of 94.5\% on Preferential Attachment graph PA(500) and 86\% classification accuracy for Erdos-Renyi graph ER(250, 0.01). The results from the table show that graph neural network-based models achieve high SHS classification accuracy. 

\textcolor{blue}{Table \ref{SPEED_GNN-SHS}} reports the speedup achieved by the proposed GNN-SHS model over the Tracking-SHS algorithm. {GNN-SHS model achieved high speedup over the Tracking-SHS algorithm while sacrificing a small amount of accuracy.} {The proposed model GNN-SHS is at least 31.8 times faster for ER(250, 0.01) network and up to 2996.9 times faster for PA(1500) over the Tracking-SHS algorithm, providing a considerable efficiency advantage.} The geometric mean speedup is always at least 37.5 times, and the average speedup over all tested datasets is 671.6 times. Results show that our graph neural network-based model GNN-SHS speeds up the SHS identification process in dynamic networks. In addition, it has been observed from the results that the speedup increases as network size increases, e.g., for PA graphs, the geometric mean speedup is 1236.4 times for a graph of 500 nodes, 1930.6 times for a graph with 1000 nodes and 2639.7 times for a graph with 1500 nodes. In \textcolor{blue}{Theorem \ref{thm2_ch6}}, we showed that the depth of GNN-SHS should be at least $\Omega({n}^2/\log^2 n)$ to solve the SHSs problem. Nevertheless, a deeper graph neural network suffers from an over-smoothing problem \cite{li2018deeper, yang2020toward}, making it challenging for GNN-SHS to differentiate between the embeddings of the nodes. In order to avoid the over-smoothing problem, we only used 2 layers in our GNN-SHS model.

\begin{table}[t!] 
\caption{Classification accuracy of GNN-SHS on synthetic datasets.}
\label{ACC_GNN-SHS}
\renewcommand{\arraystretch}{1.2}
\centering 
\footnotesize
\begin{tabular}{ll} \hlineB{1.5}
\textbf{Dataset} & \textbf{Accuracy}\\ \hlineB{1.5}
PA (500) & 94.5\%\\ 
PA (1000) & 95.33\% \\ 
PA (1500) & 96.5\%\\ 
ER (250, 0.01) & 86\%\\
ER (250, 0.5) & 90\%\\ 
ER (500, 0.04) & 87\%\\
ER (500, 0.5) & 92\%\\ \hlineB{1.5}
\end{tabular} 
\end{table}

\begin{table}[t!] 
\caption[Speedup of GNN-SHS over Tracking-SHS algorithm on synthetic datasets.]{Speedup of GNN-SHS model over Tracking-SHS algorithm over 50 edge deletions on synthetic datasets.}
\label{SPEED_GNN-SHS}
\renewcommand{\arraystretch}{1.2}
\centering 
\footnotesize
\begin{tabular}{llll} \hlineB{1.5}
\textbf{Dataset} & \textbf{Geometric Mean} & \textbf{Min} & \textbf{Max} \\ \hlineB{1.5}
PA (500) & 1236.4 & 1012.5 & 1532.7\\ 
PA (1000) & 1930.6 & 1574.2 & 2141.4\\ 
PA (1500) & 2639.7 & 2432.1 & 2996.9\\ 
ER (250, 0.01) & 37.5 & 31.8 & 40.2\\
ER (250, 0.5) & 287.3 & 263.6 & 301.2\\ 
ER (500, 0.04) & 368.2 & 354.5 & 379.3\\
ER (500, 0.5) & 2466.6 & 2015.3 & 2845.9\\ \hlineB{1.5}
\textbf{Mean (Geometric)} & 671.6 & 584.1 & 745.9\\ \hlineB{1.5}
\end{tabular} 
\end{table}

\subsubsection{Performance of GNN-SHS on Real-world Dataset}
\noindent We perform experiments on real-world datasets to determine the proposed model GNN-SHS performance for incremental and decremental batch updates. In the literature, no solution discovers SHSs for incremental and decremental batch updates; therefore, we can not compare our results with other solutions. We only report the results obtained from our experiments. We set the value of $k=5$ (number of SHSs). For each real-world dataset, we initiate with the whole network and then arbitrarily delete 5 edges from the network and add 5 edges to the network at once. In this manner, we obtain a snapshot of the graph. We then use our trained model GNN-SHS to discover SHSs in the new snapshot graph. Our empirical results in \textcolor{blue}{Table \ref{result-real}} show that our model discovers updated SHSs in less than 1 second for both Dolphin and American College Football datasets. Besides, our model achieves high classification accuracy in discovering SHSs for batch updates.

\begin{table}[t!] 
\caption{Run time and classification accuracy of GNN-SHS on real-world datasets.}
\label{result-real}
\renewcommand{\arraystretch}{1.2}
\centering 
\footnotesize
\begin{tabular}{p{2cm}ll} \hlineB{1.5} 
\textbf{Dataset} & \textbf{Run time (sec)}& \textbf{Accuracy} \\ \hlineB{1.5}
Dolphin & 0.002 & 76.92\% \\ 
Football  & 0.009 & 86.96\%\\ 
 \hlineB{1.5} 
\end{tabular} 
\end{table}

\section{Chapter Summary}\label{sec8}
\noindent The structural hole spanner discovery problem has various applications, including community detection, viral marketing, etc. However, the problem has not been studied for dynamic networks. In this chapter, we studied the SHS discovery problem for dynamic networks. We first proposed an efficient Tracking-SHS algorithm that maintains SHSs dynamically by discovering the affected set of nodes whose connectivity score updates as a result of changes in the network. We proposed a fast procedure for calculating the scores of the nodes. We also extended our proposed single edge update Tracking-SHS algorithm to a batch of edge updates. In addition, we proposed a graph neural network based model GNN-SHS, that discovers SHSs in the dynamic networks by learning low-dimensional embedding vectors of nodes. Finally, we analyzed the performance of the Tracking-SHS algorithm theoretically and showed that our proposed algorithm achieves a speedup of 1.6 times over recomputation for a particular type of graph, such as Preferential Attachment graphs. Besides, our experimental results demonstrated that the Tracking-SHS algorithm is at least 3.24 times faster than the recomputation with static algorithm, and the proposed GNN-SHS model is at least 31.8 times faster than the comparative method, demonstrating a considerable advantage in run time.

\chapter{Conclusion} 
\label{Chapter_conclusion}

In this thesis, we studied the problem of designing effective, efficient and scalable approaches for discovering bottleneck nodes and edges in the network. We aim to improve the overall network resilience by focusing on cyber defense and information diffusion application domains. Along this line, this thesis investigated two critical graph-combinatorial optimization problems. We first studied a cyber defense graph-combinatorial optimization problem, where we addressed the problem of hardening active directory systems by discovering bottleneck edges in the network. We then investigated the problem of identifying bottleneck structural hole spanner nodes, which are crucial for information diffusion in the network. We transformed the problems into graph-combinatorial optimization problems and designed machine learning approaches for discovering bottleneck nodes and edges essential for enhancing network resilience. This chapter first summarizes the key findings of this thesis and later suggests future research directions.

\section{Thesis Summary}
The key contributions of this thesis are summarized as follows.

\begin{enumerate}
    \item In \textcolor{blue}{Chapter \ref{Chapter_AD_NNDP}},  we focused on designing defensive policies to discover bottleneck edges that can be blocked to defend active directory graphs. We studied a Stackelberg game model between one attacker and one defender on an AD graph. The attacker aims to maximize their chances of reaching the domain admin, and the defender seeks to block a constant number of edges to minimize the attacker’s success rate. We first showed that the problem of computing an optimal attacking and defensive policy is \#P-hard; therefore, intractable to solve exactly. We proposed a kernelization procedure that converts the AD attack graph into a smaller condensed graph. We train a neural network to solve the attacker’s problem and design an evolutionary diversity optimization based policy to solve the defender’s problem of determining which edges to block. Once the neural network is trained, it acts as a fitness function for the defender’s policy. On the other hand, the defender’s policy generates a diverse set of blocking plans that are used to train the neural network. The diversity of training samples plays a crucial role in training the neural network and prevents the neural network from getting stuck in the local optimum. Overall, the attacker’s and defender’s policies assist each other in improving. Our experimental results on synthetic AD graphs demonstrate that the proposed approach generates effective defensive plans.
    
    \item In \textcolor{blue}{Chapter \ref{AD_RL}}, we proposed another effective and scalable edge-blocking policy for hardening large-scale active directory graphs. We studied a Stackelberg game model between one attacker and one defender on an AD graph in configurable environment settings. Each environment configuration denotes an edge-blocking plan. The defender aims to find the best environment configuration for defending the AD graphs. In contrast, the attacker plays against the environment configurations for devising an attacking policy to maximize their chances of successfully reaching the domain admin. We proposed a reinforcement learning based policy to solve the attacker’s problem and a critic network assisted evolutionary diversity optimization based policy to solve the defender’s problem. At regular intervals, the defender evaluates the environment configurations, replicating those that are good for the defender and discarding the bad ones. The attacker and defender play against each other parallelly. Our extensive empirical results show that the proposed defensive policy is scalable to large AD graphs, accurately approximates the attacker’s problem and generates effective defensive plans.
    
    \item The increasing size of networks poses a significant runtime challenge to the existing solutions for discovering SHS nodes in large-scale networks. Moreover, conventional approaches fail to discover SHS nodes across diverse networks. Therefore, in \textcolor{blue}{Chapter \ref{SHS_GNN}}, we proposed effective and efficient graph neural network models for discovering SHS nodes in large-scale and diverse networks. We first designed GraphSHS, a graph neural network model to discover SHS nodes in large-scale networks. GraphSHS utilizes the network structure and node features to learn the SHS nodes in the network. GraphSHS considers an inductive setting where the model is generalizable to new nodes of the same graph or new graphs from the same network domain. In addition, we proposed another graph neural network model, Meta-GraphSHS, to identify SHS nodes across diverse networks. Meta-GraphSHS model is based on the concept of Meta-Learning. The model learns generalizable parameters from diverse graphs in order to create a customized model that adapts its parameters according to the new unseen graphs. Theoretically, we proved that the depth of our proposed graph neural network models should be at least $\Omega(\sqrt{n}/\log n)$ to discover SHS nodes accurately. Our experimental results demonstrate that the proposed models are highly efficient and effective in discovering SHS nodes in large-scale and diverse networks.

    \item Real-world networks are highly dynamic in nature and change over time, due to which bottleneck SHS nodes also change. Currently, there is no solution that identifies SHS nodes in dynamic networks. Moreover, traditional SHS identification algorithms are considerably time-consuming and may not work efficiently for dynamic networks. Therefore, in \textcolor{blue}{Chapter \ref{SHS_LCN}}, we developed efficient approaches for discovering SHS nodes in dynamic networks. We first designed a decremental Tracking-SHS algorithm that dynamically updates Top-$k$ SHS nodes in the network. The algorithm reduces the number of re-computations by discovering affected nodes due to updates in the network and performing re-computations for the affected nodes only. In addition, we proposed GNN-SHS, a graph neural network model to identify SHS nodes in dynamic networks. GNN-SHS considers the dynamic network as a sequence of snapshots and discovers SHS nodes in these snapshots. The GNN-SHS model is able to discover SHS nodes for incremental and decremental edge updates of the network.
    In order to determine the efficiency of our proposed Tracking-SHS algorithm, we performed a theoretical analysis of the algorithm, and our results proved that the Tracking-SHS algorithm attains high speedup over static algorithms. In addition, we performed experiments on various synthetic and real-world datasets, and our results demonstrate that the proposed approaches achieve high speedup over re-computations.

\end{enumerate}

\section{Future Research Directions}
This thesis makes significant contributions towards enhancing network resilience by discovering key nodes and edges in the network. Moreover, this thesis also identified several possible research directions that can be explored in future. Some of the future research directions are discussed below:

\begin{itemize}

    \item \textbf{Determining hardness of problem - defending exact tree-like active directory graphs.} In this thesis, we proved that the problem of computing optimal defensive policy is \#P-hard. However, we can explore some special cases in order to determine the hardness of the problem, such as when the active directory graph is an exact tree. We can consider using Monte Carlo tree search algorithms \cite{coulom2007efficient} or graph neural networks \cite{scarselli2008graph} to design defensive policies for defending such AD graphs. 
    
    \item \textbf{Hardening dynamic active directory graphs.} In \textcolor{blue}{Chapter \ref{Chapter_AD_NNDP}} and \textcolor{blue}{Chapter \ref{AD_RL}}, we devised various defensive policies for defending active directories in static graphs. However, active directory graphs are highly dynamic in nature, where a large number of edges are added to the network whenever a user logs in to an account/computer. One future direction is to design defensive policies that can handle large-scale dynamic active directories. We can consider exploring continual learning \cite{parisi2019continual}, meta-learning \cite{vanschoren2018meta} or reinforcement learning \cite{kaelbling1996reinforcement} for defending dynamic active directory graphs.  
    
    \item \textbf{Discovering structural hole spanner nodes in dynamic networks by efficiently updating node embeddings.} We designed a graph neural network model in \textcolor{blue}{Chapter \ref{SHS_LCN}} to discover SHS nodes in dynamic networks by considering the network as a sequence of snapshots. Another efficient approach could be to compute the embedding vectors of the nodes and then update the node embeddings whenever there are changes in the network. Therefore, in our future work, we will consider designing graph neural network models that discover SHSs in dynamic networks by efficiently updating the node embeddings whenever there are updates in the network.
    
    \item \textbf{Incremental maintenance of structural hole spanner nodes in dynamic networks.}
    In \textcolor{blue}{Chapter \ref{SHS_LCN}}, we designed an algorithm for decremental maintenance of SHS nodes in dynamic networks and provided a theoretical analysis of the speedup achieved by the algorithm. Along this line, another research direction could be designing incremental algorithms for maintaining SHS nodes in dynamic networks. We can use various data structures, such as trees and disjoint-set \cite{galil1991data}, to store intermediate data of nodes and later use this data to find patterns between old and new scores of the affected nodes in order to efficiently update structural hole spanner nodes in dynamic networks.

\end{itemize}
\end{doublespacing}}






\printbibliography[heading=bibintoc]


\end{document}